\documentclass[acmsmall]{acmart}
\usepackage[dvipsnames,rgb]{xcolor}
\usepackage{amsfonts}

\usepackage{amsmath}
\usepackage{amsthm}
\usepackage{xspace}
\usepackage{paralist}
\usepackage{misc} 
\usepackage{tikz}
\usepackage{needspace}
\usepackage{mathtools}
\usepackage{soul} 
\usepackage{stmaryrd}
\usetikzlibrary{shapes,calc,positioning,arrows,hobby}
\usetikzlibrary{backgrounds, patterns.meta}
\makeatletter
\newtheorem*{rep@theorem}{\rep@title}
\newcommand{\newreptheorem}[2]{%
\newenvironment{rep#1}[1]{%
 \def\rep@title{#2 \ref{##1}}%
 \begin{rep@theorem}}%
 {\end{rep@theorem}}}
\makeatother

\usepackage{thm-restate}

\usepackage{ifpdf}

\ifpdf
 
\fi

\usepackage{tikz}
\usetikzlibrary{shapes,calc,positioning,arrows,automata}

\newenvironment{cclaim}[1]{%
  \par\medskip
  \noindent\textbf{Claim~#1.}%
}{%
  \\[1mm]
}

\newenvironment{cproof}[1]{%
 \def\claimnum{#1}
  \noindent\emph{Proof of claim.}
}{%
This concludes the proof of Claim~\claimnum.
}

\newcommand{\sizeof}[1]{{\left| #1 \right|}}

\usepackage{pgfplots}

\usetikzlibrary{calc}
\usepackage{amsmath}

\usetikzlibrary{decorations.pathmorphing, arrows.meta}

\tikzset{decoration={snake,amplitude=.4mm,segment length=2mm, post length=2mm, pre length=0mm},}
\definecolor{ao}{rgb}{0.0, 0.5, 0.0}

\newcommand{\dist}{\operatorname{dist}}
\renewcommand{\phi}{\varphi}
\renewcommand{\epsilon}{\varepsilon}

\newcommand{\class}[1]{\mathcal{#1}}

\newcommand{\schemaS}{\mathbf{S}}
\newcommand{\adom}{\textsf{adom}}
\newcommand{\var}{\textsf{var}}

\newcommand{\type}{\ensuremath{\mathsf{tp}}}
\newcommand{\Type}{\ensuremath{\mathsf{TP}}}

\title{Property Testing for Recursive Query Languages}

\author{Isolde Adler}
\email{isolde.adler@uni-bamberg.de}
\orcid{0000-0002-9667-9841}
\affiliation{%
  \institution{University of Bamberg}
  \city{Bamberg}
  \country{Germany}
}

\author{Carsten Lutz}
\email{clu@informatik.uni-leipzig.de}
\orcid{0000-0002-8791-6702}
\affiliation{%
  \institution{Leipzig University and ScaDS.AI Center Dresden/Leipzig}
  \city{Leipzig}
  \country{Germany}
}

\author{Quentin Manière}
\email{quentin.maniere@inria.fr}
\orcid{0000-0001-9618-8359}
\affiliation{%
  \institution{LIRMM, Inria, University of Montpellier, CNRS}
  \city{Montpellier}
  \country{France}
}

\author{Marcin Przybyłko}
\email{M.Przybylko@mimuw.edu.pl}
\orcid{0000-0003-1859-7055}
\affiliation{%
  \institution{University of Warsaw}
  \city{Warsaw}
  \country{Poland}
}
\author{Lukas Schulze}
\email{lschulze@informatik.uni-leipzig.de}
\orcid{0009-0001-3371-9121}
\affiliation{%
  \institution{Leipzig University}
  \city{Leipzig}
  \country{Germany}
}

\setcopyright{cc}
\setcctype{by}
\acmJournal{PACMMOD}
\acmYear{2026} \acmVolume{4} \acmNumber{5 (PODS)} \acmArticle{296}
\acmMonth{11} \acmDOI{10.1145/3850131}



\received{30 May 2026}
\received[revised]{25 August 2026}
\received[accepted]{1 September 2026}

\begin{abstract}
  In the context of database querying, property testing
  provides a framework for testing query answers with high confidence while inspecting only a sublinear part of the database, through completion queries and size queries.
   A fundamental result of Chen and Yoshida~\cite{chen2019testability} states
that  non-satisfaction of a Boolean conjunctive query~$q$ is
  testable with a constant number of such queries and one-sided error if and only if $q$ is equivalent to an $\alpha$-acyclic query. 
  In this article, we initiate the study of property testing for recursive query languages, focusing on two-way regular path queries (2RPQs) and 
   monadic Datalog. One of our main results is positive:
  non-answers to any 2RPQ are constant query testable with one-sided error. We extend this slightly to a certain class of monadic Datalog programs in which recursion is restricted to be linear and rule bodies must be $\alpha$-acyclic. Turning towards unrestricted monadic Datalog, we next
show that if a monadic Datalog program~$\Pi$ is not equivalent to an
$\alpha$-acyclic program, then falsity of $\Pi$ is
not constant query testable with one-sided error. This is under the  assumption 
that  all rule bodies are self-join free. We leave open the case of
monadic Datalog programs with $\alpha$-acyclic rule bodies that are not restricted to linear recursion, but observe as a first step that there exist $\alpha$-acyclic programs that are mildly non-linear and constant query testable with one-sided error.
\end{abstract}

\begin{document}

\newtheorem{problem}[theorem]{Problem}

\maketitle

\section{Introduction}

The unprecedented scale of modern data sets has made exhaustive computation increasingly impractical in many application domains. Large scientific experiments generate petabytes of observational data \cite{ivezic2014statistics}, web and social-network platforms maintain massive graph-structured data sets \cite{DBLP:books/cu/LeskovecRU14}, and modern database systems routinely manage data whose size renders even a single sequential scan a substantial computational task. These developments have motivated a broad line of research on algorithms that operate under restricted access to the input. Prominent examples include streaming algorithms \cite{DBLP:journals/fttcs/Muthukrishnan05}, sketching and synopsis techniques \cite{DBLP:journals/ftdb/CormodeGHJ12}, coresets \cite{DBLP:journals/ki/MunteanuS18}, and approximate query processing \cite{DBLP:conf/vldb/GarofalakisG01}, all of which seek to extract useful information from data while examining only a small fraction of it. Property testing represents an elegant incarnation of this general paradigm. Rather than deterministically computing an exact answer, a property tester is a probabilistic algorithm that aims to distinguish objects satisfying a given property from those that are far from satisfying it, using only a sublinear, and ideally constant, number of probes into the data. This perspective has proved fruitful in a variety of settings \cite{goldreich2017introduction} and raises a natural question in data management: 
which queries can be processed efficiently in the sense of property testing, and which cannot?

An investigation of this question has been started by Chen and Yoshida in \cite{chen2019testability}, focusing on Boolean conjunctive queries (CQs) and unions thereof (UCQs). More precisely, Chen and Yoshida study the question whether, for a fixed UCQ $q$,
falsity of $q$ is constant query testable with one-sided error, adopting a database version of the general graph model. Let us unpack this statement. For a UCQ $q$ and $\epsilon \in (0,1)$, an
\emph{$\epsilon$-tester for falsity of} $q$ is a probabilistic
  algorithm that has access to the input database $D$ only via certain queries, accepts  with probability $2/3$ if $D$ makes $q$ false, and rejects  with probability $2/3$ if
  $D$ is $\epsilon$-far from making $q$ false. The latter means that one has to remove more than $\varepsilon |D|$ facts from $D$ to make $q$ false, and the behavior of the tester is
unspecified if $D$ makes $q$ true, but  is not
$\varepsilon$-far from making it false. Falsity of
$q$ is \emph{constant query testable} if for every $\epsilon\in (0,1)$,
                there exists an $\epsilon$-tester for falsity of $q$  that
                makes at most constantly many queries to $D$, and it is constant query testable  \emph{with one-sided error} if, in addition, the tester 
always accepts if $D$ makes $q$ false.  
We have not yet explained the form of the queries through which the tester can access the database $D$, and this is where the general graph model comes into play. In Chen and Yoshida's presentation of this model, there are two types of queries. In a completion query, the tester provides a $k$-ary relation symbol $R$
  and a tuple $\bar a \in (\adom(D) \cup \{ \ast \})^k$, and receives
  as an answer a tuple $\bar a' \in \adom(D)^k$, sampled
  uniformly at random,  that satisfies  $R(\bar a') \in D$ and agrees with $\bar a$ on all positions not filled with `$*$'. We call such an $\bar a'$ an \emph{$R$-completion} of $\bar a$. In a size query, the tester provides the same input $R, \bar a$ as in a completion query, and receives as an answer the number of $R$-completions 
of $\bar a$.

We are now ready to state the main result of Chen and Yoshida: a Boolean UCQ $q$ is constant query testable with one-sided error if and only if, after deleting from $q$ all CQs that are strictly contained in another CQ in $q$, the homomorphism core of every CQ in $q$ is $\alpha$-acyclic. This result applies to bag databases, that is, to databases in which every fact may appear multiple times. Note that this distinction matters for $\varepsilon$-farness since multiplicities affect the number of fact deletions required to falsify the query. The same effect also shows up in some other dichotomy efforts in database theory, such as conjunctive query resilience \cite{DBLP:journals/pacmmod/MakhijaG23,10.1145/3806207}. 
While working with bag databases makes upper bounds more general, the lower bound constructions of Chen and Yoshida crucially depend on the availability of multiplicities.
Chen and Yoshida also discuss the practical implementation of completion and size queries,  treated as primitives in property testing, arguing that they can be implemented in polylogarithmic time if relations are represented as range trees~\cite{DBLP:journals/ipl/Bentley79}. 

The purpose of this paper is to initiate an investigation of property testing for recursive query languages, using Chen and Yoshida's framework outlined above. We focus on two-way regular path queries (2RPQs) and monadic Datalog queries. For the former, as customary, we consider 2RPQs that output a binary relation over the database domain, and an answer candidate is 
given as an additional input to the tester. Our main result is that for every 2RPQ $q$, non-answers to $q$ are constant query testable with one-sided error. This 
most positive result is obtained by reduction to 
testing unreachability in directed multi-graphs which is known to have the same property \cite{DBLP:journals/tcs/YoshidaK12}.

We then turn to monadic Datalog (MDLog) for which our results are much more partial, leading us to questions that appear to be rather challenging. Note that every UCQ is expressible as a monadic Datalog program, and thus this setting generalizes the one studied by Chen and Yoshida.
Starting with positive results, we first consider Boolean MDLog programs $\Pi$ on binary schemas in which every rule body satisfies the following conditions: (i) it contains exactly two variables, and (ii) no IDB atom contains the head variable. We call such programs \emph{strongly linear} because, informally, the
restrictions ensure that satisfaction of $\Pi$ is witnessed by a linear part of the database, without even single edges branching off, but potentially with multi-edges and reflexive loops. We prove by reduction to the 2RPQ case that all strongly linear programs
are constant query testable with one-sided error.

We next ask whether there are any MDLog programs that are constant query testable with one-sided error, and  both recursive and branching. We prove a positive
answer for a  class of programs that we call \emph{simple branching programs (SBPs)}. These allow linear recursion along a single binary relation and have two endpoints that take the form of star queries, hence are branching. Our approach is to first reduce out reflexive loops and multi-edges, and to then decompose SBPs into a collection of \emph{simple linear programs (SLPs)}. We then use an argument in the style of Menger's theorem \cite{Menger}:
$\varepsilon$-farness from falsifying an SLP corresponds to the 
absence of a small edge-cut, and a Menger-type result then guarantees a large number of pairwise edge-disjoint paths. It turns out that many of these must be of constant length, which enables a reduction to property testing $\alpha$-acyclic UCQs, as per \cite{chen2019testability}. 

We then prove the lower bound result that if a Boolean MDLog program $\Pi$ is not equivalent to an MDLog program in which all rule bodies are $\alpha$-acyclic, then falsity of $\Pi$ is not constant query testable with one-sided error. This is under the assumption that all rule bodies in $\Pi$ are self-join free, but admits relations of unrestricted arity.  It is the most technically demanding result in this article. The core technical ingredient is to identify, for every program $\Pi$
under consideration, a database $D$, an answer $\bar a \in \Pi(D)$, and a `hard substructure'
in $D$,  taking the form of a chordless cycle or a tetra (also known as a simplex),
that must be used in a non-trivialized way by every derivation of $\bar a$ from $\Pi$ in $D$. Achieving this turns out to be remarkably delicate, and it is in this step where we rely on self-join freeness. Once this crucial result is established, it enables us to  reduce from testing falsity of a CQ that takes the form of a chordless cycle or a tetra, both  hard by the results in~\cite{chen2019testability}. 

Our results on monadic Datalog leave open the status of a large class of programs, namely programs that are both recursive and branching,  beyond the rather modest SBPs. To obtain
first insight into this class, we consider two example MDLog programs that are inspired by
Steiner tree packing~\cite{sun2026steinertypepackingproblems}.
We then note that both Menger's theorem and the central lemma that we use for our positive result for SBPs can be viewed as removal lemmas~\cite{ConlonFox2013}, frequently used in the literature to obtain property testers with one-sided error~\cite{AlonShapira2008Monotone}. While lemmas typically address the removal of subgraphs of fixed size and structure, in our case we would need removal lemmas that remove branching subgraphs of varying size, due to  recursion. There appear to be only few such lemmas in the literature. In fact, we show that neither
of our two example programs enjoys a removal lemma (in one case subject to the assumption that P $\neq$ NP), but that their falsity is nevertheless constant query testable, based on different arguments. We take this as a mild
indication that constant query testability might even extend to all $\alpha$-acyclic MDLog programs.

An appendix with full proofs can be found in the long version of this paper \cite{adler2026prop}.

\paragraph{Related Work}
We have already discussed in quite some depth the work of Chen and Yoshida \cite{chen2019testability}. Adler and Harwath initiated property testing for databases of bounded-degree and proved that if such databases have bounded tree-width, then every CMSO-definable property is constant query testable \cite{AdlerH18}; see \cite{AdlerF23, AdlerKP24} for more recent developments. On the more practical side Ben-Moshe et al. investigated property testing within query processing, testing near-sortedness of relations to enable faster sort-based query evaluation \cite{BenMosheKFMFS11}. For sequential and tree-structured data, testing recursively defined properties is comparably well understood: regular word languages are constant query testable \cite{AlonKNS00} and a recent trichotomy settles the query complexity \cite{BathieFM25}. Magniez and de Rougemont extended this to regular tree languages under a different notion of $\varepsilon$-farness, motivated by approximate XML validation \cite{MagniezR07}; corresponding correctors for $\varepsilon$-close databases were studied in~\cite{BoobnaR04}.

\section{Preliminaries}

\paragraph{Databases.}  
A \emph{schema} $\schemaS$ is a non-empty finite set of 
relation symbols $R$, each associated with an arity $\mn{ar}(R) \in \mathbb{N}$. 
We say that $\schemaS$ is \emph{binary} if all relation symbols in it have arity~2. 
Fix a countably infinite set of \emph{constants} $\Cbf$.
An \emph{$\schemaS$-fact} takes the
form 
$R(\bar a)$ where 
$R \in \schemaS$ and $\bar a \in \Cbf^{\mn{ar}(R)}$.
An
\emph{$\schemaS$-instance} $I$ 
is a  multiset  of
$\schemaS$-facts, that is, $I$ assigns to every $\schemaS$-fact $\alpha$ a multiplicity $I(\alpha) \in \mathbb{N}$. We use $|I|$ to 
denote the number of facts in $I$, that is, $\sum_\alpha I(\alpha)$. An \emph{$\schemaS$-database} $D$ is a finite $\schemaS$-instance, that is, $D$ assigns non-zero multiplicity  only to finitely many facts.
As usual, we assume that multiplicities are represented in unary when a database is used as an input to a decision problem. We may write $\alpha \in I$ to mean that
$I(\alpha) > 0$. 
We use $\adom(I)$ to denote the \emph{active domain} of $I$, that is,
the set of constants that occur in $I$.   When the schema $\schemaS$ is clear from the
context, we will typically not mention it, speaking e.g.\ only of
a  database.  Let $I_1$ and $I_2$ be databases over the same schema~\Sbf. 
A \emph{homomorphism} from $I_1$
 to $I_2$ is a mapping $h:\mn{adom}(I_1) \rightarrow \adom(I_2)$ such that for
all facts $R(\bar a) \in I_1$, we have $R(h(\bar a))\in I_2$. Note that homomorphisms need not respect multiplicities. We may write $I_1 \rightarrow I_2$ to
denote the existence of a homomorphism from $I_1$ to $I_2$.

\paragraph{Two-Way Regular Path Queries}
A \emph{two-way regular path query (2RPQ)} $q$ over binary schema \Sbf is a non-deterministic finite automaton
(NFA) 
over the alphabet $\Sigma_{\schemaS}$ that consists of
the 
symbols $R$ and $R^-$ for every 
$R \in \schemaS$. We represent NFAs as a tuple 
$(Q,\Sigma_\Sbf,\Delta,s_0,F)$ 
where $Q$ is a finite set of
states, $\Sigma_{\schemaS}$ is the input alphabet, $s_0 \in Q$ is the initial state, $F \subseteq Q$ is the
set of accepting states, and
$\Delta \subseteq Q \times \Sigma_{\schemaS} \times Q$ is the
transition relation.  A \emph{path} in an $\schemaS$-instance~$I$ is a sequence of the form
$a_1P_1a_2P_2 \cdots P_{n-1}a_n$ where $a_1,\dots,a_n \in \adom(I)$ and
$P_1,\dots,P_{n-1} \in \Sigma_{\schemaS}$ such that for
$1 \leq i < n$, one of the following holds: (i)~$P_i \in \schemaS$ and
$P_i(a_i,a_{i+1}) \in I$ or~(ii)~$P_i=P^-$ and $P(a_{i+1},a_i) \in I$.
Note that we do not require $a_1,\dots,a_n$ to be distinct.
We call $a_1$ the \emph{source} of the path, $a_n$ the \emph{target},
and $P_1 \cdots P_{n-1}$ the \emph{label}. A pair $(a_1,a_2)
\in \mn{adom}(I)^{2}$ is an \emph{answer} to $q$ on $I$ if  there is a path in $I$ that has source $a_1$, target $a_2$, and
some label $w$ in the language $L(q)$ accepted by the NFA $q$. Note that, for 2RPQs and also for all other query languages studied in this article, 
the multiplicities of facts in $I$ do not play a role in the definition of answers.  We
use $q(I)$ to denote the set of all answers to $q$ on $I$.

\paragraph{Conjunctive Queries, UCQs}
A \emph{(Boolean) conjunctive query (CQ)} $q$ is a conjunction of relational atoms $R(\bar z)$ where $R \in \Sbf$ is of arity $|\bar z|$. We may write $R(\bar z) \in q$ if $R(\bar z)$ is a conjunct in $q$. 
 The set of all variables used in $q$ is denoted
 $\mn{var}(q)$.
A CQ $q$ is associated with a \emph{canonical database}
$D_q$ that is obtained from $q$ by 
viewing all relational atoms as facts and all variables as constants.
With a \emph{homomorphism} from
$q$ to an \Sbf-database~$D$, we mean a homomorphism from $D_q$ to $D$. 
 A  \emph{(Boolean) union of conjunctive queries (UCQ)} over  schema $\schemaS$ is a
disjunction $q$ of finitely many CQs over $\schemaS$. 
We write $D \models q$ and say that $q$ \emph{is true on $D$} if there is a homomorphism from some CQ in $q$ to $D$.

\paragraph{Datalog}
A \emph{Datalog rule}
$\rho$  takes the form
$R_0(\bar{y}_0) \leftarrow
R_1(\bar{y}_1)\land \cdots\land R_n(\bar{y}_n)$ where  $n> 0$ and in each relational atom $R_i(\bar y_i)$, the relation symbol $R_i$ is of arity $|\bar y_i|$. We refer
to $R_0(\bar{y}_0)$ as the \emph{head} of~$\rho$, and to
$R_1(\bar{y}_1) \wedge \cdots \wedge R_n(\bar{y}_n)$ as the
\emph{body}. Every variable that occurs in the head is required to
also occur in the body. Note that
the body of a Datalog rule is a conjunctive query. To emphasize that the variables in a rule body $q$ are $\bar z$, we may denote it by $q(\bar z)$.
%
A \emph{Datalog program} $\Pi$ is a finite set of Datalog rules
with a selected goal relation {\mn{goal}}, of any arity, that does not occur
in rule bodies. Rules that use the goal relation in the head are
\emph{goal rules}. The \emph{arity of\/ $\Pi$} is the arity of its 
\mn{goal} relation and $\Pi$ is \emph{Boolean} if it has
arity zero.  Relation symbols that occur in the head of at least one
rule of $\Pi$ are \emph{intensional (IDB) relations} and all
remaining relation symbols in $\Pi$ are \emph{extensional (EDB)
  relations}.  Note that, by definition, \mn{goal} is an IDB
relation. When all EDB relations in $\Pi$ are from schema $\Sbf$,
then we say that $\Pi$ is \emph{over (EDB) schema} $\Sbf$. We also speak of \emph{EDB atoms} and \emph{IDB atoms} in rule bodies, with the obvious meaning.
A Datalog program $\Pi$ is \emph{monadic} or an \emph{MDLog program}
if all IDB relations in $\Pi$ are monadic, with the possible exception of the goal relation. $\Pi$ is
\emph{linear} if every rule body contains at most one atom that uses an IDB relation, and it is 
\emph{self-join free} if in each rule body, every EDB relation occurs in at most one atom.
%
%

There are several equivalent ways to define the semantics of Datalog programs \cite{abiteboul1995foundations}, here we use derivations.
Let $\Pi$ be a Datalog program of arity $r$ over a schema \Sbf, $I$ an \Sbf-instance, and
$\bar a \in \mn{adom}(I)^r$. A \emph{derivation} of $\bar a$ from $\Pi$  in $I$
is a node-labeled directed finite tree $\Gamma=(V,E,\ell,\rho,h)$
where $\ell,\rho,h$ are node labeling functions. More precisely,
$\ell$ associates every  $v \in V$ with a fact $\ell(v)$ such that
the following conditions are satisfied:
\begin{enumerate}

\item if $v_0 \in V$ is the root, then $\ell(v_0)=\mn{goal}(\bar a)$;

\item if $v \in V$ is an inner node, then $\ell(v)$ takes the form
  $P(\bar a)$, $P$ an IDB relation and $\bar a \in \mn{adom}(I)^{\mn{ar}(P)}$;

\item if $v \in V$ is a leaf, then $\ell(v) \in I$.

\end{enumerate}
In addition, $\rho$ and $h$ associate every inner node $v$ with label
$\ell(v)=P(\bar a)$ with a rule $\rho(v) = P(\bar x) \leftarrow q$ from $\Pi$
and a homomorphism $h(v)$ from $q$ to the database
$\{ \ell(v') \mid v' \text{ successor of } v \}$ such that $h(\bar x)=\bar a$.
For readability, we may write $\ell_v$ in place of $\ell(v)$, and likewise
for $\rho_v$ and $h_v$. We write
$I \models \Pi(\bar a)$ if there is a derivation of $\bar a$ from $\Pi$ in $I$. For Boolean programs $\Pi$, for readability we speak of a derivation of $\Pi$ in $I$, rather than a derivation of $()$ from $\Pi$ in $I$.

Let $\Pi_1$ and $\Pi_2$ be Datalog programs of the same arity and over the same
schema \Sbf. 
We say that  $\Pi_1$ is \emph{contained} in $\Pi_2$, 
written $\Pi_1 \subseteq \Pi_2$, if $\Pi_1(D) \subseteq \Pi_2(D)$ for every
\Sbf-database $D$. Moreover, $\Pi_1$ and $\Pi_2$ are \emph{equivalent} if
$\Pi_1 \subseteq \Pi_2 \subseteq \Pi_1$.

\paragraph{Acyclicity} A \emph{join tree} for
a database $D$ is an undirected tree $T=(V,E)$ where $V$ is the set of
facts in $D$ and for each constant $a$ in $\mn{adom}(D)$, the set $\{
\alpha \in V \mid a \text{ occurs in } \alpha \}$ is a connected
subtree of~$T$.  
Then, $D$ is \emph{$\alpha$-acyclic} if it has a join
tree and a CQ $q$ is $\alpha$-acyclic if the canonical database $D_q$ is. We further say that an MDLog program $\Pi$ is
\emph{$\alpha$-acyclic} if  all rule bodies in $\Pi$ are $\alpha$-acyclic CQs. For non-monadic Datalog programs, 
there is no unique natural way to define $\alpha$-acyclicity and we refrain from doing so. In particular, one may require the rule bodies to be $\alpha$-acyclic, or the rule bodies extended with the head atoms. For MDLog programs, these two definitions clearly coincide.


\paragraph{Property Testing.}  
We use the general model of property testing for graphs and databases in the formulation of Chen and Yoshida
\cite{chen2019testability}. Let $\schemaS$ be a schema. A
\emph{pointed $\schemaS$-database} is a pair $(D,\bar a)$ with $D$ an
$\schemaS$-database and $\bar a$ a tuple over $\adom(D)$.  The
\emph{arity} of $(D,\bar a)$ is $|\bar a|$.  An
\emph{$\schemaS$-property} is a class $\class{P}$ of pointed
$\schemaS$-databases $(D,\bar a)$, all of the same arity, that is
closed under isomorphism (treating the elements of $\bar a$ as
constants and respecting multiplicities). The \emph{arity} of \Pmc is that of the pointed databases in it. For $\varepsilon \in [0,1]$, a pointed database $(D,\bar a)$ is \emph{$\varepsilon$-far
  from satisfying property $\class{P}$} if it is necessary to add or
remove more than $\varepsilon|D|$ facts from $D$ to obtain a database
in~$\class{P}$.  Since our databases are multisets, decreasing a multiplicity is one way to `remove a fact', and likewise for increasing a multiplicity and adding a fact. For the properties studied in this article, fact addition is
  actually not helpful and only removal needs to be considered. 

\smallskip In property testing, the access to the input database $D$ is only
possible via certain forms of queries. We follow Chen and Yoshida in
allowing the tester to make the following queries:
\begin{itemize}

\item \emph{Completion query}. The tester provides a  relation symbol $R$
  and a tuple $(a_1,\dots,a_{\mn{ar}(R)}) \in (\adom(D) \cup \{ \ast \})^{\mn{ar}(R)}$. It receives
  a random \emph{$R$-completion} of $\bar a$ in $D$, that
  is, a tuple is sampled
  uniformly at random
  from the set
  $$\{ (a'_1,\dots,a'_{\mn{ar}(R)},i) \mid  1 \leq i \leq D(R(a'_1,\dots,a'_{\mn{ar}(R)})) \text{ and } a'_j=a_j \text{ if } a_j \in
  \adom(D) \},
  $$
  the last component is deleted and the resulting tuple is  returned. If no $R$-completion exists, then
  this is signaled by a special symbol. 

\item \emph{Size query}. The tester provides a relation symbol $R$ and
  a tuple $\bar a \in (\adom(D) \cup \{ \ast \})^{\mn{ar}(R)}$. It receives the
  number of \emph{$R$-completions} of $\bar a$ in $D$, that is, the cardinality of the above set.
\end{itemize}
This model is related to the general graph model often used in
property testing for undirected graphs. In fact, the two models are
equivalent on undirected graphs, see \cite{chen2019testability} for
details. 
%
%
\begin{definition}[Property testing]
  \label{def:proptest}
  Let $\schemaS$ be a schema and $\class{P}$ an
  $\schemaS$-property. For $\epsilon \in (0,1)$, an
  \emph{$\epsilon$-tester} for $\class{P}$ is a probabilistic
  algorithm that takes as input a pointed database $(D,\bar a)$ of the
  same arity as~$\class{P}$, has access to
  $D$ only via completion and size queries, and that
	\begin{enumerate}
		\item accepts $(D,\bar a)$ with probability $2/3$ if
                  $(D,\bar a)\in \class{P}$ and
		\item rejects $(D,\bar a)$ with probability $2/3$ if
                  $(D,\bar a)$ is $\epsilon$-far from $\class{P}$.
                \end{enumerate}
                We say that property $\class{P}$ is \emph{constant
                  query testable} if for every $\epsilon\in (0,1)$,
                there exists an $\epsilon$-tester for $\class{P}$ that
                makes at most constantly many queries.  We say that
                $\class{P}$ is \emph{testable with one-sided error} if
                the $\epsilon$-tester always accepts $(D,\bar a)$ if
                $(D,\bar a)\in \class{P}$.
\end{definition}
Note that the
length of the tuple $\bar a$ given as part of the input is determined
by the $\schemaS$-property $\class{P}$ and thus of constant length. This tuple is
given as a direct input to the tester, that is, unlike $D$ it is not
subject to any access restrictions. Also note that the number of
queries that a constant query $\varepsilon$-tester may ask is independent of its
input $(D,\bar a)$, but may depend on $\varepsilon$. While Definition~\ref{def:proptest} is non-uniform
in~$\epsilon$, all the algorithms given in this paper are
uniform in~$\epsilon$: we provide a single
algorithm that takes $\epsilon$ as an additional input.

\smallskip

We are interested in properties defined by database queries.
In fact, every query $q$  gives rise to the property
$\mathcal{P}_q = \{ (D,\bar a) \mid \bar a \notin q(D) \}$ and thus to a property testing problem in the
sense of Definition~\ref{def:proptest}. Note the negation, that is,
$\mathcal{P}_q$ is the property of $\bar a$ \emph{not} being an answer
to $q$ on $D$. If $\Pmc_q$ is
constant query testable, for better readability we usually
say that \emph{non-answers to $q$ are constant query testable}.
Note that property testing $\mathcal{P}_q$ corresponds to a data complexity perspective where the
query is fixed and thus of constant size. For Boolean queries $q$,
this specializes to $\mathcal{P}_q = \{ D \mid D \not\models q \}$.

It might be relevant to remark that the complement of property $\mathcal{P}_q$ is not interesting for
property testing since we can always add to a given database with
$\bar a \notin q(D)$ a set of facts to achieve $\bar a \in q(D)$, 
the size of the set being independent of $D$. This means that $D$
cannot be $\epsilon$-far from satisfying $\bar a \in q(D)$ unless
$\varepsilon \in O(1/|D|)$. But then an $\varepsilon$-tester can
recover the entire database $D$ with $O(1/\varepsilon)$ queries and
thus property testing is pointless. 




%

\medskip
Several results in this article are established by means of reduction. We thus
make explicit the type of reduction that is relevant in the context of property
testing.
\begin{definition}[Property testing reduction]
\label{def:proptestreduction}
Let $\mathcal{P}_1$ be an $\Sbf_1$-property of arity $n_1$ and $\mathcal{P}_2$ an $\Sbf_2$-property of arity $n_2$. 
A \emph{property testing reduction} from $\mathcal{P}_1$ to $\mathcal{P}_2$ is a
pair $(f_{\text{DB}},f_\epsilon)$ with $f_{\text{DB}}$ a function that 
maps every pointed $\Sbf_1$-database of arity $n_1$ to a pointed $\Sbf_2$-database of arity $n_2$ and $f_\epsilon: (0,1) \rightarrow (0,1)$
such that the following properties
are satisfied for every pointed $\Sbf_1$-database $(D,\bar a)$ of arity $n_1$:
\begin{enumerate}
\item if $(D,\bar a) \in \mathcal{P}_1$, then   $f_{\text{DB}}(D,\bar a) \in \mathcal{P}_2$;
\item if $(D,\bar a)$ is $\epsilon$-far from satisfying $\mathcal{P}_1$ for some $\varepsilon \in (0,1)$, then $f_{\text{DB}}(D,\bar a)$ is $f_\epsilon(\epsilon)$-far from satisfying~$\mathcal{P}_2$;
\item the answer to any completion or size query to $D'$, with $f_{\text{DB}}(D, \bar a)=(D',\bar a'),$ can be simulated by constantly many
completion and size queries to $D$.
\end{enumerate}
\end{definition}

We make use of the following facts.
\begin{proposition}
\label{prop:reduction}
Let $\mathcal{P}_1$ be an $\Sbf_1$-property and $\mathcal{P}_2$ an $\Sbf_2$-property such that testing $\mathcal{P}_1$ reduces to testing $\mathcal{P}_2$. Then
\begin{enumerate}
\item if  $\mathcal{P}_2$ is constant query testable with one-sided error, then so is   $\mathcal{P}_1$;
\item if  $\mathcal{P}_1$ is not constant query testable with one-sided error, then neither is  $\mathcal{P}_2$.
\end{enumerate}
\end{proposition}
 For Point~(1), assume that $\mathcal{P}_2$ is constant query testable with
one-sided error and let $\varepsilon_1 \in (0,1)$. An
$\varepsilon_1$-tester for $\mathcal{P}_1$ is obtained by running an
$f_\varepsilon(\varepsilon_1)$-tester for $\mathcal{P}_2$ on
$f_{\text{DB}}(D,\bar a)$, without ever materializing it. By
Condition~(3) of Definition~\ref{def:proptestreduction}, each completion and size query to the database in $f_{\text{DB}}(D,\bar a)$ can be answered via constantly many such queries to $D$. Conditions~(1) and~(2) guarantee correctness. Point~(2) is the contrapositive of Point~(1). 


\section{Two-Way Regular Path Queries}
\label{sect:RPQs}

We show that for every 2RPQ $q$, non-answers to $q$ are constant query testable with one-sided error.  This is achieved by reduction to testing unreachability in multi-digraphs, a problem
that is known to be constant query testable with one-sided error \cite{DBLP:journals/tcs/YoshidaK12}.

Recall that a \emph{multi-digraph} takes the form $G=(V,E,L)$ with $V$
a set of vertices, $E$ a set of edges,
and $L:E \rightarrow V \times V$ a function
that assigns to each edge a source and a target vertex. A \emph{path} in $G$ is a
sequence $e_0,\dots,e_n$ of edges
 such that $t_{i} = s_{i+1}$ for  $0 \leq i < n$, assuming $L(e_i) = (s_i, t_i)$.
 We say that the path is \emph{from} $s_0$ \emph{to}
$t_n$ and that $t \in V$ is \emph{reachable} from $s \in V$ if there is a path from $s$ to $t$ in~$G$. This lifts to sets of target vertices in the expected way: $T \subseteq V$ is
reachable from $s \in V$ if some
$t \in T$ is reachable from $s$.
%
%
\emph{Multi-digraph unreachability} is the  problem to decide, given a
multi-digraph $G=(V,E,L)$, a source vertex $s \in V$, and a set of 
target vertices $T \subseteq V$, whether $T$ is
unreachable from $s$ in~$G$. This is of course the case if and
only if $T$ is unreachable from $s$ in $G$ viewed as a digraph by forgetting edge
multiplicities. The notion of $\epsilon$-farness from being unreachable, by contrast,
is sensitive to multiplicities. 

A multi-digraph $G$ can be viewed as a
bag $\schemaS_G$-database $D_G$ in an
obvious way when the schema $\schemaS_G$ consists of a single binary
relation $E$. For a fixed cardinality $k$ of the set $T$, it is 
possible to view multi-digraph unreachability as an $\schemaS_G$-property $ \mathcal{P}_{k}$
in the sense of Definition~\ref{def:proptest}.
However, there is no
reason to bound the cardinality of $T$ by a constant and we do not assume such a
restriction. The exact definition of property
testing multi-digraph unreachability parallels
Definition~\ref{def:proptest}, but now it is the multi-digraph $G$ that can be
accessed by completion and size queries, the `direct' inputs are $s$ and $T$, Point~1 requires acceptance with probability $2/3$ if $T$
is unreachable from $s$ in $G$, and Point~2 requires rejection with
probability $2/3$ if $T$ is $\varepsilon$-far from being unreachable
from $s$ in $G$, defined in the expected way. 

\begin{theorem}
  \label{thm:multidigraphproptest}
   Multi-digraph unreachability is constant query testable with one-sided error using only completion queries of the form $(a,\ast)$, and no size queries.
\end{theorem}
\begin{proof}
  This is essentially proved in \cite{DBLP:journals/tcs/YoshidaK12}. Yoshida and Kobayashi use a different testing model than the one considered in this paper, based on queries of the form `provide the outdegree of a given vertex', `return the $i$-th neighbor of a given vertex', and `return a vertex uniformly at random'. Their actual algorithm, however, which is Algorithm~2 in \cite{DBLP:journals/tcs/YoshidaK12}, implements a straightforward random walk that can be realized very naturally with completion queries of the form $(a,\ast)$. 
  A minor gap to bridge is that Yoshida and Kobayashi  consider only a single target vertex. We show that the multi-target case can be reduced to the single-target case in a rather
  direct way. 
  
  With every input $G, s, T$ to multi-digraph unreachability in our more general multi-target formulation, we associate an input $G',s,t$ to single-target multi-digraph unreachability
  by defining $G'$ to be the result of identifying in $G$ all vertices in $T$ into the single vertex $t$.   The identification  preserves edges, that is, the edges of $G'$ are exactly those of $G$, only their sources and targets may have changed. With every $\varepsilon \in (0,1)$, we
  associate the same $\varepsilon$. It is easy to verify that Conditions~(1) and~(2) of Definition~\ref{def:proptestreduction} are satisfied. We argue that so is Condition~(3).
  We may assume that the single-target tester stops as soon as it reaches $t$, and therefore never makes a completion query from $t$. Consider a completion query $(a,*)$ to $G'$. Since $a
  \neq t$, we may make the completion query $(a,*)$ to $G$. If it returns $(a,b)$ with $b \notin T$, we return $(a,b)$, and if it returns $(a,b)$ with $b \in T$, we return $t$. This clearly produces an outgoing edge from $a$ in $G'$, sampled uniformly at random.
\end{proof}
We conjecture that it is possible to further generalize Theorem~\ref{thm:multidigraphproptest} to multi-digraph unreachability with multiple source (and target) vertices. Since we do not need this generality, we refrain from working out  details.

\smallskip

We now give the  reduction 
from property testing non-answers to 2RPQs to
property testing multi-digraph unreachability.
Take
any 2RPQ $q = (Q,\Sigma_{\schemaS},\Delta, s_0,F)$. 
 It is standard to characterize answers to $q$ on an
 $\Sbf$-database $D$ in terms of the product digraph of $D$ and $q$, see for instance \cite{DBLP:journals/siamcomp/MendelzonW95}.
Define the
multi-digraph $G_{D \times q}=(V,E,L)$ as follows:
$$
\begin{array}{rcl}
  V &=&\adom(D) \times Q \\[1mm]
  E&=&\{\langle R(a_1,a_2),(s_1,S,s_2) ,i\rangle \in D \times \Delta \times \mathbb{N} \mid S \in \{R,R^-\} \text{ and } 1 \leq i \leq D(R(a_1,a_2))\} \\[2mm]
  \multicolumn{3}{l}{L(\langle R(a_1,a_2),(s_1,R,s_2),i \rangle) = (\langle a_1,s_1 \rangle,\langle a_2,s_2 \rangle) \ \ \text{ for all }\langle R(a_1,a_2),(s_1,R,s_2) \rangle
\in E} \\[1mm]
\multicolumn{3}{l}{L(\langle R(a_1,a_2),(s_1,R^-,s_2),i \rangle) = (\langle a_2,s_1 \rangle,\langle a_1,s_2 \rangle) \text{ for all } \langle R(a_1,a_2),(s_1,R^-,s_2) \rangle
\in E.}
\end{array}
$$
%
%
%
%
%
Our version of the product deviates from the standard one in being a multi-digraph, because our databases are bag databases.  We say
that edge $\langle R(a_1,a_2),(s_1,S,s_2) ,i\rangle \in E$ \emph{derives} from fact $ R(a_1,a_2)$.
Note that multiple edges may derive from the same fact $\alpha$, at most $|\Delta| \cdot D(\alpha)$ many. More importantly,
every edge derives from a unique fact. This is not the case in the usual
(non-multi) digraph version of $G_{D \times q}$ where different facts in $D$ such as $R_1(a_1,a_2), R_2(a_1,a_2)$ may give rise to the same edge
$(\langle a_1,s_1\rangle,\langle a_2,s_2 \rangle)$. 
In the appendix we prove the following.
%
   %
%
%
  %
%
%
\begin{restatable}{lemma}{lemefartranslates}
  \label{lem:efartranslates}
  Let $D$ be a database, $q=(Q,\Sigma_{\schemaS},\Delta,s_0,F)$ 
  a 2RPQ, both over the same
  schema~$\schemaS$, $a_1,a_2 \in \adom(D)$, and $\epsilon \in (0,1)$.
	Then the following holds: If $(D,a_1a_2)$ is $\epsilon$-far from
  $(a_1,a_2) \notin q(D)$ then $G_{D \times q}$ is
  $\frac{\epsilon}{|\Delta|}$-far from
  $T:= \{ a_2\} \times F$ being unreachable
  from $S:=\{\langle a_1,s_0 \rangle\}$.
\end{restatable}
%
Based on Lemma~\ref{lem:efartranslates}, we prove the main result of this section.

\begin{theorem}\label{theo:testingCRPQ}
  Let $q$ be a 2RPQ. Then non-answers to $q$ are constant query
  testable with one-sided error.
\end{theorem}
\begin{proof}
  The proof is by reduction to multi-digraph unreachability. Let $q=(Q,\Sigma_{\schemaS},\Delta, s_0,F)$ be a 2RPQ over some schema \Sbf.
  We associate with every input 
   $(D,a_1a_2)$  to $\Pmc_q$ the multi-digraph $G_{D \times q}=(V,E,L)$, the source
  vertex 
  $s = \langle a_1,s_0 \rangle$, and the set of target vertices $T = \{ a_2\} \times F$.
  Moreover, with any  $\epsilon \in (0,1)$ we associate
  $\varepsilon'=\frac{\epsilon}{|\Delta|}$. Condition~(1) of Definition~\ref{def:proptestreduction}
  is satisfied: using the construction of $G_{D \times q}$, it is easy to verify that 
  $(a_1,a_2) \notin q(D)$ implies that $T$ is unreachable from $s$ in  $G_{D \times q}$.
  Condition~(2) is satisfied by Lemma~\ref{lem:efartranslates}.
  By Theorem~\ref{thm:multidigraphproptest}, it thus remains to argue that completion queries to
  $G_{D \times q}$ of the form $(a,*)$ can be simulated by constantly many completion and size queries to $D$.

  Consider such a query. 
   To deal with the fact that the
  (unlabeled) multi-edges in $G_{D \times
    q}$ stem from different relation symbols
  $R$ in the database~$D$,
%
  we first ask two size queries to $D$ for every $R \in \schemaS$, one with
  tuple $(a,\ast)$ and one with tuple $(\ast,a)$. Let $n_R$ and $n_{R^-}$ be the numbers returned. Further let
  $m_R$ be the number of states in the automaton $q$ that can be
  reached from $s$ by making a transition for input symbol $R$, and let  $m_{R^-}$ be the number of states in $q$ that can be reached from $s$ by making a transition for input symbol $R^-$.  We
  define a probability distribution $p$ over the elements of $\Sigma_{\schemaS}$
  that reflects successor numbers of $\langle a,s\rangle$ in
  $G_{D \times q}$ by setting, for all $P \in \Sigma_{\schemaS}$:
  \begin{equation*}
  p(P)=\frac{n_P \cdot m_P}{\sum_{P' \in \Sigma_{\Sbf}} n_{P'} \cdot
    m_{P'}}.
  \tag{$*$}
  \end{equation*}
  We then draw an element $P \in \Sigma_{\schemaS}$
  according to this distribution and make a completion query to $D$
  for $P$ with input $(a,\ast)$, let the result be $a'$. We further
  choose uniformly at random a state $s'$ from $q$ that can be
  reached from $s$ by making a transition for input symbol $P$ and
  return  the pair
  $\langle a',s' \rangle$ as the result of the initial completion
  query.

  It can be verified that
  the above procedure yields a
  successor of $\langle a,s \rangle$ in $G_{D \times q}$
  chosen uniformly at random. Note that this depends on products being multi-digraphs. In the digraph version, a
  successor $\langle a',s' \rangle$ of $\langle a,s \rangle$
  may derive from multiple distinct facts  $P_1(a,a'),\dots,P_k(a,a')$. As a consequence, the sum in the
  denominator of ($*$) may then exceed the number of successors of~$\langle a,s \rangle$.
\end{proof}
%

\section{Monadic Datalog: Upper Bounds}
\label{sect:mdlogupper}

It is natural to ask whether the
favorable result for 2RPQs extends
to other recursive query languages. 
We consider monadic Datalog, first noting that the results of Chen and Yoshida 
give the following.
\begin{restatable}[\cite{chen2019testability}]{theorem}{thmCY}
\label{thm:CY}
    Let $\Pi$ be a  non-recursive monadic Datalog program. Then falsity of\/ $\Pi$ is constant query testable with one-sided error if and only if it is equivalent to an $\alpha$-acyclic non-recursive monadic Datalog program.
\end{restatable}
\begin{proof}
Chen and Yoshida prove that falsity of a reduced Boolean UCQ is
constant query testable with one-sided error if and only if the
homomorphism core of every CQ in it is $\alpha$-acyclic \cite{chen2019testability}. We recall that a UCQ is \emph{reduced} if there is no containment between any  two distinct CQs in it.
Every UCQ can be reduced by deleting redundant disjuncts,
and a CQ is equivalent to an $\alpha$-acyclic CQ if and only if its
core is $\alpha$-acyclic. As a consequence, falsity of a UCQ is constant query
testable with one-sided error if and only if it is equivalent to an
$\alpha$-acyclic UCQ.
 It remains to note the following: (1) any UCQ $q$ can be turned into an equivalent non-recursive Datalog program $\Pi$ by 
    translating every CQ in $q$ to a goal rule in $\Pi$; 
    since $\Pi$ only contains goal rules, it is monadic.
    Conversely, (2)~any non-recursive Datalog program $\Pi$ can be translated into an equivalent  non-recursive Datalog program $\Pi'$ that contains only goal rules by exhaustively replacing IDB relations in rule bodies by the
    bodies of rules that produce the IDB, in all possible ways. For monadic programs, this preserves $\alpha$-acyclicity. The program $\Pi'$ can then be turned into an equivalent UCQ.
\end{proof}
We ask whether the same is true for  monadic Datalog with recursion.
While we establish the lower bound part of the question for self-join free programs, a corresponding upper bound  remains wide open.
We do, however, observe that constant query testability extends beyond 2RPQs. 

\smallskip

For the sake of completeness, we also note that an analogous result cannot be expected for non-monadic Datalog. In fact, even a counterpart of Theorem~\ref{thm:CY} fails to hold.
\begin{example}
  Consider the non-recursive Datalog program $\Pi$ defined as
  \[
P(x,y,z)
  \leftarrow R_1(x,y) \wedge R_2(y,z) \qquad
\mn{goal}()
  \leftarrow P(x,y,z) \wedge R_3(z,x).
\]
Note that $\Pi$ is $\alpha$-acyclic, no matter whether one demands only the rule bodies to be $\alpha$-acyclic, or the rule bodies extended with the head atoms. However, $\Pi$ is clearly equivalent to the CQ
\[
q_{\triangle}
  = \exists x,y,z\,
    \bigl(
      R_1(x,y) \wedge
      R_2(y,z) \wedge
      R_3(z,x)
    \bigr),
\]
which by the results of \cite{chen2019testability} is not constant query testable with one-sided error.
\end{example}

\subsection{Strongly Linear Programs}
\label{subsect:linearacyclic}
We consider the class of Boolean linear acyclic MDLog programs
such that, in every rule body, (i)~there are exactly two variables
and (ii)~no IDB atom  uses the head variable. Here we assume that atoms of the form $\mn{true}(x)$ are permitted in rule bodies; note that such atoms can always be removed from MDLog programs, although at the expense of increasing the number of variables in the rule body. We call programs of this form \emph{strongly linear}. For simplicity, we restrict our attention to binary schemas. We conjecture, however, that the main result of this section also holds for unrestricted schemas.

Strongly linear programs are closely related to 2RPQs, but also differ in that they
 admit parallel edges and reflexive loops,  may be disconnected, and are Boolean.
We call them strongly linear because their satisfaction is  witnessed by a collection
of sequences of constants, without any branching whatsoever, not even single edges branching from the sequences. 
\begin{example}
    Consider the following strongly linear program $\Pi$. It is true on databases that  contain an $R,S$-parallel edge and an $R$-path between two $S$-selfloops of length at least two.
    \begin{align*}
        \Pi = \{&P_1(x) \leftarrow R(x,y) \land S(y,y), \quad 
        P_1(x) \leftarrow R(x,y) \land P_1(y), \quad
        P_2(x) \leftarrow S(x,x) \land R(x,y) \land P_1(y),\\
        &P_3(x) \leftarrow \mn{true}(x) \land P_2(y), \quad
        \mn{goal()} \leftarrow P_3(x) \land R(x,y) \land S(x,y)\}
    \end{align*}
\end{example}
Our aim in this section is to prove the following.
\begin{theorem}
\label{thm:linacycmdlog}
    Let $\Pi$ be a strongly linear MDLog program over a binary schema. Then falsity of\/ $\Pi$ is constant query testable with one-sided error. 
\end{theorem}
We establish Theorem~\ref{thm:linacycmdlog} by reduction to the 2RPQ case. In fact, we have defined strongly linear programs so as to enable such a reduction. We first address reflexive loops and multi-edges in rule bodies. More precisely, we show that if falsity of every strongly linear MDLog program that does not use these features is constant query testable with one-sided error, then the same is true for unrestricted strongly linear MDLog programs. To achieve this, we eliminate reflexive loops and multi-edges one-by-one, starting with reflexive loops.

Let $\Pi$ be a strongly linear MDLog program, $\alpha \leftarrow q \in \Pi$, and $x \in \mn{var}(q)$. Further let 
$\Pi'$ be obtained from $\Pi$ by manipulating the body $q$ in the following way, leaving the head $\alpha$ untouched:
\begin{itemize}
    \item remove all reflexive loops $R(x,x)$;
    \item consistently replace $x$ with a fresh variable $v$;
    \item add to $q$ the atoms $N_1(x,u),N_2(u,v)$ where $N_1,N_2$ are fresh relation symbols and $u$ is a fresh variable.
\end{itemize}
Note that while $\Pi'$ is no longer strongly linear, it can easily be rewritten into an equivalent strongly linear MDLog program by splitting up the manipulated rule, introducing fresh IDB relations. 
\begin{example}
    Consider the following three programs. $\Pi_1$ is the initial program, $\Pi_2$ the program after applying our self-loop removal and $\Pi_3$ is equivalent to $\Pi_2$ but rewritten to be strongly linear.
    \begin{align*}
        \Pi_1 &= \{\mn{goal}() \leftarrow R(x,x) \land A_1(y)\} \quad \quad 
        \Pi_2 = \{\mn{goal}() \leftarrow N_1(x,u) \land N_2(u,v) \land A_1(y)\}\\
        \Pi_3 &= \{P_1(x) \leftarrow \mn{true}(x) \land A_1(y), \quad
        P_2(x) \leftarrow N_2(x,y) \land P_1(y),\quad
        \mn{goal()} \leftarrow N_1(x,y) \land P_2(y)\}
    \end{align*}
\end{example}
In the appendix we prove the following.
\begin{restatable}{lemma}{lemreflexiveloopremoval}
	If falsity of\/ $\Pi'$ is constant query testable with one-sided error, then so is falsity of\/ $\Pi$.
\end{restatable}
%

We next eliminate parallel edges. Let $\Pi$ be a strongly linear MDLog program, $\alpha \leftarrow q \in \Pi$, and $x,y \in \mn{var}(q)$ involved in a parallel edge. Further let 
$\Pi'$ be obtained from $\Pi$ by removing from $q$ all atoms $R(x,y)$ and $R(y,x)$, and  adding to $q$ the atoms $N_1(x,z_1),N_2(z_1,z_2),N_3(z_2,y)$ where $N_1,N_2,N_3$ are fresh relation symbols and $z_1,z_2$ fresh variables.
Note that while $\Pi'$ is no longer strongly linear, it can easily be rewritten into an equivalent strongly linear MDLog program by splitting up the manipulated rule, introducing fresh IDB relations. 
Also note that this reduction does not introduce reflexive loops. 

%
\begin{restatable}{lemma}{lemparalleledgeremoval}
	If falsity of\/ $\Pi'$ is constant query testable with one-sided error, then so is falsity of\/ $\Pi$.
\end{restatable}
%
%
%
%
%
We can thus assume that rule bodies 
 in strongly linear  MDLog programs do neither contain reflexive loops nor multi-edges.
The following completes the proof of Theorem~\ref{thm:linacycmdlog}. 
\begin{restatable}{lemma}{fromMDLogtoRPQs}
  Let $\Pi$ be a strongly linear MDLog program in which  rule bodies do neither contain reflexive loops nor multi-edges. Then falsity of\/ $\Pi$ is constant query testable with one-sided error.
\end{restatable}
\begin{proof}
  The proof is by reduction to property testing of 2RPQs, exploiting Theorem~\ref{theo:testingCRPQ}. 
  Let $\Pi$ be a strongly linear MDLog program  of the form described in the lemma, over some schema $\schemaS$.
  As the reduction target, we define a 2RPQ $q$ over schema  $\schemaS' = \schemaS \uplus \{U\}$. In the reduction, we then associate every input $D$ to $\mathcal{P}_\Pi$ with an input $(D', uu)$ to $\mathcal{P}_q$.
 
 To facilitate the comprehension of the construction of $q$, we first
 make precise the databases $D'$ used in the reduction. Let $D$ be an
 \Sbf-database and set $\mn{deg}_D(a) = \sum_{R \in \schemaS} D(R(a,*)) + D(R(*,a))$ for all $a \in \mn{adom}(D)$.
    To define $D'$, take a fresh constant $u$ and set
    $$
        D'(R(a,b)) := D(R(a,b)) \text{ for all } R(a,b) \in D
        \qquad
        D'(U(u,a)) := \mn{deg}_D(a) \text{ for all } a \in \adom(D).
    $$
    It is easy to see that $|D'| = 3 |D|$.
    Note that $U$ connects $u$ to every element in the active domain of $D$, which we use to deal with the potential disconnectedness of $\Pi$. The multiplicity of the $U$-edges is chosen so as not to interfere with $\varepsilon$-farness, and to make it easy to answer size queries to $U$.
  
    The 2RPQ used as the reduction target has the form 
    $q = (Q, \Sigma_{\schemaS'}, \Delta, s_0, \{s_{\mn{acc}}\})$ with 
    \[
      Q := \{s_0, s_{\mn{init}}, s_{\mn{fin}}, s_{\mn{acc}}\}
      \;\cup\; \{s_P \mid P \text{ an IDB relation of } \Pi\}
      \;\cup\; \{s_\rho \mid \rho \text{ a disconnected rule of } \Pi\}.
    \]

    Since the reduction database is $(D',uu)$, we ask for an $L(q)$-labeled path from $u$ to $u$. 
    Intuitively, this path follows a linear derivation of $\Pi$ in the input database $D$. Being in state $s_P$ means that $P$ holds at the constant currently visited. 
    We let $q$ travel an $U$-edge at the beginning so that the
    actual derivation can start at any constant. Once that $\mn{goal}$
    has been derived, $q$ travels  $U^-$ to return to $u$. For a disconnected rule, $q$ travels $U^-$ followed by $U$ to reach 
    any constant. This is implemented via the states $s_\rho$.
    The transition relation is detailed in the appendix where we 
    also show that, if we let $\varepsilon' = \frac{\varepsilon}{3}$ for any $\varepsilon \in (0,1)$, then Conditions~(1) to~(3) of Definition~\ref{def:proptestreduction} are satisfied.
\end{proof}

\subsection{Simple Branching Programs}
\label{subsect:tinybranching}

We next ask whether there are queries that are both recursive and branching, and for which falsity is  constant query testable with one-sided error. We show 
that this is indeed the case for what we call simple branching programs, which  allow linear recursion along a single relation together with star queries at both endpoints. These stars are the only source of branching; in particular, branching is not nested inside the recursion. While this class departs only  mildly from strongly linear programs, proving that falsity  of simple branching programs is constant query testable does not appear to be possible by reduction to 2RPQs and thus requires different, more direct techniques.

Let \Sbf be a binary schema. Throughout this section, we admit relation symbols from 
$\Sigma_\schemaS$ in Datalog programs, with an atom $R^-(x,y)$  simply being an alternative presentation of $R(y,x)$.
A \emph{simple branching program (SBP)} over \Sbf
is a Boolean MDLog program that takes the following form:
\begin{equation*}
    \mn{goal}() \leftarrow P(x) \wedge  \bigwedge_{i=1}^k R_i(w_i,x)
\quad
P(x) \leftarrow S(x,y) \wedge P(y)
\quad
P(x) \leftarrow S(x,y) \wedge \bigwedge_{j=1}^\ell T_j(y,z_j)
\tag{$*$}
\end{equation*} 
where 
$S,R_1,\dots,R_k,T_1,\dots,T_\ell \in \Sigma_\Sbf$ all use distinct relation symbols. 
%
%
The goal of this section is to prove the following.
\begin{restatable}{theorem}{thmsimplebranching}
\label{thm:simple_branching}
  Let $\Pi$ be a simple branching program. Then falsity of\/ $\Pi$ is constant query
  testable with one-sided error.
\end{restatable}
We conjecture that Theorem~\ref{thm:simple_branching} can be slightly generalized, e.g.\
to SBPs that admit reflexive loops and multi-edges. We do not know, by contrast, how to deal with other seemingly harmless extensions. For instance, this is the case for replacing the $R_i$ and $T_j$ edges with paths of constant length and for dropping the distinctness of the relation symbols.

Note that, due to the symmetric definition
of SBPs, we may assume
w.l.o.g.\ that $S \in \Sbf$, that is, $S$ is not of the form $R^-$. We shall do so in what follows. 
A database $D'$ is a \emph{subdatabase} of an \Sbf-database $D$ if $D'(\alpha) \leq D(\alpha)$ for all \Sbf-facts $\alpha$. We call subdatabases $D_1,D_2$ of $D$ 
\emph{fact-disjoint} if for every fact $\alpha \in D$, $D_1(\alpha)+D_2(\alpha) \leq D(\alpha).$ A \emph{realization} of a Boolean MDLog program
$\Pi$ in a database $D$ is a minimal subdatabase
$D'$ of $D$ such that $D' \models \Pi$.
A \emph{simple linear program (SLP)} is an SBP, according to ($*$),  in which $k=\ell=1$.
%
We first show the following, using Menger's Theorem.
\begin{restatable}{lemma}{lemfarnessrealizations}
  \label{lem:farness_realizations}
  Let $\Pi$ be a simple linear program over a
  schema $\schemaS$, $D$ an $\schemaS$-database, and $\epsilon \in (0,1)$. 
  If $D$ is $\epsilon$-far from $D \not \models \Pi$, then there are more than $\frac{\epsilon}{6} \cdot |D|$ fact-disjoint realizations of\/ $\Pi$ in $D$.
\end{restatable}
%
In the product graph $G_{D\times q_\Pi}$, with $q_\Pi$ being $\Pi$ seen as a 2RPQ, let $V_1$ and $V_2$ be the sets of vertices whose automaton component is, respectively, the initial and the accepting state of $q_\Pi$. Thus, paths from $V_1$ to $V_2$ correspond to realizations of $\Pi$. The assumed $\varepsilon$-farness implies that every edge cut separating $V_1$ from $V_2$ has size greater than $\varepsilon|D|$, and Menger's theorem therefore gives a family of $m>\varepsilon|D|$ edge-disjoint paths from $V_1$ to $V_2$. These paths do not necessarily yield fact-disjoint realizations since an $S$-fact gives rise to two edges in the product graph: one corresponding to the first $S$-transition and one to the potential subsequent $S$-transitions. We show, however, that  at least $m/2$ paths share facts with at most two other paths, and greedily selecting  a path while discarding the paths that share a fact with it yields at least $m/6>(\varepsilon/6)|D|$ fact-disjoint realizations.
%
%

We now switch back to SBPs. Every SBP $\Pi$ gives rise to a 
collection of SLPs by
choosing one of the $R_i$
conjuncts in the first rule and one of the $T_j$ conjuncts in the last rule.
If a database $D$ is $\epsilon$-far from making $\Pi$ false,  we can use Lemma~\ref{lem:farness_realizations} to obtain a large number of fact-disjoint realizations of each of these SLPs and then combine them into fact-disjoint realizations of $\Pi$. This is made precise in the proof of the following result.

\begin{restatable}{lemma}{lemimplecrpqtorealizations}
  \label{lem:simplecrpq_to_realizations}
  Let $\Pi$ be a simple branching program over a schema \Sbf, $D$ an \Sbf-database, and $\epsilon \in (0,1)$. 
  If $D$ is $\epsilon$-far from making $\Pi$ false, then there are more than $\frac{\epsilon}{6} \cdot |D|$ fact-disjoint realizations of\/ $\Pi$ in $D$.
\end{restatable}
%
To prove Lemma~\ref{lem:simplecrpq_to_realizations}, we start with one of the SLPs $\Pi'$ that the SBP $\Pi$ gives rise
to, choosing $R_1$ and $T_1$. 
Lemma~\ref{lem:farness_realizations} gives more than $(\varepsilon/6)|D|$ realizations. 
We then show that there must be enough distinct remaining $R_i$- and $T_i$-edges, $i \neq 1$,  to complete these realizations into fact-disjoint realizations of $\Pi$. Otherwise, deleting the scarce edge type would destroy all witnesses more cheaply than the assumed $\varepsilon$-farness allows. 

We can now prove Theorem~\ref{thm:simple_branching}.
\thmsimplebranching*
\begin{proof}
  Let $\Pi$ be formed according to ($*$). For $m \geq 0$, let $q^m_\Pi$ be the CQ that contains the atoms
  $R_i(w_i,x_0)$ for $1 \leq i \leq k$, $S(x_0,x_1),\dots
  S(x_{m-1},x_m)$, and $T_j(x_m,z_j)$ for $1 \leq j \leq \ell$. That is, $q^m_\Pi$ is defined like $\Pi$, but fixes the $S$-path to be of length exactly $m$. Let $q$ be the UCQ that contains
  the CQs $q^m_\Pi$ for $1 \leq m \leq \frac{12}{\epsilon}$. We show the following:
\begin{enumerate}
    \item If $D \not \models \Pi$, then $D \not \models q$;
    \item if $D$ is $\epsilon$-far from $D \not \models \Pi$,
    then $D$ is $\frac{\epsilon}{12}$-far from $D \not \models q$;
\end{enumerate}
In fact, Point~1 is immediate by definition of $q$.
For Point~2, assume that $D$ is $\epsilon$-far from $D \not \models \Pi$. 
 Lemma~\ref{lem:simplecrpq_to_realizations} provides us with more than $\frac{\epsilon}{6} \cdot |D|$ many fact-disjoint realizations of $\Pi$ in $D$. 
This, in turn, means that $D$ must contain at least $\frac{\epsilon}{12} \cdot |D|$ many such realizations of size at most $\frac{12}{\epsilon}$ for the simple reason that, otherwise, there would have to be more than $|D|$ facts. Note that each of these realizations must make true one of the CQs in $q$.
Since the realizations are fact-disjoint, we thus need to remove at least $\frac{\epsilon}{12} \cdot |D|$
  many facts from $D$ to make $q$ false, proving Point~2.

Corollary~1.3 in \cite{chen2019testability} states that falsity of a UCQ $p$ 
is constant query testable with one-sided error 
if every CQ in $p$ has an $\alpha$-acyclic core.
It is straightforward to see that this is the case for $q$. We may thus obtain an $\epsilon$-tester
  $\mathfrak{A}_{\epsilon}$ for $D \not \models \Pi$ by choosing $\epsilon' = \frac{\epsilon}{12}$ and running a tester $\mathfrak{A}'_{\epsilon'}$ 
  for $q$ on $D$. 
Points~(1) and~(2) imply that $\mathfrak{A}'_{\epsilon'}$ indeed functions as
  an $\epsilon$-tester for $D \not \models \Pi$:  
  if $D \not \models \Pi$, then $D \not \models q$ by (1) 
  and thus $\mathfrak{A}'_{\epsilon'}$ accepts, and so does $\mathfrak{A}_{\epsilon}$; 
  if $D$ is $\epsilon$-far from $D \not \models \Pi$, then by (2)
  $D$ is $\frac{\epsilon}{12}$-far from $D \not \models q$ and thus $\mathfrak{A}'_{\epsilon'}$ rejects
  with probability at least $2/3$, and so does $\mathfrak{A}_{\epsilon}$.
\end{proof}

\subsection{Acyclic Non-Linear Programs Beyond Simple Branching}
\label{sec:musing}

We discuss some of the difficulties involved in finding constant query testers for $\alpha$-acyclic MDLog programs beyond SBPs. 
One natural way to approach this 
is to  establish some analogue of Lemma~\ref{lem:simplecrpq_to_realizations}. This lemma may be viewed as a `removal lemma', in the style of Ruzsa and Szemerédi's famous triangle removal lemma \cite{RuszaS}, see \cite{ConlonFox2013} for a survey
of related lemmas. Removal lemmas are in fact a standard mechanism for obtaining  property testers with one-sided error, see \cite{AlonShapira2008Monotone,goldreich2017introduction}. While most removal lemmas are concerned with removing subgraphs of fixed size and structure, both Lemma~\ref{lem:simplecrpq_to_realizations} and
Menger's theorem can be viewed as removal lemmas concerned with subgraphs of varying size. 

For going beyond SBPs, we are thus looking for  removal lemmas that remove subgraphs that are branching and of varying size. 
The closest result that we are aware of is related to Steiner tree packing.
Let $G=(V,E)$ be an undirected graph and $S \subseteq V$  a set of terminals. An \emph{$S$-tree in $G$} is an undirected tree $T=(V',E')$
that is a (not necessarily induced) subgraph of $G$ and satisfies $S \subseteq V'$.
\emph{Steiner tree packing (STP)} asks, given $G$ and $S$, to produce a  collection
of pairwise edge-disjoint $S$-trees in $G$ that is as large as possible.
Let $G=(V,E)$ be a connected undirected graph and $S \subseteq V$. A set of edges
$C \subseteq E$ is an \emph{$S$-Steiner-cut} if the graph $(V,E \setminus C)$ has at least two components that
contain vertices from $S$. 
The following is known.
\begin{theorem}[\cite{DeVosMcDonaldPivotto2016}]
\label{thm:devosetal}
Let $G=(V,E)$ be an undirected graph and $S \subseteq V$ with $|S| \geq 2$. If every
$S$-Steiner-cut in $G$ has size at least $5\ell+4$, then $G$ contains $\ell$
pairwise edge-disjoint $S$-trees.
\end{theorem}
Note that Theorem~\ref{thm:devosetal} is a
removal lemma for subgraphs of varying size, and it is branching  when $|S| > 2$. 
Kriesell's conjecture states that $5\ell+4$ can be improved to~$2\ell$. 
In our context, however, optimizing this function is not important: for the argument in the proof of Theorem~\ref{thm:simple_branching} that uses Lemma~\ref{lem:simplecrpq_to_realizations}, any function 
linear in $\ell$ suffices. 
However, we are not interested in undirected graphs, but in directed ones.

Let $G=(V,E)$ be a directed graph,  $S \subseteq V$, and $r \in S$. 
An \emph{$S,r$-tree in $G$} is a directed tree $T=(V',E')$ with root $r$
that is a subgraph of $G$ and satisfies $S \subseteq V'$.
\emph{Arc-Disjoint Directed Steiner tree packing (ADSTP)} asks, given $G$, $S$, and $r$, to produce a largest collection
of pairwise edge-disjoint $S,r$-trees in $G$ \cite{CheriyanSalavatipour2006}. 
\newcommand{\terminal}{P}
\newcommand{\pedge}{\mn{arc}}
We observe that every $S,r$-tree in a directed graph $G = (V, E)$ corresponds exactly to a realization of the following Boolean MDLog program
$\Pi^1_S$ in $D_G := \{ \pedge(u, v) \mid (u, v) \in E \}$, and vice-versa:
\begin{align*}
    \Pi_S^1 := ~  & \{ \mn{goal}() \leftarrow \bigwedge_{s \in S \setminus \{ r \}} \terminal_s(r) \} \cup \bigcup_{s \in S \setminus \{ r \}} \{ \terminal_s(s) \leftarrow \mn{true}, ~\terminal_s(x) \leftarrow \pedge(x, y) \land \terminal_s(y)  \}.
\end{align*}
For simplicity, here and in what follows we allow constants to appear in Datalog rules. In words, $\Pi^1_S$ asks whether all constants in $S$ are reachable from the constant $r$ in $D_G$.
It is known that a removal lemma for $S,r$-trees in the style of Theorem~\ref{thm:devosetal} does not exist unless P = NP, as a consequence of approximation hardness results \cite{CheriyanSalavatipour2006}.
On the other hand, falsity of $\Pi_S^1$ is constant query testable for every $S$: apply Theorem~\ref{theo:testingCRPQ} to test
whether $(r,s)$ is a non-answer to the RPQ
$\mn{arc}^*$ for all $s \in S \setminus \{r\}$,  
accept if any of the tests succeeds, and reject otherwise.

Next, let us vary the program $\Pi_S^1$ by no longer forcing the root of the tree to be $r$:
\begin{align*}
    \Pi_S^2 := ~  & \{ \mn{goal}() \leftarrow \bigwedge_{s \in S} \terminal_s(y) \} \cup \bigcup_{s \in S} \{ \terminal_s(s) \leftarrow \mn{true}, ~\terminal_s(x) \leftarrow \pedge(x, y) \land \terminal_s(y)  \}.
\end{align*}
We can prove, even without alluding to complexity-theoretic assumptions, that also the program $\Pi_S^2$ does not admit a removal lemma.
%
\begin{restatable}{lemma}{lemmaproblem}
    \label{lemma:no-bound-for-existential-center}
    For every $n > 0$, there exists a database $D$ such that
    \begin{enumerate}
        \item $D \setminus E \models \Pi_{\{ a, b, c\}}^2$ for every $E \subseteq D$ with $|E| < n$, and
        \item  $D_1 \cap D_2 \neq \emptyset$ for all realizations $D_1, D_2$ of\/ $\Pi_{\{ a, b, c\}}^2$ in $D$.
    \end{enumerate}
\end{restatable}
Unlike for $\Pi_S^1$, there is no obvious way to use  Theorem~\ref{theo:testingCRPQ} to property test
falsity of $\Pi_S^2$ with constantly many queries. In fact, 
$\Pi_S^1$ is equivalent to a strongly linear program
(with constants), but $\Pi_S^2$ is not. Nevertheless,
also falsity of $\Pi_S^2$ turns out to be constant query testable.
\begin{restatable}{lemma}{lemmaproblemtwo}
\label{lem:problemtwo}
Let $S$ be a finite set of constants.
Falsity of the program $\Pi^2_S$ is constant query testable with one-sided error.
\end{restatable}
The proof shows that, for property testing, it suffices to consider realizations in which every terminal is reached by a path of length $\big\lceil\frac{2|S|}{\varepsilon}\big\rceil$. If no such realization exists, then all existing realizations can be destroyed by removing at most $\varepsilon|D|$ facts. Unlike a removal-lemma argument, this does not extract many fact-disjoint realizations; instead, it reduces directly to testing a fixed $\alpha$-acyclic UCQ describing bounded-length realizations.
%
%
%
%
Thus
constant query testability extends beyond simple
branching programs, and one may even take this observation as a mild indication that an extension to all $\alpha$-acyclic MDLog programs is possible. That class, however, is of course  a lot richer than the 
programs $\Pi^1_S$ and $\Pi^2_S$. 
In fact, $\Pi^1_S$ and $\Pi^2_S$ do not even nest branching inside of recursion and are 
expressible as conjunctive regular path queries (CRPQs).

\section{Monadic Datalog: Lower Bounds}
\label{sect:MDLoglower}

We prove that the `only if' direction of Theorem~\ref{thm:CY} extends from non-recursive MDLog programs to recursive ones, if the  program is self-join free.
\begin{theorem}
	\label{theorem:non-testable}
	Let $\Pi$ be a  self-join free Boolean MDLog program.
	If $\Pi$ is not equivalent to an $\alpha$-acyclic MDLog program,
	 then falsity of\/ $\Pi$ is not constant query testable with one-sided error.
    %
\end{theorem}
%
%

To prove Theorem~\ref{theorem:non-testable}, we first need some
preliminaries.
A \emph{database graph} is a pair $(G, \mn{bag})$ with $G=(V,E)$ a directed graph and $\mn{bag}$ a function  that assigns to every vertex $v \in V$ a
database $\mn{bag}(v)$. We use $\mn{adom}(v)$ as a shorthand for
$\mn{adom}(\mn{bag}(v))$, take a \emph{neighbor} of a vertex $v \in V$ to be any vertex $u \in V$ with $(u,v) \in E$ or $(v,u) \in E$, and require that the following conditions are
satisfied:
\begin{enumerate}

\item if $w$ is a neighbor of $v$, then
  $|\mn{adom}(v) \cap \mn{adom}(w)| \mathrel{\leq} 1$;

\item
the set of nodes
$\{v \in V \mid  a \in \mn{adom}(v)\}$ is connected in~$G$, for each $a \in
\bigcup_{v \in V} \mn{adom}(v)$;

\item if $R(a,\dots,a) \in \mn{bag}(v)$ and $a \in \mn{adom}(w)$,
then $R(a,\dots,a) \in \mn{bag}(w)$.

\end{enumerate}
A database graph $(G, \mn{bag})$ defines the associated database
$D_{(G,\mn{bag})}$ which is the (non-disjoint) union of all databases
$\mn{bag}(v)$, $v \in V$. 
A database that derives in this way from a database graph is a \emph{bagged database}. We
may  use $\mn{adom}(G, \mn{bag})$ as shorthand for
$\mn{adom}(D_{(G,\mn{bag})})$.
If $G$ is a directed tree, we usually denote it by $T$ and call $(T, \mn{bag})$ a \emph{database tree} and $D_{(T,\mn{bag})}$ a \emph{tree-like database}. With the \emph{EDB-part} of the body of a Datalog rule, we mean the result of dropping from the body all IDB atoms.
The database $D'$ in the subsequent lemma can be `read off' from a derivation. Details are omitted, see for instance \cite{abiteboul1995foundations}. 
\begin{lemma}
  \label{lem:unrav}
  Let $\Pi$ be an MDLog program, $D$ a database, and 
  $\bar a \in \Pi(D)$. Then there is a tree-like
  database $D'$ based on a database tree $(T,\mn{bag})$ such that
  \begin{enumerate}
  \item $\bar a \in \Pi(D')$, and all
    constants from $\bar a$ occur in $\mn{adom}(v_0)$ with $v_0$ 
    the root of $T$;
  \item   $D'$ is a homomorphic pre-image of $D$;
  \item $\mn{bag}(v_0)$ is isomorphic to the EDB-part of the body of a goal rule in $\Pi$ and every other bag is isomorphic to the EDB-part of the body of a non-goal 
  rule.
  
  \end{enumerate}
\end{lemma}
Let $\Pi$ be an MDLog program over a schema $\schemaS$.
A \emph{contraction} of a (goal or non-goal) rule \mbox{$P(\bar x) \leftarrow q(\bar y) \in \Pi$} is any rule
\mbox{$P(\bar u) \leftarrow q'(\bar v)$} such that $q'(\bar v)$ can be obtained
from $q(\bar y)$ by identifying variables, and variables in the head are identified accordingly.
    We define the \emph{$\alpha$-acyclic approximation} $\Pi^*$ of $\Pi$ to be the program that contains all rules $\rho' : P(\bar u) \leftarrow q'(\bar z)$ that 
    can be obtained from the contraction of a  
     rule $\rho = P(\bar x) \leftarrow q(\bar y) \in \Pi$ by adding zero or more atoms to the body, possibly introducing fresh variables, such that (i)~$q'$ is $\alpha$-acyclic
     and (ii)~ $\sizeof{\bar z} \leq \sizeof{\bar y} + (N - 2) \cdot 2^{\sizeof{\bar y}}$ where $N$ is the maximum arity of relation symbols in $\schemaS$. 
     \begin{example}
     \label{example:acyclic-approximation}
         Consider the unary MDLog program $\Pi = \{ \mn{goal}(x) \leftarrow q(x) \}$ where $q(x) = R(x, x, y_1, y_2) \wedge S(y_2, y_3) \wedge T(y_3, x)$.
         Note that $\Pi$ is not equivalent to any $\alpha$-acyclic MDLog program and that $a \in \Pi(D)$ for $D = \{ R(a, a, a, b), S(b, c), T(c, a), R(a, b, c, d) \}$.
         Although $D$ is $\alpha$-acyclic, all $\alpha$-acyclic contractions $q'$ of $q$ satisfy $D \not\models q'(a)$.
         Nevertheless, $a \in \Pi^*(D)$ as $\Pi^*$ contains the rule $\mn{goal}(x) \leftarrow R(x, x, x, y_2) \wedge S(y_2, y_3) \wedge T(y_3, x) \wedge R(x, y_2, y_3, y_4)$.
         See Figure~\ref{fig:acyclic-approximation} for illustrations.
    \end{example}
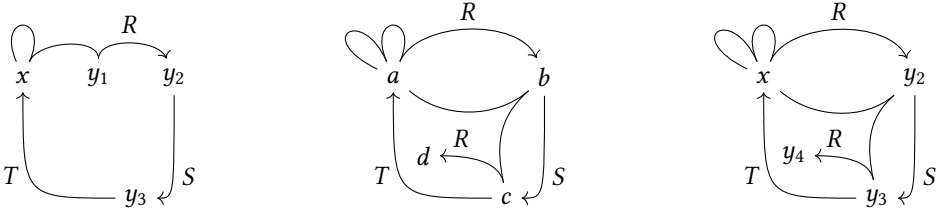
\begin{figure}
    \centering
    \begin{tikzpicture}[yscale=.8]
    \tikzset{every loop/.style={}}
    \begin{scope}[shift={(-4.9, 0)}]
        \node (x) at (0, 0) {$x$};
        \node (y1) at (1, 0) {$y_1$};
        \node (y2) at (2, 0) {$y_2$};
        \node (y3) at (1.5, -2) {$y_3$};
        \path
        (x) edge [out=110, in=70, looseness=15] (x)
        (x) edge [out=70, in=90] (y1)
        (y1) edge [->, out=90, in=110] node [above] {$R$} (y2)
        ;
        \draw [->] (y3) .. controls (0, -2) .. node [left, pos=.6] {$T$} (x);
        \draw [->] (y2) .. controls (2, -2) .. node [right, pos=.4] {$S$} (y3);
    \end{scope}
    \begin{scope}[shift={(4.9, 0)}]
        \node (a) at (0, 0) {$x$};
        \node (b) at (2, 0) {$y_2$};
        \node (c) at (1.5, -2) {$y_3$};
        \node (d) at (.4, -1.3) {$y_4$};
        \path
        (a) edge [out=150, in=110, looseness=15] (a)
        (a) edge [out=110, in=70, looseness=15]  (a)
        (a) edge [->, out=70, in=110] node [above] {$R$} (b)
        (a) edge [out=-45, in=-135] (b)
        (b) edge [out=-135, in=100] (c)
        (c) edge [->, out=100, in=0] node [near end, above] {$R$} (d)
        ;
        \draw [->] (c) .. controls (0, -2) .. node [left, pos=.6] {$T$} (a);
        \draw [->] (b) .. controls (2, -2) .. node [right, pos=.4] {$S$} (c);
    \end{scope}
    \begin{scope}[shift={(0, 0)}]
        \node (a) at (0, 0) {$a$};
        \node (b) at (2, 0) {$b$};
        \node (c) at (1.5, -2) {$c$};
        \node (d) at (.4, -1.3) {$d$};
        \path
        (a) edge [out=150, in=110, looseness=15] (a)
        (a) edge [out=110, in=70, looseness=15]  (a)
        (a) edge [->, out=70, in=110] node [above] {$R$} (b)
        (a) edge [out=-45, in=-135] (b)
        (b) edge [out=-135, in=100] (c)
        (c) edge [->, out=100, in=0] node [near end, above] {$R$} (d)
        ;
        \draw [->] (c) .. controls (0, -2) .. node [left, pos=.6] {$T$} (a);
        \draw [->] (b) .. controls (2, -2) .. node [right, pos=.4] {$S$} (c);
	\end{scope}
    \end{tikzpicture}
    \caption{Rule body $q$ (left) and $\alpha$-acyclic database $D$ (center) from Example~\ref{example:acyclic-approximation}. The $\alpha$-acyclic approximation $\Pi^*$ contains a goal rule whose body is depicted on the right.}
    \label{fig:acyclic-approximation}
\end{figure}
\Needspace{4\baselineskip}
\begin{restatable}{theorem}{theoremTHEcandidate}
	\label{theorem:THE-candidate}
	Let $\Pi$ be an MDLog program over a schema $\schemaS$.
    Then
	\begin{enumerate}
		\item for every $\schemaS$-database $D$, $\Pi^*(D) \subseteq \Pi(D)$;
		\item for every $\alpha$-acyclic $\schemaS$-database $D$, $\Pi(D) \subseteq \Pi^*(D)$;
		\item $\Pi$ is equivalent to an $\alpha$-acyclic MDLog program if and only if $\Pi$ is equivalent to $\Pi^*$.
	\end{enumerate}
\end{restatable}
The proof of Point~1 is straightforward. 
 For Point~2, we show how to use a derivation $\Gamma$ of $a$
 from $\Pi$ in $D$ to construct 
 a derivation $\Gamma$ of $a$
 from $\Pi^*$ in $D$.
 Let $\alpha \leftarrow q \in \Pi$ be a rule applied as part of $\Gamma$, via homomorphism $h$. We identify a corresponding $\alpha$-acyclic rule
 in $\Pi^*$ by starting with the contraction of $\alpha \leftarrow q$ 
 induced by $h$, and then adding atoms based on the $h$-image of $q$ in the $\alpha$-acyclic database $D$.
 Point~3 follows from Points~1, 2, and  Lemma~\ref{lem:unrav}.
%

\smallskip
Our next  aim is to prove that if an MDLog program $\Pi$ is not equivalent to its acyclic approximation~$\Pi^*$, then we can identify a database $D$ and  a `hard pattern' in $D$, taking the form of  a chordless cycle or a tetra, that \emph{must} be used in a non-trivialized way in every derivation of $\Pi$ in $D$. A non-trivialized way to use a cycle, for instance, is to bijectively map a cycle onto it, while wrapping a path around a cycle is a trivialized way to use it.

\medskip

Constants $a_1,a_2$ in a database $D$ are \emph{neighbors} if $D$ contains a fact in which both $a_1$ and $a_2$ appear.
A \emph{path} in  $D$ is a sequence $p=a_0, \dots, a_n$ of constants from $\mn{adom}(D)$ such that $a_i,a_{i+1}$ are
neighbors in $D$, for $0 \leq i <n$.
The \emph{length} of $p$ is $n$. 
A path $C=a_0, \dots, a_n$ in a database $D$ is a \emph{cycle} if $a_0=a_n$,
$n \geq 3$, and $a_i \neq a_j$
for $0 \leq i < j < n$. 
With a \emph{chordless} cycle,
we mean a cycle $C=a_0, \dots, a_n$ of length at least 4 such that for any $a_i,a_j$ with $0 \leq i < j < n$, $a_i,a_j$ being neighbors implies $j=i+1$ or $(i, j) = (0, n-1)$.
A \emph{$k$-tetra} in $D$ 
is a non-empty subset $S$ of $\mn{adom}(D)$ of
size $k$ such that for every subset $S' \subseteq S$ of size $k-1$, there
is a fact in $D$ that contains all constants in $S'$, while no
fact in $D$ contains all constants in $S$.
When simply saying \emph{tetra}, we mean a $k$-tetra with $k \geq 3$. We remark that, intuitively, there are of course also chordless cycles of length~3. These, however, require the additional condition that there is no fact that contains all constants on the cycle. In other words, they are a 3-tetra.
When referring to chordless cycles and tetras in a CQ $q$, we mean chordless cycles and tetras in $D_q$.
The following is well-known, see for instance \cite{berkholz-tutorial}.
\begin{lemma}
\label{lem:acyclictohardpattern}
    Let $D$ be a database that is not $\alpha$-acyclic. Then $D$ contains a 
    chordless cycle or a tetra.
\end{lemma}
%
%
%

%
%
We now work towards a definition of using a hard pattern in a non-trivialized way.
\begin{definition}[Induced CQ and Homomorphisms]
\label{def:inducedCQ}
    Let $\alpha \leftarrow q$ be an MDLog rule,
	$D$ a {bagged} database with database graph $((V,E),\mn{bag})$, $v \in V$, and $h$ a homomorphism from $q$ to $D$.
	 The \emph{CQ induced by $h$ and~$v$}, denoted $q_{h, v}(\bar y_{h, v})$, is obtained from $q$ by:
	\begin{itemize}
		\item[(i)] identifying any two variables $y, y'$ such that there is a path $y, y_1, \dots, y_n, y'$ in $q$, $n \geq 1$, with (a)~$h(y_1), \dots, h(y_n) \notin \mn{bag}(v)$ and {(b)~$h(y_1),h(y_n) \in \mn{bag}(u)$ for some $u \in V$};
        \item[(ii)] identifying any two  variables $y, y'$ with $h(y) = h(y')$ that occur in $\alpha$ (for goal rules);
        \item[(iii)] dropping all atoms $R(\bar u)$ such that $R(h(\bar u)) \notin \mn{bag}(v)$.
	\end{itemize}
We use $h_v$ to denote
    the homomorphism from $q_{h, v}$ to $\mn{bag}(v)$ obtained by restricting $h$ to 
    $\mn{var}(q_{h,v})$.
\end{definition}
%
It can be verified that the mapping $h_v$ from Definition~\ref{def:inducedCQ} is indeed a homomorphism from $q_{h, v}$ to $\mn{bag}(v)$. This crucially relies on Condition~1 in the definition of database graphs and on Condition~(\textit{i.b}) in the above definition.
Importantly, if $D$ is tree-like and $((V,E),\mn{bag})$ is a database tree, then Condition~(\textit{i.b}) is vacuously true.
\begin{example}
    Consider Figure~\ref{fig:induced-and-using}.
    Variables $y_3$ and $y_6$ have been identified in $q_{h, B_1}$  into the variable $y_{3, 6}$, due to the path $y_3, y_4, y_5, y_6$.
    Atoms $R_3(y_3, y_4)$, $R_4(y_4, y_5)$ and $R_5(y_5, y_6)$ have been dropped as they are mapped outside of $B_1$.
\end{example}
\begin{figure}[t]
    \centering
    \begin{tikzpicture}[scale=1.1]
    \begin{scope}[shift={(-4.9, 0)}]
        \node (y0) at (0, 0) {$y_1$};
        \node (y1) at (1, 0) {$y_2$};
        \node (y2) at (2, 0) {$y_3$};
        \node (y3) at (3, 0) {$y_4$};
        \node (y4) at (3, -1) {$y_5$};
        \node (y5) at (2, -1) {$y_6$};
        \node (y6) at (1, -1) {$y_7$};
        \node (y7) at (0, -1) {$y_8$};
        \path
        (y0) edge [->] node [above, sloped] {$R_1$} (y1)
        (y1) edge [->] node [above, sloped] {$R_2$} (y2)
        (y2) edge [->] node [above, sloped] {$R_3$} (y3)
        (y3) edge [->] node [below, sloped] {\rotatebox{180}{$R_4$}} (y4)
        (y4) edge [->] node [below, sloped] {$R_5$} (y5)
        (y5) edge [->] node [below, sloped] {$R_6$} (y6)
        (y6) edge [->] node [below, sloped] {$R_7$} (y7)
        (y7) edge [->] node [above, sloped] {$R_8$} (y0)
        ;
    \end{scope}
    \begin{scope}[shift={(4.9, 0)}]
        \node (a) at (0, 0) {$y_1$};
        \node (y2) at (1, 0) {$y_{2}$};
        \node (c) at (2, -.5) {$y_{3, 6}$};
        \node (y7) at (1, -1) {$y_{7}$};
        \node (f) at (0, -1) {$y_8$};
        \path
        (a) edge [->] node [above, sloped] {$R_1$} (y2)
        (y2) edge [->] node [above, sloped] {$R_2$} (c)
        (c) edge [->] node [below, sloped] {$R_6$} (y7)
        (y7) edge [->] node [below, sloped] {$R_7$} (f)
        (f) edge [->] node [above, sloped] {$R_8$} (a)
        ;
    \end{scope}
    \begin{scope}[shift={(0, 0)}]
		\node (bag2) at (2.7, -.55) [rectangle, rounded corners, draw=ForestGreen!70, fill=ForestGreen!40, fill opacity=.3, minimum height=1.8cm, minimum width=1.8cm] {};
		\node (bag2cap) at (3.2, -1.57) [ForestGreen] {$B_2$};
		\node (bag1) at (.8, -.45) [rectangle, rounded corners, draw=Mahogany!70, fill=Mahogany, fill opacity=.2, minimum height=1.8cm, minimum width=2.9cm] {};
		\node (bag1cap) at (-.2, -1.47) [Mahogany] {$B_1$};
		\node (a) at (0, 0) {$a$};
		\node (b) at (1, -.5) {$b$};
		\node (c) at (2, -.5) {$c$};
		\node (d) at (3, 0) {$d$};
		\node (e) at (3, -1) {$e$};
		\node (f) at (0, -1) {$f$};
		\path
		(a) edge [->] node [above, sloped] {$R_1$} (b)
		(b) edge [->, bend left] node [above, sloped] {$R_2$} (c)
		(c) edge [->] node [above, sloped] {$R_3$} (d)
		(d) edge [->] node [below, sloped] {\rotatebox{180}{$R_4$}} (e)
		(e) edge [->] node [below, sloped] {$R_5$} (c)
		(c) edge [->, bend left] node [below, sloped] {$R_6$} (b)
		(b) edge [->] node [below, sloped] {$R_7$} (f)
		(f) edge [->] node [above, sloped] {$R_8$} (a)
		;
	\end{scope}
\end{tikzpicture}
    \caption{
    CQ $q$ (left) for some MDLog rule $\alpha(y_1) \leftarrow q$, bagged database $D$ (center) with bags $B_1$ and $B_2$ whose domains intersect at $c$, and induced CQ $q_{h, B_1}$ (right) where $h$ is the unique homomorphism of $q$ to $D$.
    }
    \label{fig:induced-and-using}
\end{figure}
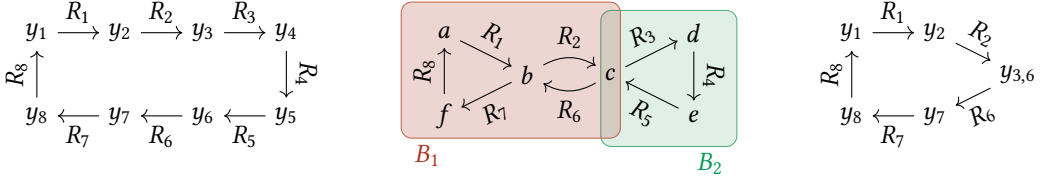
Let $D$ be a database. A \emph{bijection} between cycles $C=a_0,\dots,a_n$ and $C'=b_0,\dots,b_n$
in $D$ is a function $f: \{a_0,\dots,a_{n-1}\} \rightarrow \{ b_0,\dots,b_{n-1}\}$ such that
for every $i < n$, there is
a $j <n$ with $\{f(a_i),f(a_{i+1 \!\! \mod n})\} = \{b_j,b_{j+1 \!\! \mod n}\}$; we remind the reader that $a_n=a_0$ and $b_n=b_0$.
Note that since tetras are sets,
it is clear what a bijection between tetras is. The following definition formalizes what it means to use a chordless cycle or tetra in a non-trivialized way.
\begin{definition}[Using a Cycle/Tetra]
	Let $\alpha \leftarrow q$ be an MDLog rule, $D$ a {bagged} database with database graph $((V,E),\mn{bag})$,  and $h$ a homomorphism from $q$ to $D$.
	We say that $h$ \emph{uses} a chordless cycle (resp.\ a tetra) $X$ in $D$ if there exists a $v \in V$
    and a chordless cycle (resp.\ a tetra) $X'$ in $q_{h, v}(\bar y_{h, v})$ such that $h_v$ restricted to the variables that occur in $X'$  is a bijection between $X$ and $X'$.
     For a $k$-ary MDLog program $\Pi$ and a tuple $\bar a \in \mn{adom}(D)^k$, 
     chordless cycle  (resp.\  tetra) $X$
  is \emph{used} by a derivation $\Gamma=(V',E',\ell,\rho,h)$ of $\bar a$ from $\Pi$ in $D$ if for some $v \in V'$, $X$ is used by $h_v$.
\end{definition}

\begin{example}
    Once more consider Figure~\ref{fig:induced-and-using}.
    Note that $h$ does not use the 3-tetra $\{ a, b, f \}$ in $D$, as it fails to define a bijection from a $3$-tetra in $q_{h, B_1}$ (there is none!) to $\{ a, b, f \}$.
    In contrast, $h$ uses the 3-tetra $\{ c, d, e \}$ because $q_{h, B_2} = R_3(y_{3, 6}, y_4) \wedge R_4(y_4,y_5) \wedge R_5(y_5,y_{3, 6})$ contains the 3-tetra $\{ y_{3, 6}, y_4, y_5 \}$ and $h$ defines a bijection from this tetra  to  $\{ c, d, e \}$.
\end{example}

\begin{definition}[Crucial Cycle/Tetra]
Let $\Pi$ be an MDLog program, $D$ a {bagged} database and 
$\bar a \in \Pi(D)$. A set $S$ that
consists of 
chordless cycles and tetras in $D$ is  \emph{crucial} for $\bar a \in \Pi(D)$ if every derivation $\Gamma$ of $\bar a$ from $\Pi$ in $D$ uses some cycle or tetra in $S$. If~$\Pi$ is Boolean and $\bar a$ is the empty tuple, we instead say that $S$ is  \emph{crucial} for $D \models \Pi$.
\end{definition}
The following is the main technical result of this section. We believe that it may find applications also beyond property testing.
\begin{restatable}{theorem}{theoremOneCycleIsCrucial}
\label{theorem:one-cycle-is-crucial}
    If a self-join free MDLog program $\Pi$ is \emph{not} equivalent to an acyclic  MDLog program,
	then there are a database $D$, a tuple $\bar a \in \Pi(D)$, and a chordless cycle or tetra $X$ in $D$ such that $\{ X \}$ is crucial for $\bar a \in \Pi(D)$.
\end{restatable}

We sketch the proof of Theorem~\ref{theorem:one-cycle-is-crucial}.
Let $\Pi$ be an MDLog program  that is not equivalent to an acyclic MDLog program.
By Points~1 and~3 of Theorem~\ref{theorem:THE-candidate}, there is a database $D$ such that $\Pi^*(D) \subsetneq \Pi(D)$. Choose a  $\bar a \in \Pi(D) \setminus \Pi^*(D)$.    
From Lemma~\ref{lem:unrav}, we obtain a tree-like database $D'$ with the same properties; notably, $\bar a \notin \Pi^*(D')$ due to  Point~2.
$D'$ cannot be $\alpha$-acyclic as otherwise, by Point~2 of Theorem~\ref{theorem:THE-candidate}, we would have $\bar a \in \Pi^*(D')$. By Lemma~\ref{lem:acyclictohardpattern}, $D'$ therefore contains a chordless cycle or a tetra (in the following uniformly referred to as a pattern). 

The proof of Theorem~\ref{theorem:one-cycle-is-crucial} then proceeds in three  steps.
The first step starts with $D_0 := D'$ and constructs a sequence of potentially infinite tree-like databases $D_0, D_1, \dots, D_M$. 
Each $D_i$ is a homomorphic preimage of $D_{i-1}$, and the final database $D_M$ contains a set $S$ of either chordless cycles or tetras that is crucial for $\bar a \in \Pi(D_M)$. 
Importantly, all patterns in  $S$ are isomorphic and  contained in  bags of $D_M$ that are isomorphic.
We construct $D_i$ by  replacing non-crucial cycles or tetras in $D_{i-1}$ with homomorphic preimages thereof,  preserving at least one derivation of $\bar a$ from $\Pi$ in $D_i$.
Our aim is then to identify all 
(isomorphic) patterns in $S$  into a single pattern $X$ in order to obtain the desired crucial singleton set $\{ X \}$.
Doing that in a naive way,
however, may enable 
new derivations, potentially making $\{ X \}$ non-crucial.
This can happen in two ways: (i)~identification may introduce new cycles short enough to enable new matches of a rule body, and (ii)~it may bring two bags that were previously far apart close to each other, thereby enabling new matches of rule bodies.
In the second step, we therefore construct  another sequence $D_M, D_{M + 1}, \dots, D_{M + N}$, using a variation of the construction of the first sequence. The result is that there is now a set $S'$  
of either chordless cycles or tetras that is crucial for $\bar a \in \Pi(D_{M+N})$
such that (i$'$)~all elements of $S'$ are far apart from one another and (ii$'$)~there exist isomorphisms from $D_{M+N}$ to $D_{M+N}$ that map any one hard pattern from $S'$ to any other one. Since  (i$'$) overcomes problem~(i) and 
(ii$'$) overcomes problem (ii), in
the third step we may  safely identify all patterns in $S'$.
%

\medskip

With Theorem~\ref{theorem:one-cycle-is-crucial} in place, it remains to prove the following.
\begin{restatable}{lemma}{lemdataloglowercyclecombined}
\label{lem:datalog-lower-cyclecombined}
  Let $\Pi$ be a Boolean  MDLog program such that there exist a database $D$ with  $D \models \Pi$
  and  a chordless cycle or tetra $X$ in  $D$ 
  with   $\{ X \}$  crucial for
  $D \models \Pi$. Then 
  falsity of\/ $\Pi$ is not constant query testable with one-sided error.
\end{restatable}
We prove Lemma~\ref{lem:datalog-lower-cyclecombined} by a 
reduction  from testing falsity of a self-join free CQ $q$ that takes the form of a chordless cycle or a tetra, depending on what $X$ is. These are not constant query testable with one-sided error by the results in~\cite{chen2019testability}. 
%
With every input database $E$ for $q$, we associate an input database $E'$ for $\Pi$
by starting from the database~$D$ from Theorem~\ref{theorem:one-cycle-is-crucial}
and then replacing each constant $c$ of~$X$ by copies $(c,v)$ for every $v \in \mn{adom}(E)$; similar constructions can be found, e.g.,
in \cite{tagging-bagan, tagging-nofar,chen2019testability}. The facts in~$E'$ that involve these copies encode the facts of~$E$ while the remainder of~$D$ is copied to~$E'$ verbatim.  If $E\models q$, then a witnessing homomorphism selects one copy $(c,v)$ of every $c\in X$. Replacing the constants of~$X$ by these selected copies in a derivation of~$\Pi$ in~$D$ gives a derivation of
 $\Pi$ in $E'$. Conversely, suppose that $E'\models\Pi$. Erasing the second coordinate of every constant that occurs in a derivation of $\Pi$ in $E'$ gives a derivation of $\Pi$ in $D$. 
 Since $\{X\}$ is crucial, some rule application in this derivation uses $X$. The corresponding rule application in $E'$ uses copies $(c,v_c)$ of the constants $c\in X$. By the construction of $E'$, the second coordinates of these copies determine a match of $q$ in $E$. The multiplicities in $E'$ are chosen so that $|E'|=|D||E|$, which ensures that $\varepsilon$-farness is preserved up to the constant factor $1/|D|$. 



\section{Conclusion}

We established the fundamental result that for every 2RPQ $q$, non-answers to $q$ are constant query testable with one-sided error. We then launched
an investigation of monadic Datalog programs. Our technically most involved result is that falsity of Boolean MDLog programs that are not equivalent to an $\alpha$-acyclic MDLog program is not constant query testable with one-sided error, under the assumption of self-join freeness. Note that 
this assumption is only used in the proof of Theorem~\ref{theorem:one-cycle-is-crucial}, but not in 
the proof of Lemma~\ref{lem:datalog-lower-cyclecombined}. We expect that it can be removed at the expense of making the already rather technical proof of Theorem~\ref{theorem:one-cycle-is-crucial} even more complex, but leave details for future work. As discussed in Section~\ref{sec:musing}, the status of a large remaining class of MDLog programs remains open. One possible approach to showing that falsity of all $\alpha$-acyclic MDLog programs is constant-query testable might be to generalize
the proof of Lemma~\ref{lem:problemtwo}. We leave this for future work.
 It would also be interesting to understand how bag semantics and set semantics relate to one another. Can we improve our lower bounds from bag semantics to set semantics? Are there MDLog queries that are constant query testable under set semantics, but not under bag semantics?

\begin{acks}
Isolde Adler was supported by the Deutsche Forschungsgemeinschaft (DFG, German Research Foundation), project number 580040086. Carsten Lutz was supported by the DFG project LU 1417/4-1. Quentin Manière was supported by the ANR grant EXPAND (ANR-25-CE23-1215).
\end{acks}

\newpage
 
\bibliographystyle{alpha}
\bibliography{biblio}

\newpage

\appendix

\section{Proofs for Section~\ref{sect:RPQs}}
\label{app:proofssect3}

The following is easy to see, we omit a proof.
\begin{lemma}
  \label{lem:productlem}
  Let $D$ be a database, $q$ a 2RPQ, both over the same schema
  $\schemaS$, and $a_1,a_2 \in \adom(D)$.  Then $(a_1,a_2) \in q(D)$
  iff there is an $s_F \in F$ such
  that $\langle a_2,s_F \rangle$ is reachable from
  $\langle a_1,s_0 \rangle$ in $G_{D \times q}$.
\end{lemma}
%

\lemefartranslates*
\begin{proof}
  Assume that $D$ is $\epsilon$-far from $(a_1,a_2) \notin q(D)$ and take any submulti-digraph $G$ of $G_{D \times q}=(V,E,L)$ obtained by
  deleting $n \leq \frac{\epsilon}{|\Delta|}|E|$
  edges.   
  We have to show that $T$ is reachable from $S$ in~$G$.

  We may assume w.l.o.g.\ that 
  \begin{itemize}
      
  \item[($*$)] whenever
  an edge $\langle \alpha,\delta,i \rangle$
  is in $G$, then all edges $\langle \alpha,\delta,j \rangle$ with $1 \leq j < i$
  are also in $G$. 
  
  \end{itemize}
This is because there clearly is an automorphism on $G_{D \times q}$ that maps  $\langle \alpha,\delta,i \rangle$ to $\langle \alpha,\delta,j \rangle$, for any $i,j$.
  If we delete an edge $\langle \alpha,\delta,i \rangle$ we can thus always
  choose this edge so that $i$ is maximal, and
  the resulting graph will be isomorphic to the one that we would have obtained when choosing $i$ non-maximal. In particular, 
  $T$ is reachable from $S$ in~$G$ either in both of these graphs or in none of them.
  
  Let $D'$ be the subdatabase of $D$ that is obtained in the following way. Consider
  any fact $\alpha \in D$. The number of
  edges in $G_{D \times q}$ 
  that derive from $\alpha$ is at most
  $|\Delta| \cdot D(\alpha)$. Assume that
  $k$ of these edges have been deleted
  in $G$. We then delete $\min\{k,D(\alpha)\}$
  many copies of $\alpha$ in the construction of $D'$.
  
  Because every edge in $E$ derives
  from a unique fact in $D$, at most $n$ (copies of) facts are deleted from $D$ when
  constructing $D'$.  Since $|E| \leq |D| \cdot |\Delta|$,
  this means
  $n \leq \epsilon|D|$. 
  Thus, $D'$ is obtained
  from $D$ by deleting at most $\epsilon|D|$ facts.  It follows that
  $(a_1,a_2) \in q(D')$ since
  $D$ is $\epsilon$-far from $(a_1,a_2) \notin q(D)$. By Lemma~\ref{lem:productlem},
  it follows that $T$ is reachable from $S$ in~$G_{D' \times q}$ .

  It thus remains to observe that  $G_{D' \times q}$ is a subgraph of $G$.  In fact, any edge $\langle \alpha,\delta,i \rangle$ 
  from $E$ that is not present in $G$
  is also not present in $G_{D' \times q}$ because, by ($*$) and construction of $D'$, $\langle \alpha,\delta,i \rangle$ missing in $G$ implies that $D'(\alpha) < i$.
  Also note that $G_{D' \times q}$ may be a strict subgraph of $G$: it may be
  that an
  edge $\langle R(a_1,a_2),(s_1,R,s_2),1\rangle$ 
  was deleted in $G$, implying that $D'(R(a_1,a_2))=0$, which however 
  may result in some edge $\langle R(a_1,a_2),(s'_1,R,s'_2),1\rangle$ that is present
  in $G$ also missing from $G_{D' \times q}$.
\end{proof}

%


\section{Proofs for Section~\ref{subsect:linearacyclic}}

\lemreflexiveloopremoval*
\begin{proof}
    We prove the lemma by a reduction from testing $\mathcal{P}_\Pi$ to testing $\mathcal{P}_{\Pi'}$.
    Let $\Pi$ be over some schema $\schemaS$ and let $\schemaS' = \schemaS \uplus \{N_1, N_2\}$ with
    $N_1, N_2$ fresh binary relation symbols.
    We associate with every input $D$ to $\mathcal{P}_\Pi$ an $\schemaS'$-database $D'$ and the
    property $\mathcal{P}_{\Pi'}$.
    Moreover with any $\varepsilon \in (0,1)$ we associate $\varepsilon' = \frac{\varepsilon}{3}$.
    We prove the three conditions of Definition~\ref{def:proptestreduction} below, which yields the
    lemma. First, let us recall $\Pi'$ and construct $D'$.

    Recall that $\Pi'$ is obtained from $\Pi$ by taking a rule $\hat\rho = \alpha \leftarrow q \in \Pi$
	and a variable $x \in \mn{var}(q)$, removing all reflexive loops $F = \{R(x,x) \mid R(x,x) \in q\}$
	from $q$, consistently replacing $x$ with a fresh variable $v$ in the remaining body, and adding the
	atoms $N_1(x,u), N_2(u,v)$ with $u$ a fresh variable. The rule head $\alpha$ is left untouched. We call the resulting rule $\hat \rho '$. 
    Note that $\Pi'$ is still linear and acyclic and as mentioned in the main part of the paper can be easily rewritten into an equivalent set of strongly linear rules.

	For an $\schemaS$-database $D$, we construct the $\schemaS'$-database $D'$ as follows. We keep all
	$\schemaS$-facts, that is 
    $$D'(S(d,e)) := D(S(d,e))$$ 
    for every $S \in \schemaS$. For every $d \in \adom(D)$, introduce one fresh constant $a_d$ and set:
	\begin{align*}
		D'(N_1(d, a_d)) &:= m_d \\
		\text{if $F[x/d] \subseteq D$, } &\text{then } D'(N_2(a_d, d)) := m_d \\
		&\text{otherwise } D'(N_2(a_{N_2}, a_{N_2})) := D'(N_2(a_{N_2}, a_{N_2})) + m_d
	\end{align*}
	where $m_d := \sum_{S \in \schemaS} D(S(d, *))$ is the outdegree of $d$, and $F[x/d]$ is the set of facts obtained from $F$ by replacing every occurrence of $x$ with $d$.
	Since $\sum_{d \in \adom(D)} m_d = \sizeof{D}$, it follows that $\sizeof{D'} = 3\sizeof{D}$.

	Intuitively, the $N_1$ edge navigates from a real element $d$ to its private vertex $a_d$, and the
	$N_2$ edge returns to $d$ \emph{exactly when} the loops in $F$ hold at $d$, forcing the fresh variable
	$v$ back onto the image of $x$. When the loops do not hold, the $N_2$ edge is added as a selfloop to the global dummy element $a_{N_2}$ instead, so that we do not create a $N_1 N_2$ path but still keep the size consistent for size queries of the form $N_2(*,*)$. The multiplicity $m_d$ is chosen as the outgoing degree of $d$ for two reasons: it makes the total $N_1$- and $N_2$-multiplicity equal to $\sizeof{D}$ and it dominates
	the cost $\ell(d) := \min_{R(x,x) \in F} D(R(d,d))$ of breaking the loops at $d$.

    Now we prove the three conditions of Definition~\ref{def:proptestreduction}.

	\emph{(1) If $D \not\models \Pi$, then $D' \not\models \Pi'$.}
	Assume by contrapositive that $D' \models \Pi'$, witnessed by a derivation
	$\Gamma' = (V', E', \ell', \rho', h')$ from $\Pi'$ in $D'$.
We construct a derivation $\Gamma = (V, E, \ell, \rho, h)$ from $\Pi$ in $D$ as follows.
We replace every $\hat\rho'$ vertex with a $\hat\rho$ vertex, replace its $N_1, N_2$
children by children for the facts in $F$, and define the homomorphism labels as follows.
For a node $w$ with $\rho'(w) \neq \hat\rho'$, the rule is unchanged and we set
$h(w) := h'(w)$. For a $\hat\rho'$ vertex $\hat v$, whose rule $\hat\rho$ has body variables
$x$ and $z$, we set $h(\hat v)(x) := h'(\hat v)(v)$ and $h(\hat v)(z) := h'(\hat v)(z)$.

Consider any atom $S(t, t')$ of $\hat\rho$'s body with $S \in \schemaS$ and
$S(t,t') \notin F$; in $q'$ it appears as $S(\bar t, \bar t')$ where $\bar t, \bar t'$
are $t, t'$ with $x$ renamed to $v$. Since $D$ and $D'$ agree on $\schemaS$-facts and
$h(x) = h'(v)$, $S(h'(\bar t), h'(\bar t')) \in D'$ gives $S(h(t), h(t')) \in D$.
Atoms of other rules are unaffected by the rename and transfer directly.

	It remains to treat the loops in $F$. Let $d = h'(w)(x)$. By definition of $N_1$ in $D'$, the only
	$N_1$ edge leaving $d$ points to $a_d$, so $h'(u) = a_d$. The atom $N_2(u, v)$ then requires an $N_2$
	edge leaving $a_d$, which by construction exists only if $F[x/d] \subseteq D$, and in that case points
	to $d$. Hence $h'(v) = d = h'(x)$ and $F[x/d] \subseteq D$. The latter means every loop $R(x,x) \in F$
	maps to $R(d,d) \in D$, providing the required leaf children. Thus $\Gamma$ is a derivation and
	$D \models \Pi$.

	\emph{(2) If $D$ is $\epsilon$-far from $D \not\models \Pi$, then $D'$ is $\epsilon'$-far from
	$D' \not\models \Pi'$.}
	By contrapositive, assume $D'$ is not $\epsilon'$-far from $D' \not\models \Pi'$: there exists a
	sub-database $E' \subseteq D'$ with $\sizeof{E'} \leq \epsilon'\sizeof{D'}$ (hence
	$\sizeof{E'} \leq \epsilon\sizeof{D}$) and $D' \setminus E' \not\models \Pi'$. Choose $E'$ minimal w.r.t.\
	pointwise multiplicity. In particular facts in $E'$ then have the same multiplicity as in $D'$. We construct $E \subseteq D$ with
	$\sizeof{E} \leq \epsilon\sizeof{D}$ and $D \setminus E \not\models \Pi$. 

    Let $L = \{d \in \adom(D) \mid N_1(d, a_d) \in E' \text{ or } N_2(a_d, d) \in E'\}$.
    We initialize $E$ with $E := (E' \cap D)$ and then for all $d \in L$, we choose a fact $R(d,d)$ in $F[x/d]$ with $D(R(d,d)) = \ell(d)$ and set $E(R(d,d)) := \ell(d)$.
    
	We first check $\sizeof{E} \leq \sizeof{E'}$. Each $d \in L$ contributes a distinct $N_1$ or $N_2$ edge of multiplicity $m_d$ to $E'$ so
	$\sum_{d \in L} \ell(d) \leq \sum_{d \in L} m_d \leq \sum_{d,e \in \adom(D)}\sizeof{E'(N_1(d,e) + E'(N_2(d,e))}$. Hence
	$\sizeof{E} \leq \sizeof{E' \cap D} + \sum_{d \in L} \ell(d) \leq \sizeof{E'}$.

	It remains to argue $D \setminus E \not\models \Pi$. Since $E'$ is minimal, every fact in $E'$ is helping
	to block some homomorphism from a rule body of $\Pi'$ into $D'$. For $\schemaS$-atoms this transfers to
	$\Pi$ and $D$ since $E$ contains $E' \cap D$. The only rule that could behave differently is $\hat\rho$.
	Assume by contradiction that $D \setminus E \models \Pi$ via some derivation $\Gamma$ from $\Pi$ in $D$. We show that then $D' \setminus E' \models \Pi'$. 
    By the structure
	of $\Pi$ there is a $\hat\rho$ vertex; let it map $x$ to $d$. As $\Gamma$ maps the body of $\hat\rho$
	into $D \setminus E$, we have $F[x/d] \subseteq D \setminus E$. In particular $d \notin L$ (otherwise our definition of $E$ would break $F[x/d]$ mapping to $D \setminus E$), so $E'$ contains neither
	$N_1(d, a_d)$ nor $N_2(a_d, d)$. Since $F[x/d] \subseteq D$, the fact $N_2(a_d, d)$ is present in $D'$, and both $N_1(d, a_d)$ and $N_2(a_d,d)$ survive in $D' \setminus E'$. The remaining $\schemaS$-atoms of $\hat\rho$ survive
	in $D' \setminus E'$ as well (their facts lie in $D \setminus E$, hence outside $E' \cap D \subseteq E'$).
	Thus $\hat\rho'$ can fire at $d$ (mapping $u \mapsto a_d$, $v \mapsto d$) and the remainder of $\Gamma$
	lifts, giving $D' \setminus E' \models \Pi'$, a contradiction.

	\emph{(3) The answer to any completion or size query to $D'$ can be answered by constantly many
	completion or size queries to $D$.}
	Completion (resp.\ size) queries on $D'$ regarding $\schemaS$-predicates are answered via the same
	query to $D$. We treat the fresh predicates $N_1, N_2$. For brevity we sometimes write $F[x/d] \subseteq D$ which can be decided by a size query $R(d,d)$ for each $R(x,x) \in F$, and $m_d$ can be obtained from the size queries $S(d,*)$ for 
	$S \in \schemaS$. As $|\schemaS|$ and $F$ are fixed, each case below uses constantly many queries to $D$.
	\begin{itemize}
		\item Completion query $N_1(c, *)$ returns $a_c$ if $c = d \in \adom(D)$, and $\bot$ otherwise.
		\item Completion query $N_1(*, c)$ returns $d$ if $c = a_d$, and $\bot$ otherwise.
		\item Completion query $N_2(c, *)$. If $c = a_d$, return $d$ if $F[x/d] \subseteq D$ and $\bot$ if
		not. If $c = d \in \adom(D)$, return $a_d$ if $F[x/d] \not\subseteq D$ and $\bot$ if not.
		Otherwise return $\bot$.
		\item Completion query $N_2(*, c)$. If $c = d \in \adom(D)$, return $a_d$ if $F[x/d] \subseteq D$ and
		$\bot$ if not. If $c = a_d$, return $d$ if $F[x/d] \not\subseteq D$ and $\bot$ if not. Otherwise
		return $\bot$.
        \item Completion query $N_1(*, *)$. We sample a relation symbol $S \in \schemaS$ with
        probability proportional to $\sizeof{S}$ (obtained via the size queries $S(*,*)$), then ask
        a completion query $S(*,*)$ to obtain a fact $S(d,e)$, and return $N_1(d, a_d)$. Since
        sampling a random $\schemaS$-fact and taking its first argument yields $d$ with probability
        $m_d / \sizeof{D}$, this returns each $N_1$-fact with probability proportional to its
        multiplicity, as required.
        \item Completion query $N_2(*, *)$. We sample a fact $S(d,e)$ as above. If $F[x/d] \subseteq D$
        we return $N_2(a_d, d)$, and otherwise we return $N_2(d, a_d)$. As both
        have multiplicity $m_d$, and $d$ is sampled with probability
        $m_d / \sizeof{D}$, this returns each $N_2$-fact with probability proportional to its
        multiplicity, as required.
		\item Size query $N_1(c, *)$ returns $m_c$ if $c = d \in \adom(D)$, and $0$ otherwise.
		\item Size query $N_1(*, c)$ returns $m_d$ if $c = a_d$, and $0$ otherwise.
		\item Size query $N_2(c, *)$ returns $m_d$ if $c = a_d$ and $F[x/d] \subseteq D$, or if
		$c = d \in \adom(D)$ and $F[x/d] \not\subseteq D$; otherwise $0$.
		\item Size query $N_2(*, c)$ returns $m_d$ if $c = d \in \adom(D)$ and $F[x/d] \subseteq D$, or if
		$c = a_d$ and $F[x/d] \not\subseteq D$; otherwise $0$.
		\item Size queries $N_1(*,*)$ and $N_2(*,*)$ return $\sizeof{D}$.
	\end{itemize}
\end{proof}

\lemparalleledgeremoval*
\begin{proof}
    We prove the lemma by a reduction from $\mathcal{P}_\Pi$ to $\mathcal{P}_{\Pi'}$.
    Let $\Pi$ be a strongly linear MDLog program of schema $\schemaS$ and let
    $\schemaS' = \schemaS \uplus \{N_1, N_2, N_3\}$ with $N_1, N_2, N_3$ fresh binary relation symbols.
    We associate with every input $D$ to $\mathcal{P}_\Pi$ an $\schemaS'$-database $D'$ and the
    property $\mathcal{P}_{\Pi'}$.
    Moreover with any $\varepsilon \in (0,1)$ we associate $\varepsilon' = \frac{\varepsilon}{7}$.
    We prove the three conditions of Definition~\ref{def:proptestreduction} below, which yields the
    lemma. First, let us recall $\Pi'$ and construct $D'$.

    Recall that $\Pi'$ is obtained from $\Pi$ as follows. Let $\hat \rho = \alpha \leftarrow q \in \Pi$ and $x,y \in \mn{var}(q)$, and let $F = \{R(x,y) \mid R(x,y) \in q\} \cup \{R(y,x) \mid R(y,x) \in q\}$ be the set of parallel edges between $x$ and $y$ in $q$. Then $\Pi'$ is obtained from $\Pi$ by removing $F$ from $q$, and  adding to $q$ the atoms $N_1(x,z_1),N_2(z_1,z_2),N_3(z_2,y)$ where $z_1,z_2$ are fresh variables.
    Note that $\Pi'$ is still linear and acyclic and as mentioned in the main part of the paper can be easily rewritten into an equivalent set of strongly linear rules.
    
    For an $\schemaS$-database $D$, we construct the $\schemaS'$-database $D'$ as follows. 
    For every fact $S(d,e) \in D$ with multiplicity $m := D(S(d,e))$, we introduce three fresh constants $a_{S(d,e)}$, $b_{S(d,e)}, c_{S(d,e)}$ and set:
    \begin{align*}
        D'(S(d,e)) &:= m\\
        D'(N_1(d, a_{S(d,e)})) &:= m \qquad 
        D'(N_1(e, b_{S(d,e)})) := m\\
        D'(N_3(a_{S(d,e)}, d)) &:= m \qquad
        D'(N_3(b_{S(d,e)}, e)) := m \\
        \text{if $F[x/d, y/e] \subseteq D$, } & \text{then } D'(N_2(a_{S(d,e)}, b_{S(d,e)})) := m \\
        &\text{otherwise } D'(N_2(a_{S(d,e)}, c_{S(d,e)})) := m\\
        \text{if $F[x/e, y/d] \subseteq D$, } & \text{then } D'(N_2(b_{S(d,e)},a_{S(d,e)})) := m\\
         &\text{otherwise } D'(N_2(b_{S(d,e)}, c_{S(d,e)})) := m
    \end{align*}
    where $F[x/d, y/e]$ is the set of facts obtained from $F$ by replacing every occurrence of $x$ in a fact with $d$ and every occurrence of $y$ with $e$.
    It follows that $\sizeof{D'} = 7 \sizeof{D}$. 
    
    Intuitively, $N_1$ and $N_3$ edges are there to isolate the $N_2$ edges (which do the actual work) for each fact and that isolation is necessary for simulating completion queries. 
    An $N_2$ edge from $a_{S(d,e)}$ to $b_{S(d,e)}$ represents that the parallel edges in $F$ exist in the direction $d$ to $e$ (and analogously for the inverse direction). 
    If this is not the case we instead add a dummy edge to $c_{S(d,e)}$ which is needed to simulate size queries of the form $N_2(*,*)$.

    Now we prove the three conditions of Definition~\ref{def:proptestreduction}.

    \emph{(1) If $D \not\models \Pi$, then $D' \not\models \Pi'$.} Assume by contrapositive, that $D' \models \Pi'$. Let $\Gamma' = (V', E', \ell', \rho', h')$ be a derivation from $\Pi'$ in $D'$. 
    Recall that to construct $\Pi'$ we took a rule $\hat \rho = \alpha \leftarrow q \in \Pi$ and changed it to $\hat{\rho}' = \alpha \leftarrow q'$ with $q' = q \setminus F \cup N_1(x, z_1), N_2(z_1,z_2), N_3(z_2,y)$. 
    We construct a derivation $\Gamma = (V, E, \ell, \rho, h)$ from $\Pi$ in $D$. 
    It is very similar to $\Gamma'$, we only replace $\hat \rho'$ vertices with $\hat \rho$ vertices and replace the $N_1, N_2, N_3$ children (of those nodes) with children for the facts in $F$. 
    For any node $v \in V$ with rule $\rho(v) = P(x) \leftarrow q$, the homomorphism $h(v)$ is simply $h(v) := h'(v)|_{\mn{var}(q)}$.
    For all rule bodies $q$ in $\Pi$ and $S \in q \setminus F$ it is immediate that $S(h'(z), h'(z')) \in D'$ implies $S(h(z), h(z')) \in D$. 
    
    We turn to the atoms in $F$ recalling that they are binary with variables $x$ and $y$.  
    We let $d = h'(x)$ and $e = h'(y)$.
     By definition of relations $N_1$ and $N_3$ in $D'$, it is immediate that either $h'(z_1) = a_{S(d, e)}$ and $h'(z_2) = b_{S(d, e)}$, or $h'(z_1) = b_{S(e, d)}$ and $h'(z_2) = a_{S(e, d)}$, for some relation $S$.
        We treat the first case, the second being symmetrical.
	By definition of $h'$, we have $N_2(a_{S(d, e)}, b_{S(d, e)}) \in D'$, thus, by definition of $N_2$, we know that $F[x/d, y/e] \subseteq D$ which is the desired conclusion.

    \emph{(2) If $D$ is $\epsilon$-far from $D \not\models \Pi$, then $D'$ is $\epsilon'$-far from $D' \not\models \Pi'$.}
    By contrapositive, assume $D'$ is not $\epsilon'$-far from $D' \not\models \Pi'$:
		there exists a sub-database $G' \subseteq D'$ such that $|G'| \leq \epsilon' \sizeof{D'}$ (and by size of $D$ and $\varepsilon$ then also $|G'| \leq \epsilon \cdot \sizeof{D}$), and $D' \setminus G' \not\models \Pi'$.
    Choose such a $G'$ that is minimal w.r.t.\ pointwise multiplicity.
    We need to prove that $D$ is not $\epsilon$-far from $D \not\models \Pi$.
    To this end, we construct a sub-database $G$ of $D$ such that $|G| \leq \epsilon \sizeof{D}$ and $D \setminus G \not\models \Pi$. We define for every $S(d,e) \in D$:
    \begin{align*}
        m_1 &:= G'(S(d,e))\\
        m_2 &:= \max (G'(N_1(d, a_{S(d, e)})), G'(N_2(a_{S(d, e)}, b_{S(d, e)})), G'(N_3(b_{S(d, e)}, e)))\\
        m_3 &:= \max (G'(N_1(e, b_{S(d, e)})), G'(N_2(b_{S(d, e)}, a_{S(d, e)})), G'(N_3(a_{S(d, e)}, d)))\\
        G(S(d,e)) &:= \max(m_1,m_2,m_3)
    \end{align*}

    It is clear that $|G| \leq |G'|$ so it remains to argue that $D \setminus G \not \models \Pi$.
    Since $G'$ is minimal, every fact in $G'$ is helping by making some homomorphism from some rule body of $\Pi'$ into $D'$ impossible. 
    In particular, it is then clear that facts in $G'$ all have the same multiplicity as in $D'$.
    For relation symbols $S \in \schemaS$, this obviously transfers to $\Pi$ and $D$ (by definition of $m_1$). The only rule left in $\Pi$ which could behave differently is the rule $\hat \rho$. 
    Assume by contradiction that $D \setminus G \models \Pi$ via some derivation $\Gamma = (V,E,\ell,\rho,h)$ from $\Pi$ in $D$. We show that then $D' \setminus G' \models \Pi'$, contradicting our initial assumption.
    By our earlier reasoning there must be a $\hat \rho$ vertex $v \in V$. Let $h(v)$ then map $x$ to $d$ and $y$ to $e$.
    Since no element of $F[x/d,y/e]$ can be in $G$ for $h$ to successfully map the body of $\hat \rho$, $m_2$ and $m_3$ for all those facts have to be $0$. It is then straightforward to see that we can map the body of $\hat \rho'$ to $D' \setminus G'$ accordingly. In particular, assume we have an $R(d,e) \in F[x/d,y/e]$, then $N_2(a_{R(d,e)}, b_{R(d,e)})$ is in $D'$ because $F[x/d,y/e] \subseteq D$ and $m_2 = 0$ implies that the forward path $N_1(d,a_{R(d,e)}), N_2(a_{R(d,e)}, b_{R(d,e)}), N_3(b_{R(d,e)},e)$ survives in $D' \setminus G'$ (analogously for $m_3$ and the backwards path).
    
    \emph{(3) The answer to any completion or size query to $D'$ can be answered by constantly many completion or size queries to $D$.}
    Completion (resp.\ size queries) on $D'$ regarding predicates from the schema $\schemaS$ are answered via the same query to $D$.
	We now treat the case of completion and size queries involving the fresh predicates $N_1$, $N_2$ and $N_3$. As $|\schemaS|$ and $F$ are fixed, each case below uses constantly many queries to $D$.
        \begin{itemize}
            \item Completion query $N_1(c, *)$ returns $\bot$ if $c \notin \adom(D)$, otherwise we use size queries $S(c, *)$ and $S(*, c)$ to compute the number of neighbors for $c$ for each relation $S \in \schemaS$. We then select one such relation $S_0$ at random according to the distribution given by their relative contributions to the overall number of neighbors of $c$ in $D$. We then use a completion query $S_0(c, *)$ (resp. $S_0(*, c)$) to obtain an element $e$ and we return $a_{S_0(c, e)}$ (resp. $b_{S_0(e, c)}$), which is as desired.
            \item Completion query $N_1(*, c)$ returns $d$ if $c = a_{S(d, e)}$, $e$ if $c = b_{S(d, e)}$, and $\bot$ otherwise.
            \item Completion queries $N_3(*,d)$ and $N_3(c,*)$ work symmetrically to those for $N_1$.
            \item Completion query $N_2(c, *)$. If $c$ is neither $a_{S(d,e)}$ nor $b_{S(d,e)}$ for some $S(d,e) \in D$, then return $\bot$. 
            Otherwise, if $c = a_{S(d, e)}$ then return $b_{S(d, e)}$ if $F[x/d, y/e] \subseteq D$ and return $c_{S(d,e)}$ if not.
            Lastly, if $c = b_{S(d, e)}$ then return $a_{S(d, e)}$ if $F[x/e, y/d] \subseteq D$ and return $c_{S(d,e)}$ if not.
            To check whether $F[x/d, y/e] \subseteq D$ simply ask size queries for all facts in $F[x/d, y/e]$ to $D$.
            \item Completion query $N_2(*,c)$.
            If $c$ is neither $a_{S(d,e)}$, nor $b_{S(d,e)}$, nor $c_{S(d,e)}$ for some $S(d,e) \in D$, then return $\bot$. 
            Otherwise if $c = b_{S(d,e)}$, then return $a_{S(d,e)}$, if $F[x/d,y/e] \subseteq D$ and $\bot$ if not. 
            Otherwise if $c = a_{S(d,e)}$, then return $b_{S(d,e)}$, if $F[x/e,y/d] \subseteq D$ and $\bot$ if not. 
            And lastly if $c = c_{S(d,e)}$, then let $P = \{\}$ be the set representing possible predecessors. We add $a_{S(d,e)}$ to $P$, if $F[x/d,y/e] \not \subseteq D$ and add $b_{S(d,e)}$ to $P$, if $F[x/e,y/d] \not \subseteq D$. Then we choose an element uniformly from $P$ and return it or return $\bot$ if $P$ is empty.
            \item Completion query $N_1(*, *)$. We sample a fact $S(d,e)$ of $D$ with probability
            proportional to its multiplicity $m := D(S(d,e))$ (via size queries $S(*,*)$ for $S \in \schemaS$
            followed by a completion query), then return $N_1(d, a_{S(d,e)})$ or $N_1(e, b_{S(d,e)})$, each
            with probability $\frac{1}{2}$. As $S(d,e)$ is sampled with probability $m/\sizeof{D}$ and each
            fact contributes exactly these two $N_1$-edges of multiplicity $m$, every $N_1$-fact is returned
            with probability proportional to its multiplicity, as required.
            
            \item Completion query $N_3(*, *)$ works symmetrically, returning $N_3(a_{S(d,e)}, d)$ or
            $N_3(b_{S(d,e)}, e)$, each with probability $\frac{1}{2}$, for a sampled fact $S(d,e)$.
            
            \item Completion query $N_2(*, *)$. We sample a fact $S(d,e)$ as above and then take the set $E$ of two $N_2$ edges constructed for $S(d,e)$ in $D'$. So if $F[x/d, y/e] \subseteq D$, then $N_2(a_{S(d,e)}, b_{S(d,e)}) \in E$ and if not then $N_2(a_{S(d,e)}, c_{S(d,e)}) \in E$. Analogously for $F[x/e, y/d] \subseteq D$.
            We then pick one of those two edges in $E$ by choosing uniformly with probability $\frac{1}{2}$.
            Each fact contributes exactly these two $N_2$-edges of
            multiplicity $m$, so every $N_2$-fact is returned with probability proportional to its multiplicity,
            as required.
            \item Size query $N_1(c,*)$ returns $\sum_{S \in \schemaS} D(S(c,*)) + D(S(*,c))$.
            \item Size query $N_1(*,c)$. If $c = a_{S(d,e)}$ or $c = b_{S(d,e)}$ for some $S \in \schemaS$, then we return $D(S(d,e))$. Otherwise return $0$.
            \item Size queries $N_3(*,c)$ and $N_3(c,*)$ work symmetrically to those for $N_1$.
            \item Size query $N_2(c,*)$. If $c = a_{S(d,e)}$ or $c = b_{S(d,e)}$ for some $S \in \schemaS$, then return $D(S(d,e))$. Otherwise return $0$.
            \item Size query $N_2(*,c)$.
            If $c = b_{S(d,e)}$, return $D(S(d,e))$ if $F[x/d, y/e] \subseteq D$ and $0$, if not.
            Otherwise, if $c = a_{S(d,e)}$, return $D(S(d,e))$ if $F[x/e, y/d] \subseteq D$ and $0$, if not. 
            Otherwise, if $c = c_{S(d,e)}$, check if $F[x/d, y/e] \subseteq D$ and if $F[x/e, y/d] \subseteq D$. If both are true, return $0$; if one is true return $D(S(d,e))$ and if both are false return $2\cdot D(S(d,e))$.
            Lastly if $c$ is neither of those $3$ return $0$.
            \item Size queries $N_1(*,*)$, $N_2(*,*)$, or $N_3(*,*)$ return $2 \cdot |D|$.
        \end{itemize}
        Note that $D(S(d,e))$ can simply be checked by the value of a size query $S(d,e)$ to $D$ and 
        $F[x/d,y/e] \subseteq D$ can be checked by doing a size query for every fact in $F[x/d,y/e]$.
\end{proof}

\fromMDLogtoRPQs*
\begin{proof}
  What remains is to define the transition relation $\Delta$ and prove the three Conditions of Definition~\ref{def:proptestreduction}.
    Let us first introduce some notation for defining the transitions in $\Delta$.
    For each rule $\rho \in \Pi$, we fix a \emph{source variable} $x^\mn{src}_\rho$ and \emph{target variable}
    $x^\mn{tgt}_\rho$:
    \begin{itemize}
        \item if $\rho$ has an IDB-atom $P(x)$ in its body, we set $x^\mn{src}_\rho = x$ and $x^\mn{tgt}_\rho = y$;
        \item if $\rho$ has a non-goal head $P(y)$, we set $x^\mn{tgt}_\rho = y$ and $x^\mn{src}_\rho = x$;
        \item if $\rho$ has neither of those, we assign $x$ and $y$ arbitrarily to $x^\mn{src}_\rho$ and $x^\mn{tgt}_\rho$.
    \end{itemize}
Similarly for each rule $\rho$ we fix a \emph{source state} $s^\mn{src}_\rho$ and \emph{target state} $s^\mn{tgt}_\rho$:
    \begin{itemize}
        \item if $\rho$ has an IDB atom $P(x^\mn{src}_\rho)$ in its body, then $s^\mn{src}_\rho = s_P$; otherwise $s^\mn{src}_\rho = s_{\mn{init}}$;
        \item if $\rho$ has a non-goal head atom ${P_0}$, then $s^\mn{tgt}_\rho = s_{P_0}$;
        otherwise $s^\mn{tgt}_\rho = s_{\mn{fin}}$.
    \end{itemize}
    Now we are ready to define $\Delta$:
    \begin{itemize}
      \item add $(s_0, U, s_{\mn{init}})$ and $(s_{\mn{fin}}, U^-, s_{\mn{acc}})$;
      \item for every $\rho \in \Pi$ that contains a binary atom $R(x,y)$, if $x = x^\mn{src}_\rho$, we add $(s^\mn{src}_\rho, R, s^\mn{tgt}_\rho)$, otherwise $x = x^\mn{tgt}_\rho$ and we add $(s^\mn{src}_\rho, R^-, s^\mn{tgt}_\rho)$;
      \item for every $\rho \in \Pi$ that does not contain a binary atom, we add the two transitions $(s^\mn{src}_\rho, U^-, s_\rho)$ and $(s_\rho, U, s^\mn{tgt}_\rho)$.
    \end{itemize}
    
    For an intuition see the paragraph in the main part of the paper.
    Now we prove the three conditions of Definition~\ref{def:proptestreduction}.

   \emph{(1) If $D \not\models \Pi$, then $uu \not \in q(D')$.}
We argue the contrapositive. Assume $uu \in q(D')$, witnessed by an accepting run
\[
  s_0 \xrightarrow{R_1} s_1 \xrightarrow{R_2} s_2 \xrightarrow{R_3} \cdots \xrightarrow{R_n} s_{\mn{acc}}
\]
of $q$ on a path $p = c_0 R_1 c_1 \ldots R_n c_n$ in $D'$ with $c_0 = c_n = u$, where each
$R_i \in \Sigma_{\schemaS'}$ labels a transition consistent with $D'$.

The only transition leaving $s_0$ is $(s_0, U, s_{\mn{init}})$, so $s_1 = s_{\mn{init}}$, $R_1 = U$,
and $c_1 \in \adom(D)$. Similarly, the only transition entering $s_{\mn{acc}}$ is
$(s_{\mn{fin}}, U^-, s_{\mn{acc}})$, so $s_{n-1} = s_{\mn{fin}}$ and $R_n = U^-$.
The remaining transitions $s_1 \to \cdots \to s_{n-1}$ each correspond to a rule of $\Pi$ by
construction of $\Delta$: a connected rule contributes one transition between states $s^\mn{src}_\rho$
and $s^\mn{tgt}_\rho$, and a disconnected rule contributes the pair through its rule-specific
intermediate state $s_\rho$, which can only be entered and exited by these two transitions.

Reading off these rules in order yields a sequence $\rho_1, \ldots, \rho_k \in \Pi$ such that for
$i < k$, $s^\mn{tgt}_{\rho_i} = s^\mn{src}_{\rho_{i+1}}$, that $s^\mn{src}_{\rho_1} = s_{\mn{init}}$,
and that $s^\mn{tgt}_{\rho_k} = s_{\mn{fin}}$. By definition of $s^\mn{src}$ and $s^\mn{tgt}$, this
means $\rho_1$ has no IDB atom in the body, $\rho_k$ is a goal rule, and each
intermediate $\rho_{i+1}$ contains an IDB atom $P_i(x^\mn{src}_{\rho_{i+1}})$ in its body, where $P_i$ is the head relation of $\rho_i$. 
The elements visited along the path give a homomorphism for each rule body: the connected
rule's edge is read at consecutive path elements with the orientation prescribed by $\Delta$, and the
disconnected rule's two-step detour through $u$ assigns $x^\mn{src}_\rho$ and $x^\mn{tgt}_\rho$ to the two elements
surrounding $u$, with the $\mn{true}$ atom satisfied vacuously. Linking these bottom-up yields a derivation of $\mn{goal}$ on the  sub-database of $D'$ consisting of only $\schemaS$-facts, which by construction equals $D$, so $D \models \Pi$.

\emph{(2) If $D$ is $\epsilon$-far from $D \not\models \Pi$, then $D'$ is $\epsilon'$-far
from $uu \not \in q(D')$.}
By contrapositive, assume $D'$ is not $\epsilon'$-far from $D' \not\models q(u,u)$: there exists
$E' \subseteq D'$ with $|E'| \leq \epsilon'|D'| = \epsilon|D|$ and
$(u,u) \notin q(D' \setminus E')$. We choose $E'$ minimal with respect to pointwise multiplicity which in particular implies that facts in $E'$ have the same multiplicity as in $D'$.
Let $G := \{d \in \adom(D) \mid U(u,d) \in E'\}$. We construct $E \subseteq D$ from $E'$ by
removing the $U$-facts of $E'$ and adding, for every $d \in G$, all $\schemaS$-facts incident to
$d$. Formally,
\[
  E(S(d,e)) :=
  \begin{cases}
    D(S(d,e)) & \text{if } d \in G \text{ or } e \in G, \\
    E'(S(d,e)) & \text{otherwise,}
  \end{cases}
\]
for all $S \in \schemaS$. We first verify $|E| \leq |E'|$. Each $d \in G$ contributes
multiplicity $\deg_D(d)$ to $|E'|$ via the $U$-fact $U(u,d)$. The total multiplicity of
$\schemaS$-facts incident to $d$ in $D$ is also $\deg_D(d)$ (by definition of $\deg_D$), so for
every $U(u,d)$ removed we add at most $D'(U(u,d))$ new facts to $E$. The remaining
$\schemaS$-facts of $E'$ are inherited unchanged.

It remains to show $D \setminus E \not\models \Pi$. Assume by contradiction that
$D \setminus E \models \Pi$ via some derivation $\Gamma$ of $\Pi$ in $D$. We will show that then
$(D' \setminus E') \models q(u,u)$. By strong linearity there is a simple path $v_k, \dots, v_0$
in $\Gamma$ that visits every inner node and with $v_0$ the root. We construct an accepting run $\sigma$ of $q(u,u)$ along this path, starting with the
transition $(s_0, U, s_\mn{init})$ from $u$ to $d_k := h(v_k)(x^\mn{src}_{\rho(v_k)})$.

Consider $v_i$ with $i \in \{k, \ldots, 0\}$ in turn, and let the head of $\rho(v_i)$ be $P(y)$
(with $P = \mn{goal}$ if $i = 0$). If the body of $\rho(v_i)$ is connected, it contains a binary
atom $R(x,y)$ ($R(y,x)$ is symmetrical) mapped by $h(v_i)$ to some fact
$R(d_i, e_i) \in (D \setminus E) \subseteq (D' \setminus E')$, and we extend $\sigma$ by
$(s^\mn{src}_{\rho(v_i)}, R, s_P)$ from $d_i$ to $e_i$.

Otherwise the body of $\rho(v_i)$ is disconnected and has the form $P'(x) \land \mn{true}(y)$
(or symmetrical), with $h(v_i)$ mapping $x$ and $y$ to elements
$d_i, e_i \in \adom(D \setminus E)$. Since $d_i \in \adom(D \setminus E)$, it appears in some $\schemaS$-fact of $D \setminus E$.
By construction of $E$, this means $d_i \notin G$, so $U(u, d_i) \in D' \setminus E'$, and
analogously for $e_i$. We extend $\sigma$ by
$(s^\mn{src}_{\rho(v_i)}, U^-, s_{\rho(v_i)})$ from $d_i$ to $u$ and
$(s_{\rho(v_i)}, U, s_P)$ from $u$ to $e_i$.

After processing $v_0$ we have $s_P = s_\mn{fin}$, and by the same active-domain argument
$U(u, e_0) \in D' \setminus E'$, so we close $\sigma$ with $(s_\mn{fin}, U^-, s_\mn{acc})$ from
$e_0$ back to $u$. The resulting $\sigma$ is an accepting run of $q(u,u)$ on $D' \setminus E'$, as required.

\emph{(3) The answer to any completion or size query to $D'$ can be answered by constantly many completion or size queries to $D$.}
Completion and size queries on $D'$ regarding $\schemaS$-predicates are answered via the same
query to $D$. We treat queries involving the fresh relation $U$. Throughout, $|\schemaS|$ is
fixed, so each case uses constantly many queries to $D$.

\begin{itemize}
  \item Completion query $U(c, *)$. If $c \neq u$, return $\bot$, since $D'$ has no $U$-fact with
  first argument in $\adom(D)$. Otherwise sample a fact $S(a_1,a_2)$ of $D$ uniformly at random by
  multiplicity (first sample $S \in \schemaS$ with probability $\sizeof{S}/|D|$ via size queries
  $S(*,*)$, then ask a completion query $S(*,*)$), pick a position $i \in \{1, 2\}$ uniformly, and
  return $U(u, a_i)$. The probability of returning $U(u, c')$ for any
  $c' \in \adom(D)$ is
  $\sum_S \frac{\sizeof{S}}{|D|} \cdot \frac{D(S(c',*)) + D(S(*,c'))}{\sizeof{S}} \cdot \frac{1}{2}
   = \frac{\mn{deg}_D(c')}{2|D|}$, matching $D'(U(u,c'))$ divided by the total $U$-mass $2|D|$.

  \item Completion query $U(*, c)$. If $c \in \adom(D)$, return $u$. Otherwise, return $\bot$.

  \item Completion query $U(*, *)$. Same procedure as $U(u, *)$.

  \item Size query $U(c, *)$. If $c = u$, return $2|D|$. Otherwise return $0$.

  \item Size query $U(*, c)$. If $c \in \adom(D)$, return $\mn{deg}_D(c)$. Otherwise return $0$.

  \item Size query $U(*, *)$ returns $2|D|$.
\end{itemize}
Where $c \in \adom(D)$ can be checked by doing size queries $S(c,*)$ and $S(*,c)$ for each $S \in \schemaS$ and $\mn{deg}_D(c)$ in the same way.
\end{proof}

\section{Proofs for Section~\ref{subsect:tinybranching}}

\lemfarnessrealizations*

\begin{proof}
\newcommand{\gda}{\ensuremath{G_{D \times q_\Pi}}\xspace}
\begin{figure}
    \centering
    \begin{tikzpicture}[
        >=stealth,           
        thick,               
        vertex/.style = {circle, draw, minimum size=6mm, inner sep=0pt},
        accepting/.style = {circle, draw, double, minimum size=6mm, inner sep=0pt},
        divertex/.style = {circle, draw, minimum size=8mm, inner sep=0pt}
    ]

    \node (D) at (0.7, 3.25) {$D$ (all edges multiplicity 1):};
    \node (A) at (-1, -0.75) {$q_\Pi$:};
    \node (G) at (5, 3) {$\gda$:};

    \node[vertex] (a) at (0,2.5) {$a$};
    \node[vertex] (b) at (0,0.5) {$b$};
    \node[vertex] (c) at (2.5,0.5) {$c$};
    \node[vertex] (d) at (2.5,2.5) {$d$};

    \node[vertex] (se) at (-0.5,-1.5) {$s_\epsilon$};
    \node[vertex] (sr) at (1,-1.5) {$s_R$};
    \node[vertex] (ss) at (2.5,-1.5) {$s_S$};
    \node[accepting] (st) at (4,-1.5) {$s_T$};

    \node[divertex] (ae) at (6,2.25) {$a,s_\epsilon$};
    \node[divertex] (ar) at (6,0.75) {$a,s_R$};
    \node[divertex] (be) at (8,2.25) {$b,s_\epsilon$};
    \node[divertex] (br) at (8,0.75) {$b,s_R$};
    \node[divertex] (bs) at (8,-0.75) {$b,s_S$};
    \node[divertex] (ct) at (10,1.5) {$c,s_T$};
    \node[divertex] (cs) at (10,-0.75) {$c,s_S$};
    \node[divertex] (dt) at (12,1.5) {$d,s_T$};

    \draw[->] (a) -- node[left] {$S$} (b);        
    \draw[->] (b) -- node[below] {$S$} (c);         
    \draw[->] (c) -- node[left] {$T$} (d);         
    \draw[->] (a) edge [loop left] node {$R$} (a);  
    \draw[->] (b) edge [loop left] node {$R$} (b);  
    \draw[->] (c) edge [loop right] node {$T$} (c);  

    \draw[->] (se) -- node[above] {$R$} (sr);  
    \draw[->] (sr) -- node[above] {$S$} (ss);  
    \draw[->] (ss) edge [loop above] node {$S$} (ss);  
    \draw[->] (ss) -- node[above] {$T$} (st);  
    
    \draw[->] (ae) edge (ar);  
    \draw[->] (ar) edge (bs);  
    \draw[->] (be) edge (br);  
    \draw[->] (br) edge (cs);  
    \draw[->] (bs) edge (cs);  
    \draw[->] (cs) edge (ct);  
    \draw[->] (cs) edge (dt);  

    \node (init) at (-1.5,-1.5) {};
    \draw[->] (init) -- (se);
  
  \end{tikzpicture}
  \caption{Example database, NFA, and product digraph. Disconnected vertices in $\gda$ omitted.}
  \label{fig:DtoG}
\end{figure}
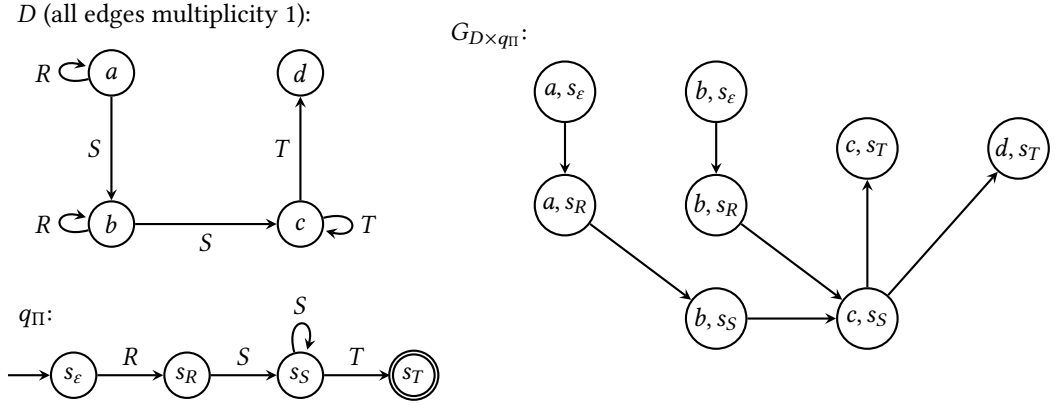
This proof is split into multiple plarts. First we note that it suffices to prove the lemma for 2RPQs and on Databases with multiplicity one. Both are pretty straightforward reductions. 
Then comes the main part of the proof, for an intuition see the paragraph in the main part of the paper.

For easier notation we use the relation symbol $R$ to represent $R_1$ and $T$ for $T_1$ so $\Pi$ is an SLP of the following form:
$$
\Pi
= \{\mn{goal}() \leftarrow P(x) \wedge  R(w,x),\quad 
P(x) \leftarrow S(x,y) \wedge P(y),\quad
P(x) \leftarrow S(x,y) \wedge T(y,z)\}.
$$

\textbf{From SLP to 2RPQs.} 
Here we show that we can consider the following 2RPQ instead of $\Pi$ to prove the lemma.
Consider the 2RPQ $q_\Pi = (Q, \Sigma_\schemaS, s_\epsilon, \{s_T\}, \Delta)$  with
\begin{align*}
    Q &= \{s_\epsilon, s_{R}, s_S, s_{T}\} \text{ and}\\
    \Delta &= \{(s_\epsilon, R, s_{R}), (s_{R}, S, s_S), (s_S, S, s_S), (s_S, T, s_{T})\}.
\end{align*}
A graphical presentation is in Figure~\ref{fig:DtoG}.
In this proof we want to use 2RPQs in a \emph{Boolean} way. We say $D$ \emph{satisfies} $q$ (written $D \models q$), if there are $a,b \in \adom(D)$ such that $(a,b) \in q(D)$.
To continue with $q_\Pi$ instead of $\Pi$ we need to show two things:
\begin{enumerate}
    \item if $D$ is $\epsilon$-far from $D \not \models \Pi$, then $D$ is $\epsilon$-far from $D \models q_\Pi$;
    \item if there are more than $\frac{\epsilon}{6}\cdot |D|$ fact-disjoint realizations of $q_\Pi$ in $D$ then there are also at least that many of $\Pi$ in $D$.
\end{enumerate}

Similar to realizations of MDLog programs, we define a \emph{realization} of a Boolean 2RPQ
$q$ in a database $E$ as a minimal database
$E' \subseteq E$ such that $E' \models q$.

We show that every realization of $q_\Pi$ in $D$ is also a realization of $\Pi$ in $D$ and vice versa:
\begin{itemize}
    \item every realization $D_r$ of $q_\Pi$ in $D$ is a path of some length $n$ and satisfies $D_r \models \Pi$ by applying, in order, the third rule of $\Pi$ once, the second rule $n-2$ times, and the goal rule once;
    \item conversely, every realization $D_r$ of $\Pi$ in $D$ gives rise to a derivation that, by the shape of $\Pi$, traces a $R \cdot S^* \cdot T$ path in $D_r$ and is therefore a realization of $q_\Pi$ in $D$.
\end{itemize}
it is clear that minimality is preserved in both directions.
Now for proving the two points from above. Point~(2) is an immediate consequence. For Point~(1), assume by contradiction that $D$ is only $\epsilon'$-far from $D \models q$ for some $\epsilon' < \epsilon$.
Let $G \subseteq D$ be a database witnessing this, that is, $|G| = \epsilon' * |D|$ and $D \setminus G \not \models q$.
But since $\epsilon > \epsilon'$, we still have $D \setminus G \models \Pi$, and we can just take a realization witnessing this and it is then also a realization of $q$ in $D \setminus G$ contradiction that such a $\epsilon'$ exists and thus proving Point~(1).
It thus suffices to prove that if $D$ is $\varepsilon$-far from $q_\Pi$ being false on $D$, then there are more than $\frac{\varepsilon}{6} \cdot |D|$ fact-disjoint realizations of $q_\Pi$ in $D$.

\medskip
\textbf{From bag to set.}
To make the next (and main part) of the proof of the lemma easier, we show here that it suffices to consider Databases of multiplicity one. To do this we have to adjust the database and query slightly, basically introducing new relation symbols for each multiplicity index.

Let $\Sbf'$ be the binary schema obtained from $\Sbf$ by replacing each relation symbol $P \in \Sbf$ with $|D|$ many fresh relation symbols $P^{1}, \dots, P^{|D|}$. 
We define the $\Sbf'$-database $D'$ on the same active domain $\adom(D') = \adom(D)$ as follows: for every fact $P(a,b) \in D$ with $D(P(a,b)) = m$, the database $D'$ contains the $m$ facts $P^{1}(a,b), \dots, P^{m}(a,b)$. 
Note that $D'$ has now only multiplicity one facts (it is a set-database) and that $|D'| = |D|$. 

We further define the 2RPQ $q'_\Pi = (Q, \Sigma_{\Sbf'}, s_\epsilon, \{s_T\}, \Delta')$ by replacing each transition $(s, P, s') \in \Delta$ with $P \in \Sbf$ by the transitions $(s, P^{i}, s')$ for $1 \leq i \leq |D|$, and likewise for transitions of the form $(s, P^-, s')$. 
In what follows, when we refer to an \emph{$R$-fact} (resp.\ \emph{$S$-fact}, \emph{$T$-fact}) of $D'$, we mean any fact in $D'$ that uses a relation symbol of the form $R^{i}$ (resp.\ $S^{i}$, $T^{i}$).

Similar to the SLP to 2RPQ part, we want prove the following two points to show that we can continue the proof of the lemma while only considering set-databases:
\begin{enumerate}
    \item if $D$ is $\epsilon$-far from $D \not \models q_\Pi$, then $D'$ is $\epsilon$-far from $D' \models q_\Pi'$;
    \item if there are more than $\frac{\epsilon}{6}\cdot |D|$ fact-disjoint realizations of $q_\Pi'$ in $D'$ then there are also at least that many of $q_\Pi$ in $D$.
\end{enumerate}

There is an obvious way to translate a set of $m$ fact-disjoint realizations of $q_\Pi$ in $D$ to $m$ fact-disjoint ones of $q_\Pi'$ in $D'$ in vice versa:
In the one direction, we add multiplicity indexes to the relation symbols (and fact-disjointness implies that these relation symbols exist in $D'$) and in the other direction we drop the multiplicity indexes (and the definition of $D'$ implies that if there were e.g. $k$ facts with $P^1(a,b), \dots, P^k(a,b)$ in the realizations, then $D(P(a,b)) \geq k$).

For Point~(2), we use the translation from realizations of $q_\Pi'$ in $D'$ to realizations of $q_\Pi$ in $D$ outlined above.

For Point~(1), assume by contradiction that $D'$ is only $\epsilon'$-far from $D' \models q_\Pi'$ for some $\epsilon' < \epsilon$.
Let $G' \subseteq D'$ be a database witnessing this, that is, $|G'| = \epsilon' \cdot |D'|$ and $D' \setminus G' \not \models q_\Pi'$.
We can translate the facts in $G'$ as outlined above to obtain a database $G$ over the schema $\schemaS$. Since $D$ is $\epsilon$-far and $\epsilon > \epsilon'$, there is thus a realization of $q_\Pi$ in $D \setminus G$. 
It is straightforward to see that we can translate this back to a realization of $q_\Pi'$ in $D' \setminus G'$.

\medskip
\textbf{Main part.}
\newcommand{\gdap}{\ensuremath{G_{D' \times q_\Pi'}}\xspace}
Now we continue with the main part of proving the lemma which now only has to be done for the set-database $D'$ (so we will use set-semantics, which basically just implies that all facts in the database have multiplicity exactly one) and the 2RPQ $q_\Pi'$.
The general idea is outlined in the main part of the paper and can be summarized by these four points which can be read as a chain of implications, with Point~(1) being the initial assumption of the lemma (after our reductions to 2RPQs and multiplicity one) and Point~(4) the desired result:
\begin{enumerate}
  \item $D'$ is $\varepsilon$-far from $q'_\Pi$ being false $D'$;
  \item the minimum edge cut separating $V_2$ from $V_1$ in $\gdap$ has cardinality exceeding $\varepsilon \cdot |D'|$;
  \item there are more than $\varepsilon \cdot |D'|$ edge-disjoint $V_1$-$V_2$ paths in $\gdap$;
  \item there are more than $\frac{\varepsilon}{6} \cdot |D'|$ fact-disjoint realizations of $q'_\Pi$ in $D'$.
\end{enumerate}

A few notes. 
The product graph $\gdap = (V, E, L)$ is defined as in Section~\ref{sect:RPQs}. Since all multiplicities in $D'$ are one, each pair $(\alpha, \delta) \in D' \times \Delta'$ gives rise to exactly one edge, so we may drop the multiplicity index and view edges as pairs $\langle \alpha, \delta \rangle$. We further set $V_1 = \{\langle a, s_\epsilon \rangle \mid a \in \adom(D')\}$ and $V_2 = \{\langle a, s_T \rangle \mid a \in \adom(D')\}$, and define edge cuts separating $V_2$ from $V_1$ as before.

\medskip
  ``$1 \Rightarrow 2$'' 
This implication is closely related to the proof of Lemma~\ref{lem:efartranslates}.
  Assume for a contradiction that the minimum edge cut $E_\mn{cut} \subseteq E$ that separates $V_1$ and $V_2$ in \gda contains no more than $\epsilon \cdot |D|$ edges.
  Recall that every edge in $E$ derives from a unique fact in $D$. 
  Let $F_{\mn{cut}} \subseteq D'$ be the set of facts in $D'$ that $E_\mn{cut}$ is derived from. 
  Since $|F_{\mn{cut}}| \leq \epsilon \cdot |D'|$ and $D'$ is $\epsilon$-far from making $q_\Pi'$ false there is still a realization of $q_\Pi'$ in $D' \setminus F_{\mn{cut}}$. This realization takes the form of a path, and it is straightforward to construct from it a path in $\gdap$ that starts at some $v \in V_1$ and ends at some $v' \in V_2$. Since the original path uses no edge from $F_{\mn{cut}}$, the constructed path uses no edge in $E_{\mn{cut}}$.
  This contradicts $E_{\mn{cut}}$ being a cut and thus we are done.

    \medskip
  ``$2 \Rightarrow 3$'' 
  This implication immediately follows from the famous theorem of Menger \cite{Menger}. We use a slight variation of it,
proved e.g.\ as Theorem 1.1.18 in \cite{GrinbergMath530Lec27}.
\begin{theorem}
  \label{thm:menger}
  Let $G$ be a multi-digraph and $S,T \subseteq V(G)$ disjoint.  If every minimum
  $(S,T)$-edge-cut in $G$ has size at least $k$, then $G$ contains 
	at least $k$ edge-disjoint paths from $S$ to $T$.
\end{theorem}
  
  \medskip

``$3 \Rightarrow 4$''. 
For the proof of this implication we are using the specific structure of $q_\Pi'$ to argue that when converting the edge-disjoint paths in $\gdap$ back to paths in $D'$, then there are still a lot (though less than before) of non-overlapping paths.

  We first note that for each  fact in $D'$ there can be multiple edges in the product graph that derive from it.
  For instance, the two edges pointing to $\langle c, s_S\rangle$ in Figure~\ref{fig:DtoG} both derive from the fact $S(b,c)$ in $D'$. 
  This makes the construction of fact-disjoint realizations in $D'$ non-trivial. In particular, we cannot simply transition
  from paths in $\gdap$ to the facts in 
  $D'$ that they derive from because the 
  resulting realizations may overlap even 
  if the original paths do not.
  For our simple 2RPQ $q_\Pi'$, however, this can  only happen for $S$-facts in $D'$, since $S$ is the only edge that appears more than once in $q_\Pi'$.
  Let us  analyze more precisely how edge-disjoint paths in $\gdap$ might lead to overlapping realizations in $D'$.

  For a path $p = e_1 \cdots e_{n-1}$ in $\gdap$, we use $o(p)$ to denote the \emph{origin} of $p$, that is, the path
  $a_1 P_1 a_2 \cdots P_{n-1} a_n$ in $D'$ with $P_i(a_{i-1}, a_i)$ the fact in $D'$ that edge $e_i$ derived from.
  We say that paths $p_1, p_2$ in $\gdap$ \emph{origin-overlap} if $o(p_1)$ and $o(p_2)$ share at least one fact.
  For example, if $p_1$ and $p_2$ are the two edge-disjoint paths from $V_1$ to $V_2$ in Figure~\ref{fig:DtoG}, then
  $o(p_1)$ and $o(p_2)$ share the fact $S(b,c)$ and thus $p_1$ and $p_2$ origin-overlap.
  Let us analyze how $V_1$-$V_2$-paths can origin-overlap in general.
  We distinguish two types of edges in $\gdap$ that derive from an $S$-fact:
  \emph{initial $S$-edges} connect a
 $\langle *, s_R\rangle$ vertex to a $\langle *, s_S\rangle$ vertex and
 \emph{recurrent $S$-edges} connect
 a
 $\langle *, s_S\rangle$ vertex to a $\langle *, s_S\rangle$ vertex. Observe that by
definition of $q_\Pi'$ and $\gdap$, every
$V_1$-$V_2$-path in $\gdap$ contains exactly
one initial $S$-edge. The following is a consequence of the definition of $q_\Pi'$ and $\gdap$:
 \begin{itemize}
     \item[$(o)$] Each $S$-fact in $D'$ gives rise to exactly two edges in $\gdap$, one initial $S$-edge and  one recurrent $S$-edge.
 \end{itemize}
Next, we prove three claims but let us first give an intuition of how they fit together.
Claim~3 is the combinatorial heart of the argument: in any set $P$ of $m$ edge-disjoint $V_1$-$V_2$ paths, at most half of the paths origin-overlap with three or more other paths. 
The remaining paths overlap with at most two others each, which is
what makes the final greedy selection lose only a factor of~$6$: a factor~$2$ for restricting to the good half, and a factor~$3$ for discarding a chosen path together with its at most two overlapping partners.

The reason such a bound can be expected is an asymmetry between the two kinds of $S$-edges. 
Claim~1 says that the origin of two edge-disjoint paths in $\gdap$ can share a fact only via the initial $S$-edge on one of them and the recurrent $S$-edge on the other. 
In particular, overlaps caused by $R$- and $T$-facts are impossible, and so are overlaps between two initial or two recurrent $S$-edges. 
Since every $V_1$-$V_2$ path contains \emph{exactly one} initial $S$-edge, a path $p$ can be involved in at most one overlap through its own initial $S$-edge: 
at most one other path uses the recurrent $S$-edge deriving from the same fact. 
Every further overlap of $p$ must therefore be of the opposite kind, that is, some other path $p'$ has its unique initial $S$-edge deriving from a fact whose recurrent $S$-edge lies on $p$. Consequently, if $p$ origin-overlaps with at least three paths, at least two of them, say $p_1$ and $p_2$, are of this second
kind, and we may associate the pair $\{p_1,p_2\}$ with $p$.

Claim~2 is exactly what guarantees that different paths are associated with disjoint pairs: distinct paths carry distinct recurrent $S$-edges, these derive from distinct facts, and hence the paths whose initial $S$-edge matches one of them are distinct too. 
So if more than $m/2$ paths had three or more overlaps, the associated pairs would already exhaust more than $m$ distinct paths, which is impossible as $|P| = m$.

  \begin{cclaim}{1} Let $p_1,p_2$ be edge-disjoint $V_1$-$V_2$ paths. 
      If an edge $e$ 
      on $p_1$ and an edge $e'$ 
      on $p_2$
      have the same origin, then 
      $e$ is an initial $S$-edge and
      $e'$ a recurrent $S$-edge, or vice versa. 
\end{cclaim}
\begin{cproof}{1}
By definition fo $q_\Pi'$, the $R$- and $T$-facts occur only once each, so every $R$-fact and $T$-fact give rise to exactly one edge in $\gdap$.
  This implies that if $e$ and $e'$ have the same origin, then it can only be an $S$-fact. Observation~$(o)$ and the
  edge-disjointness of $p_1$ and $p_2$ imply that $e$ is an initial $S$-edge and $e'$ an recurrent $S$-edge or vice versa.
\end{cproof}
\begin{cclaim}{2}
    Let $p_1,p_2,p_1',p_2'$ be  $V_1$-$V_2$ paths, with $p_1,p_2$ edge-disjoint, and let the initial $S$-edge on $p_i'$ derive from the same fact as some edge on $p_i$, for $i \in \{1,2\}$.
    Then $p_1 \neq p_2$ implies $p_1' \neq p_2'$.
\end{cclaim}
\begin{cproof}{2}
Assume that $p_1 \neq p_2$. Let $e_{i,1}$ be the initial $S$-edge on $p_1$
   and $e_{i,2}$  the initial $S$-edge on~$p_2$. Since the initial $S$-edge on $p'_1$ derives from the
  same fact as some edge on~$p_2$, by Observation~(o) the latter edge must be a
  recurrent $S$-edge, and likewise for
  $p'_2$ and $p_2$. Let these recurrent $S$-edges be $e_{r,1}$ and $e_{r,2}$, respectively. Since $p_1 \neq p_2$ and
  $p_1, p_2$ are edge-disjoint, we must 
  have  $e_{r,1} \neq e_{r,2}$. By Observation~(o), the origins of $e_{r,1}$
  and $e_{r,2}$ are thus distinct. Another application of Observation~(o) yields that,
  then, $e_{i,1}$ and $e_{i,2}$ are also distinct. Consequently, $p'_1 \neq p'_2$.
\end{cproof}
We now prove the bound announced above, using Claim~1 to locate the overlaps and
Claim~2 to keep the associated pairs disjoint.
\begin{cclaim}{3}
    Let $P$ be a set of $m \geq 1$ edge-disjoint $V_1$-$V_2$ paths in $\gdap$.
  Then there are at most $\frac{m}{2}$ 
  paths in $P$ that  origin-overlap with at least three other paths in $P$.
\end{cclaim}
\begin{cproof}{3}
 Assume to the contrary of what we have to show that  $P$ contains more than $\frac{m}{2}$ paths that origin-overlap with at least three other paths in $P$. Let  
  $P' \subseteq P$ 
  be a set of $\frac{m}{2}+1$ such paths.
  We show how to choose, for each $p \in P'$, a set
  $\gamma(p) \subseteq P$ of cardinality  two  such that
  \begin{itemize}
      \item[$(*)$] $\gamma(p) \cap \gamma(p') = \emptyset$, for all distinct $p,p' \in P'$.
  \end{itemize}
  Then $\bigcup_{p \in P'} \gamma(p)$ has cardinality at
  least $m+2$, in contradiction to $|P| = m$.

  Let $p \in P'$. By assumption, $p$ origin-overlaps with at least three distinct other paths $p_1, p_2, p_3 \in P$.
  Then for all $i \in \{1,2,3\}$ there are edges $e_i$ on
  $p_i$ and $e'_i$ on $p$ that derive from the same fact.
  Since $p,p_1,p_2,p_3$ are all from $P$ and thus edge-disjoint, Claim~1 yields that $e_i$ is an initial $S$-edge and $e_i'$ a
  recurrent $S$-edge that derive from the same fact,
  or vice versa.
  Observation~(o) then implies that $e_1',e'_2,e'_3$ are pairwise distinct.
  Since there is only one initial $S$-edge on $p$, we are left with the following two cases (up to swapping $e_1,e_2,e_3$):
  \begin{enumerate}
    \item $e_1$ and $e_2$ are initial $S$-edges
    and $e_3$ is a recurrent $S$-edge
    \item $e_1,e_2,e_3$ are initial $S$-edges.
  \end{enumerate} 
   Set $\gamma(p)= \{ p_1, p_2 \}$.
    It remains to prove~$(*)$. Let $p,p' \in P$ be distinct,
    $\gamma(p)=\{p_1,p_2\}$, and $\gamma(p')=\{p_1',p_2'\}$.
    Then the initial $S$-edge of $p_i$ derives
    from the same fact as a recurrent $S$-edge on $p$
    for $i \in \{1,2\}$, and likewise for $p_i'$ and $p'$.
    Since $p,p'$ are distinct and thus edge-disjoint,  Claim~2 implies that $\{p_1,p_2\} \cap \{p'_1,p'_2\} = \emptyset$.
\end{cproof}~
\medskip

  We now return to proving the ``$2 \Rightarrow 3$'' implication. Let $m > \epsilon \cdot |D|$ be the maximum number of edge-disjoint $V_1$-$V_2$ paths in $\gdap$ and let $P$ be a set that contains $m$ edge-disjoint $V_1$-$V_2$-paths.
  By Claim~3, we then obtain a set $P' \subseteq P$ of cardinality at least $\frac{m}{2}$ such that each path $p
   \in P'$ origin-overlaps with at most
   two other paths in $P'$.
  Hence we may obtain the required $\frac{\epsilon \cdot |D|}{6}$ fact-disjoint realizations of $q_\Pi'$ in $D'$
  by  repeatedly choosing a path $p \in P'$ and eliminating the (at most $2$) paths from $P'$ that origin-overlap with $p$. 
\end{proof}

\lemimplecrpqtorealizations*
\begin{proof}
  Since $\Pi$ is an SBP it looks as follows:
$$
\Pi = \{\mn{goal}() \leftarrow P(x) \wedge  \bigwedge_{i=1}^k R_i(w_i,x), \quad
P(x) \leftarrow S(x,y) \wedge P(y),\quad
P(x) \leftarrow S(x,y) \wedge \bigwedge_{j=1}^\ell T_j(y,z_j)\}.
$$

Before we go into the proof a short notational reminder about multiplicities. When we speak about `adding' or `removing' a fact from a database we mean increasing or decreasing its multiplicity. 
  
For $m \in \mathbb{N}$, we say that a
database $D'$ is \emph{$m$-far from $D' \not \models \Pi$} if we need to remove more than $m$ facts from $D'$ to make it a non-answer for $\Pi$, and $m$ is maximal with this property. 
Thus  $D$ is  $m$-far from $D \not \models \Pi$, 
for some $m \geq \epsilon \cdot |D|$.

Let $D'$  be a database obtained from $D$ by exhaustively
choosing and removing facts, preserving
the property of $m$-farness from $D' \not \models \Pi$. As a result, (i)~$D' \subseteq D$, (ii)~$D'$ is $m$-far from $D' \not \models \Pi$, and (iii)~removing any fact from $D'$ results in a database that is $m-1$-far from $D' \not \models \Pi$.

Let $\Pi$ be as in the definition of SBPs above and let $\Pi'$ be the 
following SLP:
%
%
$$
\Pi' = \{\mn{goal}() \leftarrow P(x) \wedge  R_1(w,x),\quad 
P(x) \leftarrow S(x,y) \wedge P(y),\quad
P(x) \leftarrow S(x,y) \wedge T_1(y,z)\}.
$$

It is straightforward to see that $D'$ being $m$-far from $D' \not \models \Pi$ implies $D'$ also being $m$-far from $D' \not \models \Pi'$.
By Lemma~\ref{lem:farness_realizations}, there are more than $\frac{m}{6}$ fact-disjoint realizations of $\Pi'$ in~$D'$.
Let $\Rmf$ be the set of all such realizations. 
Note that each realization in \Rmf is an $R_1S^+T_1$-path.
We complete the realizations in \Rmf to fact-disjoint realizations of $\Pi$ in $D'$, that is, we add to
the second element on each such path an 
incoming $R_i$-edge for $1 < i \leq k$ and to each second last element an outgoing $T_j$-edge for $1 < j \leq \ell$, in a way such that fact-disjointness is preserved.
This clearly proves the lemma.

We only consider the $R_i$-edges, as the treatment of the $T_j$-edges is fully symmetric.
We also assume that the $R_i \in \Sbf$, rather than $R_i$
being of the form $R^-$. Dealing with the latter case is
again fully symmetric, and in fact only amounts to swapping arguments in $R_i$-facts.

We iterate over all $c \in \mn{adom}(D')$ that occur in
some realization in \Rmf as the second element. Take any such $c$, let $\Rmf_c$
be the set of all realizations in $\Rmf$ that use $c$
as the second element, and let $n$ be the cardinality of $\Rmf_c$. 
It suffices to show that $D'$ contains at least $n$ facts of the form $R_i(*,c)$ for all $1 < i \leq k$, as we can then add a different such fact to each of the $n$ realizations in $\Rmf_c$. Moreover, the $R_i(*,c)$-facts added for different $c$ are clearly distinct anyway.

Fix an $\ell$ with $1 < \ell \leq k$. 
We now prove that there are sufficient $R_\ell(*,c)$ facts, and this argument obviously transfers to all other indices between $2$ and $k$.
Assume to the contrary of what we want to show that $D'$ contains strictly less
than $n$ facts of the form $R_\ell(*,c)$. Let $D''$ be
the database obtained from $D'$ by choosing some $R_1(b,c) \in D'$ and setting $D''((R_1(b,c)) \coloneqq D'(R_1(b,c)) -1$.
Due to Property~(iii) of $D'$, $D''$ is $m-1$-far from $D'' \not \models \Pi$. Consequently, there is a database $F$ with $F \subseteq  D''$ of size $|F| = m$ such that $D'' \setminus F \not \models \Pi$. 

{\bf Claim.} 
$F$ contains all facts from $D''$ that are of the form $R_1(*,c)$.
%
\\[1mm]
To prove the claim, assume to the contrary that $D''(R_1(b',c)) > F(R_1(b',c))$ for some fact $R_1(b', c) \in D'$.  We know that $D'' \setminus F \not \models \Pi$.
But then $D' \setminus F \not \models \Pi$ because
  any realization $E$ of $\Pi$ in $D' \setminus F$ gives rise
to a realization of $\Pi$ in $D'' \setminus F$:
\begin{itemize}

    \item if $E$ does not contain the fact $R_1(b,c)$ that was
    deleted during the construction of $D''$, then
    $E$ is a realization of $\Pi$ in $D'' \setminus F$;

    \item otherwise, we can construct a realization of
    $\Pi$ in $D'' \setminus F$ by taking $E$ and replacing the fact $R_1(b,c)$ with $R_1(b',c)$.
        
\end{itemize}
But $D' \setminus F \not\models \Pi$ contradicts
Property~(ii) of $D'$. This finishes the proof of 
the claim.

\smallskip

We now know that $F$ contains all facts from $D''$ that are
of the form $R_1(*,c)$. The combined multiplicities of those $R_1(*,c)$ facts is $n-1$ since each of the $n$ fact-disjoint realizations in $\Rmf_c$ identifies one such fact in $D'$ and during the construction of $D''$ only one fact was removed. Let $F'$ be the database obtained from $F$ by removing all facts of the form $R_1(*,c)$ and instead adding all facts from $D'$ that are of the form $R_\ell(*,c)$. Recall that we had assumed
that there are at most $n-1$ facts of the latter kind. Thus,
$|F| = m$ implies $|F'| = m$. We argue that $D' \setminus F' \not\models \Pi$, in contradiction to $D'$ being
$m$-far from $D' \not\models \Pi$. 

Assume to the contrary that
$D' \setminus F' \models \Pi$. Then there is a derivation $\Gamma = (V, E, \ell, \rho,  h)$ of $\Pi$ in $D' \setminus F'$.
We  show that $\Gamma$ is also a derivation  of $\Pi$ in $D'' \setminus F$, implying $D'' \setminus F \models \Pi$,
which is a contradiction.  Consider 
the homomorphic image of $h$ over all nodes $v \in V$ in $D'$. It cannot contain a fact of the form $R_1(*,c)$:
since $R_1 \neq S$ and $R_1$ does not occur in any of
the $T_j$, the only atom in all the rule bodies that $h$ could map to 
such a fact is $R_1(w_1,x)$ in the root node. But then $h$ must map the fact $R_i(w_i,x)$ to a fact in $D' \setminus F'$ of the form $R_i(*,c)$ while by choice of~$F'$, $D' \setminus F'$
contains no fact of this form. However, the only facts
that occur in $D' \setminus F'$, but not in $D'' \setminus F$ are of the form $R_1(*,c)$. Since no such fact is in the
range of $h$, it follows that $h$ is also a homomorphism from all the rule bodies to $D'' \setminus F$ and thus $\Gamma$ is also a derivation  of $\Pi$ in $D'' \setminus F$, as desired. 
\end{proof}

\section{Proofs for Section~\ref{sec:musing}}

\lemmaproblem*

\begin{proof}
We are given a number $n$ and then we construct $D$ from the lemma as depicted in Figure~\ref{fig:brickwall-size-3}. 
Formally $D$ is defined as follows. For $0 \le i, j \le 2n-1$, the database $D$ contains the following facts:
\begin{itemize}
  \item $\mathrm{arc}((i,j),(i+1,j))$ for even $j$ and $\mathrm{arc}((i+1,j),(i,j))$ for odd $j$,
        where $i < 2n-1$;
  \item $\mathrm{arc}((i,j),(i,j-1))$ for $j \ge 1$ and $i \equiv j+1 \pmod 2$;
  \item $\mathrm{arc}((0,j),a)$ for even $j$, $\mathrm{arc}((2n-1,j),b)$ for odd $j$, and $\mathrm{arc}((i,0),c)$ for odd $i$.
\end{itemize}

Intuitively $D$ then contains $a,b,c$ and a $2n \times 2n$ grid satisfying these properties:
\begin{itemize}
    \item the orientation of the horizontal rows alternates;
    \item the vertical edges are always from top to bottom, there are $n$ per row, and they are spaced such that if you want to traverse two consecutive rows downward you have to take at least one horizontal edge in between the two vertical edges;
    \item each vertex in the grid has degree at most $3$.
\end{itemize}

For Point~(2), let $D_1, D_2$ be realizations of $\Pi^2_{\{a,b,c\}}$ in $D$ and assume $D_1 \cap D_2 = \emptyset$. 
By minimality, $D_1$ viewed as an undirected graph is a tree whose leaves are exactly $a$, $b$, and $c$. 
We observe that a path in $D_2$ cannot pass through a vertex of $D_1$ other than $a,b,c$: such a vertex already has two edges in $D_1$ and since each vertex in the grid has degree at most three, at most one edge is left for $D_2$. 
Since $D$ is a planar graph where $a$, $b$, and $c$ are only connected to distinct sides of the subcubic grid in $D$ (see Figure~\ref{fig:brickwall-size-3}), $D_1$ splits the drawing into three regions, each with only two of the three terminals on its boundary. 
If the root of $D_2$ is in such a region, one path from it has to cross $D_1$ to reach the third terminal, which is not possible by the above observation. If instead the root of $D_2$ lies on $D_1$, it has at most one edge left, which leads into a single region and we can apply the same argument. 
Lastly the root cannot be $a,b$, or $c$ since they don't have outgoing edges.
Hence $D_2$ cannot reach all terminals, contradiction the assumption and proving Point~(2).

For Point~(1), assume that we remove a set of edges $E$ with $|E| < n$. 
We divide the $2n$ rows of the grid into $n$ \emph{strips}: for even $k$ with $0 \leq k \leq 2n-2$, the strip of rows $k$ and $k+1$ consists of the edges in these two rows, the vertical edges between them, and the edges from these rows to $a$ and $b$.
Note that no two strips share any edge. 
Since $|E| < n$, there is a strip that is untouched by $E$.
Let $r$ be the upper right corner of this strip. 
Then $r$ reaches $a$ and $b$, since the connections are part of the strip.
What remains to be shown is that $r$ also reaches $c$.
Note that from $r$, within its strip, we can reach the starting points of all $n$ edges that leave the strip downwards (for the bottom strip, these are the edges to $c$).
From each of those elements there is a path to $c$ that zigzags downwards (see the paths to $c$ of the realizations in Figure~\ref{fig:brickwall-size-3} as an example of this). 
Since each path only uses the column of its starting point and the column to its left, and the starting points lie in every other column, the paths are pairwise disjoint.
This means one of them avoids $E$, proving that there is a path from $r$ to $c$. 
Thus $D \setminus E \models \Pi^2_{\{a,b,c\}}$, as required.

\begin{figure}[t]
    \centering
    \begin{tikzpicture}
        \newlength{\colorThick}
        \setlength{\colorThick}{2.5pt}
        \colorlet{ColFirst}{Red!80}
        \colorlet{ColSecond}{NavyBlue!80}
        \colorlet{ColMix}{Purple!80}
    
        \node[circle, inner sep=2pt, outer sep=0pt, minimum width=14pt] (a) at (-2, 2.4) {$\bf a$};
        \node[circle, inner sep=2pt, outer sep=0pt, minimum width=14pt] (b) at (10, 3.2) {$\bf b$};
        \node[circle, inner sep=2pt, outer sep=0pt, minimum width=14pt] (c) at (4.8, -1.5) {$\bf c$};

        \foreach \i in {0,...,7}
        {
        \foreach \j in {0,...,7}
        {
        \pgfmathsetmacro{\x}{1.2 * \i};
        \pgfmathsetmacro{\y}{.8 * \j};
        \node[circle, fill=black, minimum width = 1.5pt, inner sep = 1pt, outer sep=1.5pt]
        (\i\j) at (\x,\y) {};
        }
        }

        \node[label=below:{$\scriptstyle 0,0$}] () at (00) {};
        \node[label=below:{$\scriptstyle 7,0$}] () at (70) {};
        \node[label=above:{$\scriptstyle 0,7$}] () at (07) {};
        \node[label=above:{$\scriptstyle 7,7$}] () at (77) {};

        \begin{scope}[on background layer]
        
        \node[fill=ColFirst, circle, minimum width=15pt] at (56) {};

        \draw[ColFirst, line width=\colorThick] (56) -- (66);
        \node[fill=ColFirst, circle] at (66) {};
        \draw[ColFirst, line width=\colorThick] (66) -- (76);
        \node[fill=ColFirst, circle] at (76) {};
        \draw[ColFirst, line width=\colorThick] (76) -- (75);
        \node[fill=ColFirst, circle] at (75) {};
        \draw[ColFirst, line width=\colorThick] (75) -- (b);

        \draw[ColFirst, line width=\colorThick] (56) -- (55);
        \node[fill=ColFirst, circle] at (55) {};
        \draw[ColFirst, line width=\colorThick] (55) -- (45);
        \node[fill=ColFirst, circle] at (45) {};
        \draw[ColFirst, line width=\colorThick] (45) -- (44);
        \node[fill=ColFirst, circle] at (44) {};
        \draw[ColFirst, line width=\colorThick] (44) -- (54);
        \node[fill=ColFirst, circle] at (54) {};
        \draw[ColFirst, line width=\colorThick] (54) -- (53);
        \node[fill=ColFirst, circle] at (53) {};
        \draw[ColFirst, line width=\colorThick] (53) -- (43);
        \node[fill=ColFirst, circle] at (43) {};
        \draw[ColFirst, line width=\colorThick] (43) -- (42);
        \draw[ColFirst, line width=\colorThick] (52) -- (51);
        \node[fill=ColFirst, circle] at (51) {};
        \draw[ColFirst, line width=\colorThick] (51) -- (41);
        \node[fill=ColFirst, circle] at (41) {};
        \draw[ColFirst, line width=\colorThick] (41) -- (40);
        \node[fill=ColFirst, circle] at (40) {};
        \draw[ColFirst, line width=\colorThick] (40) -- (50);
        \node[fill=ColFirst, circle] at (50) {};
        \draw[ColFirst, line width=\colorThick] (50) -- (c);

        \draw[ColFirst, line width=\colorThick] (45) -- (35);
        \node[fill=ColFirst, circle] at (35) {};
        \draw[ColFirst, line width=\colorThick] (35) -- (25);
        \node[fill=ColFirst, circle] at (25) {};
        \draw[ColFirst, line width=\colorThick] (25) -- (15);
        \node[fill=ColFirst, circle] at (15) {};
        \draw[ColFirst, line width=\colorThick] (15) -- (05);
        \node[fill=ColFirst, circle] at (05) {};
        \draw[ColFirst, line width=\colorThick] (05) -- (04);
        \node[fill=ColFirst, circle] at (04) {};
        \draw[ColFirst, line width=\colorThick] (04) -- (a);

        \node[fill=ColSecond, circle, minimum width=15pt] at (23) {};

        \draw[ColSecond, line width=\colorThick] (23) -- (13);
        \node[fill=ColSecond, circle] at (13) {};
        \draw[ColSecond, line width=\colorThick] (13) -- (03);
        \node[fill=ColSecond, circle] at (03) {};
        \draw[ColSecond, line width=\colorThick] (03) -- (02);
        \node[fill=ColSecond, circle] at (02) {};
        \draw[ColSecond, line width=\colorThick] (02) -- (a);
        
        \draw[ColSecond, line width=\colorThick] (23) -- (22);
        \node[fill=ColSecond, circle] at (22) {};
        \draw[ColSecond, line width=\colorThick] (22) -- (32);
        \node[fill=ColSecond, circle] at (32) {};
        
        \draw[ColSecond, line width=\colorThick] (32) -- (42);
        \draw[ColSecond, line width=\colorThick] (52) -- (62);
        \node[fill=ColSecond, circle] at (62) {};
        \draw[ColSecond, line width=\colorThick] (62) -- (72);
        \node[fill=ColSecond, circle] at (72) {};
        \draw[ColSecond, line width=\colorThick] (72) -- (71);
        \node[fill=ColSecond, circle] at (71) {};
        \draw[ColSecond, line width=\colorThick] (71) -- (b);
        
        \draw[ColSecond, line width=\colorThick] (32) -- (31);
        \node[fill=ColSecond, circle] at (31) {};
        \draw[ColSecond, line width=\colorThick] (31) -- (21);
        \node[fill=ColSecond, circle] at (21) {};
        \draw[ColSecond, line width=\colorThick] (21) -- (20);
        \node[fill=ColSecond, circle] at (20) {};
        \draw[ColSecond, line width=\colorThick] (20) -- (30);
        \node[fill=ColSecond, circle] at (30) {};
        \draw[ColSecond, line width=\colorThick] (30) -- (c);


        
        


        \fill[pattern={Lines[angle=135, line width=3pt, distance=6pt]}, pattern color=ColFirst, 
        postaction={pattern={Lines[angle=135, line width=3pt, distance=6pt, xshift=4.2pt]}, pattern color=ColSecond}] 
        (a) circle (7.5pt);
        \fill[pattern={Lines[angle=135, line width=3pt, distance=6pt]}, pattern color=ColFirst, 
        postaction={pattern={Lines[angle=135, line width=3pt, distance=6pt, xshift=4.2pt]}, pattern color=ColSecond}] 
        (b) circle (7.5pt);
        \fill[pattern={Lines[angle=135, line width=3pt, distance=6pt]}, pattern color=ColFirst, 
        postaction={pattern={Lines[angle=135, line width=3pt, distance=6pt, xshift=4.2pt]}, pattern color=ColSecond}] 
        (c) circle (7.5pt);
        
        \fill[pattern={Lines[angle=135, line width=3pt, distance=6pt]}, pattern color=ColFirst, 
        postaction={pattern={Lines[angle=135, line width=3pt, distance=6pt, xshift=4.2pt]}, pattern color=ColSecond}] 
        (52) circle (4.5pt);
        
        \fill[pattern={Lines[angle=135, line width=3pt, distance=6pt]}, pattern color=ColFirst, 
        postaction={pattern={Lines[angle=135, line width=3pt, distance=6pt, xshift=4.2pt]}, pattern color=ColSecond}] 
        (52.155) -- (52.205) -- (42.335) -- (42.25) -- cycle;

        \fill[pattern={Lines[angle=135, line width=3pt, distance=6pt]}, pattern color=ColFirst, 
        postaction={pattern={Lines[angle=135, line width=3pt, distance=6pt, xshift=4.2pt]}, pattern color=ColSecond}] 
        (42) circle (4.5pt);

        
        \end{scope}

        \foreach \j in {0,...,3}
        {
        \pgfmathtruncatemacro{\y}{2 * \j};
        \pgfmathtruncatemacro{\Y}{2 * \j + 1};
        \path
        (7\Y) edge [->] (b)
        (0\y) edge [->] (a)
        ;
        \foreach \i in {0,...,6}
        {
        \pgfmathtruncatemacro{\x}{\i + 1};
        \path
        (\i\y) edge [->] (\x\y)
        (\x\Y) edge [->] (\i\Y)
        ;
        }
        }
        
        \foreach \i in {0,...,3}
        {
        \pgfmathtruncatemacro{\x}{2 * \i};
        \pgfmathtruncatemacro{\X}{2 * \i + 1};
        \path
        (\X0) edge [->] (c)
        ;
        \foreach \j in {0,...,2}
        {
        \pgfmathtruncatemacro{\y}{2 * \j};
        \pgfmathtruncatemacro{\Y}{2 * \j + 1};
        \pgfmathtruncatemacro{\J}{2 * \j + 2};
        \path
        (\x\Y) edge [->] (\x\y)
        (\X\J) edge [->] (\X\Y)
        ;
        }
        \path
        (\x7) edge [->] (\x6)
        ;
        }

\end{tikzpicture}
    \caption{Database from Lemma~\ref{lemma:no-bound-for-existential-center} for $n = 4$.
    Two realizations of $\Pi^2_{\{ a, b, c \}}$ are highlighted in red and in blue.
    One needs to remove at least $4$ edges (\emph{e.g.}\ those $4$ involving $a$) to break all realizations of $\Pi^2_{\{ a, b, c \}}$.
    Any two realizations share at least one edge.
    Marked coordinates on the corner elements show how to generalize the construction for any $n$.
    }
    \label{fig:brickwall-size-3}
\end{figure}
\end{proof}

We now move towards the proof of Lemma~\ref{lem:problemtwo} and recall that it is interested in the testability of the following MDLog program:
\begin{align*}
    \Pi_S^2 := ~  & \{ \mn{goal}() \leftarrow \bigwedge_{s \in S} \terminal_s(y) \} \cup \bigcup_{s \in S} \{ \terminal_s(s) \leftarrow \mn{true}, ~\terminal_s(x) \leftarrow \pedge(x, y) \land \terminal_s(y)  \}.
\end{align*}
 where $S$ is a finite set of constants.
In words, $\Pi_S^2$ simply asks whether there exists a database element that can reach all the constants of interest.
Already for $|S| = 3$, we know that the approach via removal lemmas is a dead-end, see Lemma~\ref{lemma:no-bound-for-existential-center}.
We adopt a different perspective and essentially show that realizations of $\Pi_S^2$ past a certain size become superfluous for property testing.

For a given limit length $\ell$, we define the UCQ $q_\ell$ as the union of all possible CQs:
\[
q' := \exists y\ \bigwedge_{s \in S} \pedge^{\ell_s}(y, s),
\]
where each $\ell_s$, for $s \in S$, is an integer such that $0 \leq \ell_s \leq \ell$; and $\pedge^{k}(u, v)$ denotes a path of length $k$ from $u$ to $v$.
Although the original testability result from Chen and Yoshida \cite{chen2019testability} does not support constants in the (U)CQ, it is not hard to accommodate them and establish the following:
\begin{lemma}
    Falsity of $q_\ell$ is constant query testable with one-sided error.
\end{lemma}
\begin{proof}
    We extend the original schema of interest $\schemaS$ with one unary predicate $P_s$ per $s \in S$ and denote $\schemaS' := \schemaS \cup \{ P_s \mid s \in S \}$.
    Consider the following UCQ $q'_\ell$ as the union of all possible CQs:
\[
q' := \exists y, y_{s_1} \dots y_{s_{|S|}}\ \bigwedge_{s \in S} \left( \pedge^{\ell_s}(y, y_s) \land P_s(y_s) \right),
\]
where each $\ell_s$, for $s \in S$, is an integer such that $0 \leq \ell_s \leq \ell$, and $s_1, \dots, s_{|S|}$ are the elements of $S$.
The homomorphism core of every such $q'$ is $\alpha$-acyclic, and no other CQ from $q'_\ell$ is contained in $q'$.
Therefore, by \cite{chen2019testability}, falsity of $q'_\ell$ is constant query testable with one-sided error.

Given $\varepsilon \in (0, 1)$, we now build an $\varepsilon$-tester $T$ for $q_\ell$ based on an $\frac{\varepsilon}{|S| + 1}$-tester $T'$ for $q'_\ell$.
On an input $\schemaS$-database $D$, we simply run $T'$ for $q'_\ell$ on the $\schemaS'$-database $D$ enriched by fact $P_s(s)$ with multiplicity $|D|$ for each $s \in S$.
We denote $D'$ this extended database and observe that $|D'| = (|S| + 1)|D|$.
Notice that $T$ does not need to materialize $D'$: whenever $T'$ asks a completion query for some $P_s$, $T$ returns $s$, whenever $T'$ asks a size query for some $P_s$, $T$ returns $|D|$, which can be computed using $|\schemaS|$ size queries on the original database $D$.
The tester $T$ then accepts iff $T'$ does.
It remains to prove the correctness of $T$.

If $D \not\models q_\ell$, then it is immediate that $D' \not\models q'_\ell$ and $T'$ accepts as it is one-sided error. Therefore $T$ accepts as desired.

If $D$ is $\varepsilon$-far from $D \not\models q_\ell$, we claim that $D'$ is $\frac{\varepsilon}{|S| + 1}$-far from $D' \not\models q'_\ell$, which would provide the desired probabilistic guarantee and conclude the proof.
By contrapositive, assume that $D'$ is not $\frac{\varepsilon}{|S| + 1}$-far from $D' \not\models q'_\ell$, that is we can make $q'_\ell$ false by removing a set $E' \subseteq D'$ of less than $\frac{\varepsilon}{|S| + 1}|D'|$ facts from $D'$.
Recalling that $|D'| = (|S| + 1)|D|$, we have $|E'| < \varepsilon|D|$.
In particular, because the multiplicity of each $P_s(s)$ is $|D|$, we have $D' \setminus E' \models P_s(s)$.
Therefore, the removal of $E := E' \cap D$ is also sufficient to falsify $q'_\ell$ in $D'$, and it is then immediate that $D \setminus E \not\models q_\ell$.
Since $|E| \leq |E'| < \varepsilon|D|$, we obtain that $D$ is not $\varepsilon$-far from $D \not\models q_\ell$, as desired.
\end{proof}

We now leverage testability of all $q_\ell$ to establish the following:

\lemmaproblemtwo*

\begin{proof}
    Let $\varepsilon \in (0, 1)$, and set $\ell := \left\lceil\frac{2|S|}{\varepsilon}\right\rceil$.
    We let $\varepsilon' := \frac{\varepsilon}{2}$ and claim that the $\varepsilon'$-tester for $q_\ell$ is a $\varepsilon$-tester for $\Pi_S^2$.
    Let $D$ be a database.
    \smallskip

    If $D$ falsifies $\Pi_S^2$, then it is clear that $D \not\models q_\ell$ since $q_\ell$ is contained in $\Pi_S^2$.
    Since the $\varepsilon'$-tester for $q_\ell$ is one-sided error, it accepts $D$, as desired.
    \smallskip

    If $D$ is $\varepsilon$-far from falsifying $\Pi_S^2$, we have to show that $D$ is $\varepsilon'$-far from $D \not\models q_\ell$.
    It remains to establish the latter claim.
    We use contraposition: Assuming that $D$ is not $\varepsilon'$-far from $D \not\models q_\ell$ ,we want to prove that $D$ is not $\varepsilon$-far from falsifying $\Pi_S^2$.
    Since $D$ is not $\varepsilon'$-far from $D \not\models q_\ell$, there exists $E \subseteq D$ such that $D \setminus E \not\models q_\ell$ and $|E| \leq \varepsilon'|D|$, that is $|E| \leq \frac{\varepsilon}{2}|D|$.
    We now prove that by further removing at most $\frac{\varepsilon}{2}|D|$ additional facts from $D' := D \setminus E$, we can make $\Pi_S^2$ false.
    Clearly, 
    the resulting database can be obtained from $D$ by removing at most $\varepsilon|D|$ facts, as desired.
    
    To find the facts to remove, we analyze the neighborhood of each constant $s \in S$.
    For  $0 \leq i \leq \ell$, we define $N_i(s)$ as the elements of $D'$ that can reach $s$ on a path of length at  most $i$. 
    Since $D'\not\models q_\ell$, we have $\bigcap_{s \in S} N_\ell(s) = \emptyset$.
    In particular, for every realization of $\Pi_S^2$ in $D'$ that maps the variable $y$ from the goal rule to an element $e \in D'$, there exists some $s \in S$ such that $e \notin N_\ell(s)$.
    By definition, there exists a path from $e$ to $s$ in this realization of $\Pi_S^2$, and thus, for $0 \leq i < \ell$ there is at least one fact $\pedge(u, v)$ such that $u \notin N_i(s)$ and $v \in N_i(s)$.
    Notice that $u \in N_{i+1}(s)$.
    In words, the path must transition at least once from $N_{i+1}(s)$ to $N_{i}(s)$.
    For each $s \in S$, we are going to chose some $i_s < \ell$ and remove all facts that transition from $N_{i_s+1}(s)$ to $N_{i_s}(s)$.
    For $s \in S$ and $0 \leq i < \ell$, let $T_i(s) := \{ \pedge(u, v) \in D' \mid u \notin N_i(s) \land v \in N_i(s) \}$ be the set of facts that transition from $N_{i+1}(s)$ to $N_{i}(s)$.
    It is clear that, for a fixed $s$, every fact of $D'$ can only belong to at most one $T_i(s)$.
    Therefore $\sum_{0 \leq i < \ell} |T_i(s)| \leq |D'|$ for every $s \in S$.
    We define $i_s$ as one of the indices $i$ such that $|T_i(s)|$ is smallest.
    In particular, $|T_{i_s}(s)| \leq \frac{|D'|}{\ell}$ as otherwise we would have $\sum_{0 \leq i < \ell} |T_i(s)| > |D'|$.
    We now remove from $D'$ the following set of facts $E' := \bigcup_{s \in S} T_{i_s}(s)$, whose size is bounded by $|S| \times \frac{|D'|}{\ell}$.
    Recalling that $\ell = \frac{2|S|}{\varepsilon}$, we obtain $|E'| \leq \frac{\varepsilon}{2} |D'|$, thus $|E'| \leq \frac{\varepsilon}{2} |D|$ as desired.
\end{proof}

\section{Proof of Theorem~\ref{theorem:THE-candidate}}

The following proof of Theorem~\ref{theorem:THE-candidate} uses Lemmas~\ref{lem:unrav} and \ref{lem:acyclictohardpattern}, stated after Theorem~\ref{theorem:THE-candidate} in the main body of the paper but that are independent from it.

\theoremTHEcandidate*
\begin{proof}
Let $\Pi$ be an MDLog program and $\Pi^*$ its $\alpha$-acyclic approximation.
\medskip

   {\bf Point~1.} Let $D$ be a database.
	Assume $\bar a \in \Pi^*(D)$ and take a derivation $\Gamma = (V, E, \ell, \rho, h)$ of $\bar a$ from $\Pi^*$ in $D$.
	We convert $\Gamma$ into a derivation $\Gamma'$ of $\bar a$ from $\Pi$ in $D$.
    We simply set $\Gamma' := (V, E, \ell, \rho', h')$ where $\rho'$ and $h'$ are defined as follows.
    For each $w \in V$, $\rho(w)$ is a rule $P(\bar u) \leftarrow p(\bar z)$ with $p(\bar z)$ being obtained as a contraction $P(\bar u) \leftarrow q'(\bar v)$ of some rule $P(\bar x) \leftarrow q(\bar y) \in \Pi$.
    We denote $\sigma : \bar y \rightarrow \bar v$ the performed identification of variables by the contraction.
    We let $\rho'(w) := P(\bar x) \leftarrow q(\bar y)$ and $h'(w) := h(w) \circ \sigma$.
    It is immediate that $\Gamma'$ is a derivation of $\bar a$ from $\Pi$ in $D$.
    \medskip

    {\bf Point~2.}
    Let $D$ be an $\alpha$-acyclic database.
	From Point~1 above, we already have $\Pi^*(D) \subseteq \Pi(D)$.
	It remains to prove $\Pi(D) \subseteq \Pi^*(D)$.
	Assume $\bar a \in \Pi(D)$ and consider a derivation $\Gamma := (V, E, \ell, \rho, h)$ of $\bar a$ from $\Pi$ in $D$.
	We convert $\Gamma$ into a derivation $\Gamma'$ of $\bar a$ from $\Pi^*$ in $D$.
    We set $\Gamma' := (V, E, \ell, \rho', h')$ where $\rho'$ and $h'$ are defined as follows.

    For each $w \in V$, $h(w)$ defines a contraction $P(\bar u) \leftarrow p(\bar v)$ of a rule $P(\bar x) \leftarrow q(\bar y) \in \Pi$ by identifying variables of $\bar y$ mapped on the same element by $h(w)$.
    We denote $\sigma : \bar y \rightarrow \bar v$ the performed identification of variables by the contraction.
	By construction, we have an injective homomorphism $g : p(\bar v) \rightarrow D$ such that $g \circ \sigma = h(w)$.
	

	Using the assumption that $D$ is $\alpha$-acyclic, we would like to argue that $P(\bar u) \leftarrow p(\bar v)$ is an $\alpha$-acyclic rule, and thus obtain a derivation of $\bar a$ from $\Pi*$ in $D$ simply by setting $\rho'(w) := P(\bar u) \leftarrow p(\bar v)$ and $h'(w) := g$.
	It is however possible that $P(\bar u) \leftarrow p(\bar v)$ is not $\alpha$-acyclic, despite $p(\bar v)$ being isomorphic to the subset of facts $g(p)$ from the $\alpha$-acyclic database $D$.
	By Lemma~\ref{lem:acyclictohardpattern}, there are two possibilities: either $p(\bar v)$ contains a chordless cycle $c_0, \dots, c_m$ with $m \geq 4$ or $p(\bar v)$ contains a tetra $\{ b_1, \dots, b_n\}$ with $n \geq 3$.

	In the first case, since $D$ is $\alpha$-acyclic, there must be two elements $g(c_i), g(c_j)$, with $0 \leq i < j \leq m - 1$ but $j \neq i+1$, that are neighbors in $D$, that is there exists a fact $r(\bar d) \in D$ with $g(c_i), g(c_j) \in \bar d$.
	We extend $p(\bar v)$ into $p'(\bar z)$ with an atom $r(\bar w)$ where $\bar w$ is the tuple of variables obtained by replacing every occurrence of $g(c_i)$ in $\bar d$ by the variable $c_i$, every occurrence of $g(c_j)$ in $\bar d$ by the variable $c_j$, and every occurrence of another element $d$ in $\bar d$ by a fresh variable $v_{c_i, c_j, d}$.
	Note that we introduce at most $N - 2$ fresh variables and that every freshly introduced variable only appears in the fresh atom of $p'(\bar z)$, and thus cannot appear along a chordless cycle nor within a tetra in $p'(\bar z)$.
	We update $g$ accordingly by setting $g : v_{c_i, c_j, d} \mapsto d$ for each freshly introduced variable $v_{c_i, c_j, d}$.
	
	In the second case, since $D$ is $\alpha$-acyclic, there must be a fact $r(\bar d)$ that covers the tetra, that is $g(b_1), \dots, g(b_n) \in \bar d$.
	We extend $p(\bar v)$ into $p'(\bar z)$ with an atom $r(\bar w)$ where $\bar w$ is the tuple of variables obtained by replacing every occurrence of some $g(b_k)$ in $\bar d$ by the variable $b_k$ and every occurrence of another element $d$ in $\bar d$ by a fresh variable $v_{b_1, \dots, b_n, d}$.
	Note that we introduce at most $N - n$ fresh variables and that every freshly introduced variable only appears in the fresh atom of $p'(\bar z)$, and thus cannot appear along a chordless cycle nor within a tetra in $p'(\bar z)$.
	We update $g$ accordingly by setting $g : v_{b_1, \dots, b_n, d} \mapsto d$ for each freshly introduced variable $v_{b_1, \dots, b_n, d}$.
	
	We iterate this process until the resulting $p'(\bar z)$ is $\alpha$-acyclic to comply with Condition~(i) in the definition of $\Pi^*$
	As freshly introduced variables cannot contribute to chordless cycles nor tetras, this terminates in at most $2^{\sizeof{\bar v}}$ steps:
	we only need to introduce a chord between $2$ variables from $\bar v$ or to cover a $k$-tetra on variables of $\bar v$ for some $k \geq 3$ once to break the pattern.
	The number of steps is thus bounded by the number of subsets of $\bar v$, and we have $\sizeof{\bar v} \leq \sizeof{\bar y}$.
	As we introduced at most $N - 2$ variables at each step, the resulting $\bar z$ complies with Condition~(ii) in the definition of $\Pi^*$.
	We therefore obtain a rule $P(\bar u) \leftarrow p'(\bar z)$ belonging to $\Pi^*$.
	We can now set $\rho'(w) := P(\bar u) \leftarrow p'(\bar z)$ and $h'(w) := g$ in $\Gamma'$.
	We obtain the derivation $\Gamma'$ of $\bar a$ from $\Pi^*$ in $D$.
    \medskip
 
    {\bf Point~3.} 
	If $\Pi$ is equivalent to $\Pi^*$, then it is trivial that $\Pi$ is equivalent to an $\alpha$-acyclic MDLog program as $\Pi^*$ is such a program.
	
	Conversely, assume that $\Pi$ is equivalent to an $\alpha$-acyclic MDLog program $\Pi'$.
	We need to prove that $\Pi$ is equivalent to $\Pi^*$.
	Let $D$ be a database.
	From Point~1 above, we have $\Pi^*(D) \subseteq \Pi(D)$.
	It remains to prove that $\Pi(D) \subseteq \Pi^*(D)$.
	Assume $\bar a \in \Pi(D)$.
	Since $\Pi$ and $\Pi'$ are equivalent, we have $\bar a \in \Pi'(D)$.
	Consider a derivation $\Gamma$ of $\bar a$ from $\Pi'$ in $D$.
    Use Lemma~\ref{lem:unrav} to construct the tree-like database $D'$ corresponding to $\Gamma$ and denote $h : D' \rightarrow D$ the associated homomorphism (Point~2 of Lemma~\ref{lem:unrav}).
	By Point~1 of Lemma~\ref{lem:unrav}, we have $\bar a \in \Pi'(D')$, and thus $\bar a \in \Pi(D')$ as $\Pi$ and $\Pi'$ are equivalent. 
	Furthermore, since $\Pi'$ is $\alpha$-acyclic, the database $D'$ is $\alpha$-acyclic (using Point~3 of Lemma~\ref{lem:unrav} and the tree-like structure of $D'$).
	Therefore, by Point~2 above (we mean Point~2 of Lemma~\ref{theorem:THE-candidate}), we have $\Pi(D') = \Pi^*(D')$ and in particular $\bar a \in \Pi^*(D')$.
	Via the homomorphism $h$, this yields $\bar a \in \Pi^*(D)$ as desired.
\end{proof}

\section{Proof of Theorem~\ref{theorem:one-cycle-is-crucial}}

\theoremOneCycleIsCrucial*



We first emphasize that Theorem~\ref{theorem:one-cycle-is-crucial} holds even under the set semantics of databases: the produced database $D$ only requires multiplicities of facts to be $0$ or $1$.
This is a direct consequence of the definitions of cycles, tetras, and derivations of a Datalog program being independent of higher multiplicities.
When property testing comes into play, higher multiplicities become relevant and are notably used in the proof of Lemma~\ref{lem:datalog-lower-cyclecombined} and thus of Theorem~\ref{theorem:non-testable}.
Indeed, the proof of Lemma~\ref{lem:datalog-lower-cyclecombined} builds upon the results of Chen and Yoshida \cite{chen2019testability} which crucially use bag databases.

We now resume the proof of Theorem~\ref{theorem:one-cycle-is-crucial} and detail the tree steps sketched in the main body of the paper.
We recall that we already obtained $D_0$, a tree-like database containing some hard patterns and with a tuple $\bar a \in \Pi(D_0) \setminus \Pi^*(D_0)$.

We begin with the first step, producing the initial sequence $D_0, \dots, D_M$ of databases.
To do so, we will choose a set of isomorphic bags from $D_k$ and unravel those all together to obtain $D_{k+1}$.
More precisely, the set of unraveled isomorphic bags is defined as those bags in the tree-like database that share a same \textit{type}.
In particular, each $D_k$ along the sequence will be associated with a typing function $\type_k$.
Before explaining how to produce the successive $D_k$ and $\type_k$, we need to clarify how to define $\type_0$.
The following definition prepares this.


\begin{definition}[Types of bags]
	\label{definition:basic-type}
	A \emph{concrete type} is a bag from $D_0$.
	An \emph{abstract type} is anything obtained from the body of a rule in $\Pi$ by:
	\begin{itemize}
		\item[(i).] Dropping all IDB relations; and
		\item[(ii).] Consider a contraction; and
		\item[(iii).] Potentially dropping some variables and the atoms in which those variables were involved.
	\end{itemize}
	We consider abstract types up to isomorphism, so there are only finitely many abstract types.
	A \emph{type} is either a concrete type or an abstract type.
\end{definition}

Notice that it is possible for two concrete types to be isomorphic, or for a concrete and an abstract type to be isomorphic, but we still distinguish those.
We define an order on the set of possible types.
This order will be used in the construction of the sequence of databases $D_0, \dots, D_k$ and will be crucial to prove that our objective can be achieved in finitely many steps.

\begin{definition}
	\label{definition:order-on-types}
	Given two types $t_1$ and $t_2$, we denote $t_1 < t_2$ if the following conditions are satisfied:
	\begin{enumerate}
		\item There exists an homomorphism $h : t_1 \rightarrow t_2$;
		\item $t_1$ and $t_2$ are not isomorphic.
	\end{enumerate}
\end{definition}

Notice that, due to self-join freeness, the homomorphism required by (1) is unique and thus (2) could be rephrased as ``$h$ is not an isomorphism''.

\begin{lemma}
	\label{lemma:order-on-types}
	The relation $<$ defines a strict partial order on the set of all possible types.
\end{lemma}

\begin{proof}
	Condition 2 in the definition of $<$ makes clear that $<$ is irreflexive.
	
	We now prove that $<$ is transitive.
	Assume we have $t_1 < t_2$ and $t_2 < t_3$.
	Condition~1 in the definition of $<$ provides two homomorphisms $h_1 : t_1 \rightarrow t_2$ and $h_2 : t_2 \rightarrow t_3$.
	Let $h : t_1 \rightarrow t_3$ be $h := h_2 \circ h_1$.
	By construction, $h$ is an homomorphism and it remains to prove Condition~2 of $t_1 < t_3$.
	By contradiction, assume that $h$ is an isomorphism.
	Then $h_1$ also is, which contradicts $t_1 < t_2$.
\end{proof}

We now go back to the construction of the sequence of databases and their associated typing functions.
The typing function $\type_0$ on $D_0$ only uses the concrete types, in the trivial way:
	the typing function $\type_0$ in $D_0$ is defined by $\type_0(B) := B$ for every bag $B$ of $D_0$.

Before explaining how to move from a $D_k$ whose bags are typed by a typing function $\type_k$ to the next in the sequence, we need to mention some useful properties of $D_k$ and $\type_k$ that will be maintained throughout the construction of the sequence.

\begin{definition}[$D_0$-compatibility]
\label{def:d0-compatibility}
We say that a database $D$ is $D_0$-compatible if it satisfies the three following conditions: 
\begin{enumerate}
    \item The (domain of the) root bag of $D$ contains $\bar a$;
    \item $\bar a \in \Pi(D)$;
    \item there exists an homomorphism $\delta$ from $D$ to $D_0$ with $\delta(\bar a) = \bar a$.
\end{enumerate}
\end{definition}
Notice that a $D_0$-compatible database $D$ must be non $\alpha$-acyclic as otherwise $\Pi$ and $\Pi^*$ would coincide on $D$ (Points~1 and 2 of Theorem~\ref{theorem:THE-candidate}), thus there would be a derivation of $\bar a$ from $\Pi^*$ in $D$, and, via $\delta$, a derivation of $\bar a$ from $\Pi^*$ in $D_0$ which is a contradiction.
In particular, $D_0$ is $D_0$-compatible and thus is also non $\alpha$-acyclic.


\begin{definition}[Good typing]
	\label{def:good-typing}
	We say that $\type$ is a good typing function on a tree-like database $D$ if every bag $B$ in $D$ is isomorphic to its type $\type(B)$.
\end{definition}

It is immediate that $\type_0$ is a good typing function on $D_0$.
We can always extend a good typing function $\type$ to also type hard patterns in a tree-like database $D$.
Consider a hard pattern $P$ in a bag $B$ of $D$.
Since $\type$ is a good-typing, $B$ and its type $\type(B)$ are isomorphic.
Due to self-join freeness, there exists a unique isomorphism $h_B : B \rightarrow \type(B)$, and we set $\type(P) := (\type(B), h_B(P))$.
Given a type $t$, the notation $\type^{-1}(t)$ refers, as is standard, to the set of bags (or chordless cycles, or tetras) in $D$ that have this type.


\begin{definition}[Halting Condition]
\label{def:haltingcondition}
	We say that a $D_0$-compatible tree-like database $D$ and a good typing function $\type$ on $D$ satisfy the halting condition if there exists a type $t$ for chordless cycles or for tetras such that $\type^{-1}(t)$ is crucial for $\bar a \in \Pi(D)$.
\end{definition}

\subsection*{Obtaining a set of hard patterns}

Having defined $D_0$ and $\type_0$, we now explain how to construct $D_{k+1}$ and $\type_{k+1}$ from $D_k$ and $\type_k$ if the latter do not satisfy the halting condition.
To simplify notation, we'll refer to $D_k$ as $D$, to $\type_k$ as $\type$ and to the constructed $D_{k+1}$ as $D'$ and $\type_{k+1}$ as $\type'$.

Assume that we have a $D_0$-compatible tree-like database $D$ and a good typing function $\type$ on $D$, such that they do not satisfy the halting condition.
Since $D$ is not $\alpha$-acyclic, there exists at least one type of bag $s$ such that $\type^{-1}(s)$ is non-empty and contains a chordless cycle or a tetra (Lemma~\ref{lem:acyclictohardpattern}).
Among these types, we pick one that is maximal for $<$ and denote it $B_0$\footnote{This maximality assumption is not essential but greatly simplifies the further argument showing that the halting condition is reached in finitely-many steps; see Lemma~\ref{lem:halting}}
By definition, this type of bag $B_0$ contains a hard pattern, and we let $t := (B_0, P_0)$ be the type of this hard pattern. 
Due to the halting condition not being satisfied, the set $\type^{-1}(t)$ of hard patterns is non-crucial for $\bar a \in \Pi(D)$.
In particular, there exists a derivation $\Gamma := (V_\Gamma, E_\Gamma, \ell_\Gamma, \rho_\Gamma, h_\Gamma)$ of $\bar a$ from $\Pi$ in $D$ that does not use any hard pattern from $\type^{-1}(t)$.
We now construct $D'$ as a tree-like database, along with good typing function $\type'$ on $D'$ and a homomorphism $h'$ from $D'$ to $D$.

\paragraph*{Root bag of $D'$.}
We denote $B$ the root bag of $D$.
\begin{itemize}
    \item If $\type(B) \neq B_0$, we let $B$ be the root bag of $D'$ in which case we also set $h'|_B := \mathsf{Id}$ and $\type'(B) := \type(B)$.
    \item Otherwise we have $\type(B) = t$, and we take a closer look at the derivation $\Gamma$.
We denote $v_0$ the root of $V_\Gamma$, we let $q_0$ be the body of the goal rule $\rho_\Gamma(v_0)$ and $h_0 := h_\Gamma(v_0)$.
We denote $p$ the subquery of $q_0$ obtained by dropping all IDB relations and now view $h_0$ as an homomorphism from $p$ to $D$.
Consider the induced CQ ${p}_{h_0, B}$ and notice that it defines an abstract type.
Additionally, since $\bar a \in \adom(B)$ by $D_0$-compatibility of $D$, ${p}_{h_0, B}$ contains all the variables $x_1, \dots, x_m$ of $q_0$ that occur in the head of the goal rule $\rho_\Gamma(v_0)$, potentially identified according to $h_0$.
We define the root bag $B'$ of $D'$ as an isomorphic copy of $p_{h_0, B}$ in which the (potentially identified) variables $x_1, \dots, x_m$ have been renamed into $h_0(x_1), \dots, h_0(x_m)$ (note that, by definition, two identified variables are renamed as the same element from $\bar a$). 
In particular, $\bar a \subseteq \adom(B')$.
We let $\type'(B') := p_{h_0, B}$ and $h'|_{B'} := p|_{\mn{var}(p_{h_0, B})}$.
\end{itemize}


\paragraph*{Further bags of $D'$.}
Assume that we constructed a bag $B'$ in $D'$ and that $h'$ and $\type'$ have already been defined on $B'$.
For each element $e \in \adom(B')$ and each bag $B$ of $D$ 
such that
$B'$ does not already have a successor bag $B''$ with $e \in \adom(B') \cap \adom(B'')$ and $h'(B'') \subseteq B$, we distinguish two cases:
\begin{itemize}
	
	\item 
    If $\type(B) \neq B_0$, then we introduce a fresh copy $B''$ of $B$ as a successor bag of $B'$.
	If $h'(e)$ occurs in $B$, then we additionally connect $B'$ and $B''$ by identifying elements $e$ from $B'$ and the fresh copy of $h'(e)$ from $B''$. 
	We set $\type'(B'') := \type(B)$ and $h'|_{B''}$ to map an element $c$ from $B''$ on an element $d$ from $B$ iff $c$ is the fresh copy of $d$.
	
	\item 
	If $\type(B) = B_0$, then we consider each abstract type $t' < B$.
	We denote $h_{t'} : t' \rightarrow B$ the homomorphism provided by the definition of $t' < B$ and $P_B$ the hard pattern of $B$ that corresponds to $P_0$ in $B_0$ (recall that $B_0$ and $B$ are isomorphic and that this isomorphism is unique due to self-join freeness).
	If $h_{t'}$ uses the hard pattern $P_B$, we do nothing.
	Otherwise:
    \begin{itemize}
        \item If $h'(e)$ does not occur in $B$, then we simply introduce an isomorphic copy $B''$ of $t'$ as a successor bag of $B'$.
        \item Otherwise, let $e_1, \dots, e_n$ be the preimages of $h'(e)$ by $h_{t'}$.
	           For each $1 \leq i \leq n$, we introduce a fresh copy $B''$ of $t'$ as a successor bag of $B'$ and connect those two by identifying elements $e$ from $B'$ and the fresh copy of $e_i$ from $B''$. 
    \end{itemize}
    %
	%
	We set $\type'(B'') := t'$ and 
    $h'|_{B''} := h_{t'} \circ h_{B''}$
    where $h_{B''} : B'' \rightarrow t'$ is the isomorphism from $B''$ to $t'$ (recall that $B''$ is defined as an isomorphic copy of $t'$).
	
	
\end{itemize}

At each iteration of the above, $D'$ and $h'$ are being extended so we can safely take the potentially infinite union of their successive values to obtain the final $D'$ and $h'$.
We say that $D'$ with typing function $\type'$ is the unraveling of type $t$ in $D$ with respect to $\type$.

We first observe that the construction of $D'$ may only introduce hard patterns whose type are smaller or equal than those already present in $D$.
\begin{lemma}
    \label{lemma:types-are-simplifying}
    For all type of bag $t'$, if $t'$ is realized in $D'$ and contains a hard pattern, then there exists a type $t$ realized in $D$ such that $t' \leq t$ and $t$ contains a hard pattern.
\end{lemma}

\begin{proof}
We examine each possible step in the construction of $D'$.
\begin{itemize}
    \item Root bag of $D'$.
    \begin{itemize}
        \item If the root bag $B$ of $D'$ is an isomorphic copy of the root bag $B$ of $D$, their types are equal and the claim holds.
        \item Otherwise, the root bag $B'$ of $D'$ is obtained as an isomorphic copy of ${p}_{h_0, B}$ (see the construction for the definition of $B$, $h_0$ and $p$) and we notice that $B' < B_0$ where $B_0$ is the type of the root bag of $D$.
        Indeed, $h_0$ is the desired homomorphism from ${p}_{h_0, B}$ to $B_0$, and cannot be an isomorphism as otherwise the hard pattern $P_B$ in $B$ corresponding to the hard pattern $P_0$ in $B_0$ (recall that $B$ and $B_0$ are isomorphic and that this isomorphism is unique due to self-join freeness) would be used by $h_0$ contradicting that it is not used by $\Gamma$.
    \end{itemize}
    \item Further bags of $D'$.
    \begin{itemize}
        \item If we introduce a fresh copy of an existing bag, then their types are equal and the claim holds.
        \item Otherwise, the construction only introduces fresh copies of strictly smaller types than the type $B_0$ of the bag of interest in $D$, and that by definition $B_0$ contains a hard pattern, thus the claim holds. 
    \end{itemize}
\end{itemize}
\end{proof}

We now need to verify that this unraveling preserves $D_0$-compatibility and good-typing.

\begin{lemma}
	\label{lem:main-technical-lemma-lower-bound}
	$D'$ is $D_0$-compatible and $\type'$ is a good typing function on $D'$.
\end{lemma}

\begin{proof}
	$\type'$ is a good typing function by construction.
    For Points~1 and 3 of $D_0$-compatibility, we set $\delta' := \delta \circ h'$, which is clearly a homomorphism since it is the composition of two homomorphisms, and we examine the two cases in the constructions of the root node $B'$ of $D'$:
    \begin{itemize}
        \item If it is obtained as a copy of the root node $B$ of $D$, then Point~1 of the $D_0$-compatibility of $D$ immediately yields that $\bar a \in B'$ and since $h'$ is the identity on $B'$, and via Point~3 in the $D_0$-compatibility of $D$, we obtain $\delta'(\bar a) = \bar a$.
        \item Otherwise $B'$ has been obtained as an isomorphic copy of $p_{h_0, B}$ (see the construction for the definition of $h_0$ and $p$ and $B$). It is already made clear in the construction that $\bar a \subseteq \adom(B')$, and by definition $h'(\bar a) = h_0(\bar a) = \bar a$ thus $\delta'(\bar a) = \bar a$ as desired.
    \end{itemize}
    
	We now turn to Point~2 of $D_0$-compatibility, which is the most challenging part.
	We need to prove that $\bar a \in \Pi(D')$.
    Recall that the type $t := (B_0, P_0)$ that guided the unraveling is non-empty and that $\type^{-1}(t)$ is not crucial for $\bar a \in \Pi(D)$, which is witnessed by the derivation $\Gamma := (V_\Gamma, E_\Gamma, \ell_\Gamma, \rho_\Gamma, h_\Gamma)$ of $\bar a$ from $\Pi$ in $D$ that does not use any hard pattern from $\type^{-1}(t)$.
    We translate the derivation $\Gamma$ of $\bar a$ from $\Pi$ in $D$, used in the construction of $D'$, into a derivation $\Gamma'$ of $\bar a$ from $\Pi$ in $D'$.
    This $\Gamma'$ takes the form $\Gamma' := (V_{\Gamma'}, E_{\Gamma'}, \ell_{\Gamma'}, \rho_{\Gamma'}, h_{\Gamma'})$.
    We define the different nodes of $\Gamma'$ starting from the root of $V_{\Gamma'}$.
    Let $r$ be the root of $V$.
    We let $r$ to also be the root of $V_{\Gamma'}$.
    We also let $\ell_{\Gamma'}(r) := \ell_{\Gamma}(r)$ and $\rho_{\Gamma'}(r) := \rho_\Gamma(r)$.
    We denote $q$ the body of $\rho_{\Gamma'}(r)$ and $g := h_{\Gamma}(r)$ the homomorphism of $q$ to the database induced by successors of $r$ in $V_\Gamma$.
    In particular, $g$ does not use any hard patterns from $\type^{-1}(t)$.
    Based on $g$, we construct the homomorphism $h_{\Gamma'}(r)$, henceforth denoted $g'$, that shall map $q$ to the database induced by successors of $r \in V_{\Gamma'}$.
    As we shall verify along the construction, the built $g'$ will be an homomorphism from the EDB atoms of $q$ to $D'$ and satisfy $h' \circ g' = g$.
    We proceed by induction on each connected component of $q$.
    Consider $p$ one such connected component.
    \begin{itemize}
    	\item Base case.
    	We start by finding an appropriate bag $B'$ in $D'$ from where to start the construction of $g'$ on $p$.
		This is achieved as follows:
    		\begin{itemize}
    			\item If $p$ contains one of the variable occurring in the head of the goal rule $\rho_{\Gamma'}(r)$, we let $B'$ be the root bag of $D'$ and $B$ be the root bag of $D$.
    			
    			\item If $p$ does not contain any of the variables occurring in the head of the goal rule $\rho_{\Gamma'}(r)$, we pick $B$ in $D$ such that $g(v) \in \adom(B)$ for at least one variable $v$ of $p$.
    			We then observe that there exists a bag $B'$ and an element $e$ in $B$ such that $h'(e) = g(v)$:
    			\begin{itemize}
    			\item If $B$ is the root bag of $D$, then the root bag of $D'$ contains either directly $g(v)$ if $\type(B) \neq B_0$, or, if $\type(B) = B_0$, in which case $h'(g(v)) = g(v)$; or a copy of $v$ (potentially renamed in some $a \in \bar a$) stemming from the induced $q_{g, B}$, in which case $h'(v) = g(v)$ (or $h'(a) = g(v)$ if renaming happened).
    			\item Otherwise, the root bag of $D'$ has received as a successor: either an isomorphic copy of $B$ if $\type(B) \neq B_0$, in which we find a fresh copy $e$ of $g(v)$ and $h'(e) = g(v)$; or a fresh copy of $p_{g, B}$ in which the copy $e$ of $v$ satisfies $h'(e) = g(v)$. 
    			\end{itemize}
    		\end{itemize}
    	With these choices of $B'$ and $B$ at hand, we can now define $g'$ on a subpart $p_0$ of $p$ using the identified $B'$ and will see how to extend $g'$ to the complete $p$ in the induction case.
    	
    	We define $p_0$ as the restriction of $p$ to the biggest set of variables $V_0$ such that: $V_0$ is a connected set of variables in $p$ and $g(V_0) \subseteq B$.
    	Now, if $\type(B) \neq B_0$, then $B'$ is isomorphic to $B$ and we can define $g'$ on $V_0$ as the composition $(h'|_{B'})^{-1} \circ g|_{V_0}$, where we recall that $h'|_{B'}$ has been defined as the (unique) isomorphism from $B$ to $B'$.
    	Otherwise, that is $\type(B) = B_0$, we have that $B'$ is an isomorphic copy of $p_{g, B}$.
    	Since $p_0 \subseteq p_{g, B}$, we define $g'$ on $V_0$ as $(h'|_{B'})^{-1} \circ g|_{V_0}$ where we recall that $h'|_{B'}$ has been defined as the unique isomorphism from (the copy of) $p_{g, B}$ to $B'$.
    	
    	Note that in both cases, the constructed $g'$ is an homomorphism, from the EDB atoms of $q$ that only involve variables of $V_0$ to $D'$, and that it satisfies $h' \circ g' = g$.
 
    	\item Induction case.
    	Assume already constructed $g'$ for some connected subset of variables $V_0$ of $p$ and assume that $\mn{var}(p) \setminus V_0 \neq \emptyset$. 
    	As $p$ is connected, there exists a variable $v \notin V_0$ connected to a variable $v_0 \in V_0$ by an atom $\alpha$.
    	Consider a bag $B$ of $D$ containing $g(\alpha)$; in particular we have that $g(v_0), g(v) \in B$.
    	Consider the query $p'$ obtained as the connected component of $p$ in the induced query $p_{g, B}$,
    	 and let $V := \var(p') \setminus V_0$.
    	We will now extend $g'$ to these variables $V$.
    	
    	Let $B''$ be a bag of $D'$ containing $g'(v_0)$.
    	By construction, if $\type(B) \neq B_0$, then $B''$ has a successor $B'$ that is an isomorphic copy of $B$ and we can extend $g'$ on $V$ as the composition $(h'|_{B'})^{-1} \circ g|_{V}$, where we recall that $h'|_{B'}$ has been defined as the (unique) isomorphism from the isomorphic copy of $B$ to $B'$.
    	Otherwise $\type(B) = B_0$, and $B''$ has a successor $B'$ that is an isomorphic copy of $p_{g, B}$.
    	Since $p' \subseteq p_{g, B}$, we define $g'$ on $V$ as $(h'|_{B'})^{-1} \circ g|_{V}$ where we recall that $h'|_{B'}$ has been defined as the unique isomorphism from (the copy of) $p_{g, B}$ to $B'$.
    	
    	Here again, we obtained that $g'$, now extended to $V_0 \cup V$, satisfies: $h' \circ g' = g$.
    	We also argue that $g'$ defines an homomorphism from the EDB atoms of $q$ that only involve variables of $V_0 \cup V$ to $D'$.
    	If an EDB atom of $q$ only involves variables from $V_0$, we are done by induction hypothesis.
    	Otherwise, assume that we have an EDB atom $\beta$ of $q$ that involves at least one variable $v'$ from $V$.
    	It suffices to prove that $\beta \in p'$, which then guarantees that, in both cases in the definition of the extension of $g'$ to $V_0 \cup V$, $g'(\beta) \in D'$.
    	By definition of $p'$, it thus suffices to prove that every variable of $\beta$ belongs to the connected component of variable $v$ in the induced query $p_{g, B}$.
    	Consider a variable $u$ occurring in $\beta$.
    	Since $p$ is connected, there exists a path of variables from $u$ to $v$ in $p$.
    	If all atoms along that path are mapped within $B$, then it is direct that $u$ and $v$ belongs to the same connected component in $p_{g, B}$.
    	Otherwise, some atoms along that path may be mapped outside of $B$.
    	However, due to $D$ being tree-shaped, each time the image of this path exits the domain of $B$, it must re-enter it next time via the same domain element in $B$.
    	The definition of the induced $p_{g, B}$ then identifies the two corresponding variables that are both mapped on that domain element of the bag $B$.
    	The path from $u$ to $v$ thus translates into a (now shortened) path in $p_{g, B}$, as desired.
    	
    \end{itemize}
    
    We obtain the final $g'$ as the union of the constructed $g'$ on each connected component of $q$.
    We now need to guarantee that $g'$ defines an homomorphism from $q$ to the database induced by successors nodes of $r$ in $\Gamma'$.
    Note that we already verified that $g'$ defines an homomorphism of the EDB atoms of $q$ to $D'$.
    Therefore, we can simply introduce leaves as successor of $r$ that refer to these facts from $D'$.
    It remains to treat the case of IDB atoms, which we handle by introducing further inner nodes in $\Gamma'$.
    
    Assume there exists $\gamma(y)$ an IDB atom of $q$.
    By definition of $\Gamma$, there must be a successor $s$ of $r$ in $V_\Gamma$ such that $\ell(s) = \gamma(g(y))$.
    We introduce a successor $s'$ of $r$ in $V_{\Gamma'}$ with $\ell_{\Gamma'}(s') := \gamma(g'(y))$ and $\rho_{\Gamma'}(s') := \rho_\Gamma(s)$.
    To define $h_{\Gamma'}(s)$, we proceed as for the root case and don't reproduce the full construction here.
    We only highlight that, from the construction above we have $h'(g'(y)) = g(y)$, which echoes the property $h'(\bar a) = \bar a$, relevant for the root node of $\Gamma'$, and allows to pick the correct starting bag in the base case of the induction on connected components.
    The rest of the argument is the same.
    
    Notice that this process terminates as each branch in the resulting $V_{\Gamma'}$ will be no longer than the longest branch in $V_\Gamma$.
    The obtained $V_{\Gamma'}$ is therefore finite as desired, and $\Gamma'$ is the intended derivation: we have $\bar a \in \Pi(D')$.

\end{proof}

Starting with $D_0$, we can now produce a sequence $D_0, \dots, D_k, \dots$ of $D_0$-compatible databases with respective good typing functions $\type_0, \dots, \type_k, \dots$
We now prove that this sequence must be finite, that is the halting condition is reached in finitely many steps.

\begin{lemma}
\label{lem:halting}
    There exists $k \geq 0$ such that $D_k$ and $\type_k$ satisfy the halting condition.
\end{lemma}

\begin{proof}
	By construction, each $D_{k+1}$ is obtained by replacing \emph{all} occurrences of a maximal bag type $t$ from $D$ by smaller types for $<$.
    The maximality assumption warrants that $t$ can never be obtained again during the construction.
	As there are only finitely many types of bags, this can only be done finitely many times and therefore the halting condition is reached in finitely many steps.
	
    To be more formal, one can proceed as follows: define the type of a bagged database $D$ as the set $\Type(D) := \{ \type(B) \mid B \in \mn{bag}(D), B \text{ contains a hard pattern} \}$.
	There are still finitely-many different such types since the number of types of bags is finite in the first place.
    We define an order on the types of bagged databases; denote $T \prec T'$ if it satisfies the following conditions:
    \begin{enumerate}
        \item for all $t \in T$, there exists $t' \in T'$ such that $t \leq t'$; and 
        \item there exists $t' \in T'$ such that for all $t \in T$, we have $t_{max} \not\leq t$.
    \end{enumerate}
	We verify that $\prec$ defines an order on types of bagged databases.
    Irreflexivity follows from Condition~2. 
    For transitivity, assume we have $T_1 \prec T_2$ and $T_2 \prec T_3$.
    \begin{enumerate}
        \item 
      Let $t_1 \in T_1$.
    By Condition~1 of $T_1 \prec T_2$, there exists a type $t_2 \in T_2$ with $t_1 \leq t_2$.
    By Condition~1 of $T_2 \prec T_3$, there exists a type $t_3 \in T_3$ with $t_2 \leq t_3$.
    By transitivity of $<$, we obtain $t_1 < t_3$ so that Condition~1 of $T_1 \prec T_3$ holds.
    	\item 
    Let $t_3$ be the type provided by Condition~2 of $T_2 \prec T_3$, that is for all $t \in T_2$, we have $t_3 \not\leq t$.
    Let $t_1$ be a type from $T_1$.
    By Lemma~\ref{lemma:types-are-simplifying}, there exists $t_2 \in T_2$ such that $t_1 \leq t_2$.
    Assume by contradiction that $t_3 \leq t_1$.
    Transitivity of $\leq$ yields $t_3 \leq t_2$, which is a contradiction.
    \end{enumerate}
    
    We now claim that $\Type(D_{k+1}) \prec \Type(D_k)$ holds.
    Condition~1 is guaranteed by Lemma~\ref{lemma:types-are-simplifying}.
    For Condition~2, set $t_{max}$ to be the unraveled type $B_0$ in the construction of $D_{k+1}$.
    Since this type has been chosen as maximal among realized types in $D_k$ that contain some hard pattern, and that the construction of $D_{k+1}$ replaces every occurrence of type $t_{max}$ with bags of strictly smaller types (and leaves other bags unchanged), Condition~2 also holds.
\end{proof}


\subsection*{Increasing distance between hard patterns}
\label{sec:step_2}

We now describe the second step in the proof of
Theorem~\ref{theorem:one-cycle-is-crucial}, which takes the database produced by
the first step -- containing a crucial set of hard patterns -- and produces
a database in which the hard patterns are pairwise far apart. Afterwards the
hard patterns also have the same behaviour at the database scale: there are
automorphisms mapping any one hard pattern to any other. We fix a desired distance~$N$ and run the
following procedure.

\paragraph{Setup.}
Let $D_M$ be the $D_0$-compatible tree-like database produced by
Lemma~\ref{lem:halting}, with good typing function $\type_M$. By the
halting condition there is a hard pattern type $t_0 = (B_0, P_0)$ such that
$\type_M^{-1}(t_0)$ is crucial for $\bar a \in \Pi(D_M)$.

Let $\mathcal B = \{B_1, \dots, B_k\}$ be the bags containing the hard
patterns of $\type_M^{-1}(t_0)$. We assign to each bag~$B_i$ a fresh
concrete type $t_i := B_i$ (updating $\type$ accordingly) and set
\[
  \mathcal T := \{ (t_1, P_1), \dots, (t_k, P_k) \},
\]
where each $P_i$ is the renaming of $P_0$ under the isomorphism
$B_0 \to t_i = B_i$. Thus $(t_i, P_i) = \type(P_i)$ is the hard pattern type of the pattern in~$B_i$.

\paragraph{Invariants.}
Throughout the procedure we maintain a tree-like database~$D$, a typing
function~$\type$ on~$D$, and a set~$\mathcal T$ of hard pattern types with:
\begin{enumerate}
  \item[(I)] $S_{\mathcal T} := \bigcup_{t \in \mathcal T} \type^{-1}(t)$
    is crucial for $\bar a \in \Pi(D)$;
  \item[(II)] $D$ is $D_0$-compatible;
  \item[(III)] $\type$ is a good typing function on~$D$.
\end{enumerate}

\paragraph{The procedure.}
The procedure consists of two nested loops.

\medskip
\noindent\textbf{Inner loop.} While $|\mathcal T| > 1$, pick a hard pattern
type $t = (t', P) \in \mathcal T$ and distinguish:
\begin{itemize}
  \item \textbf{Case~A:} some derivation of $\bar a$ from $\Pi$ in $D$ uses no
    hard pattern in $\type^{-1}(t)$. 
    \begin{itemize}
      \item Unravel $t$ in~$D$ with respect to
    such a derivation~$\Gamma$, obtaining~$D'$ with typing~$\type'$ and
    homomorphism $h' \colon D' \to D$.
    \item Update $D := D'$, $\type := \type'$,
    and $\mathcal T := \mathcal T \setminus \{t\}$.
    \end{itemize} 
  \item \textbf{Case~B:} every derivation of $\bar a$ from $\Pi$ in $D$ uses at
    least one hard pattern in $\type^{-1}(t)$. 
    \begin{itemize}
      \item Update $\mathcal T := \{t\}$.
    \end{itemize}
\end{itemize}

\noindent\textbf{Outer loop.} When the inner loop terminates,
$|\mathcal T| = 1$, say $\mathcal T = \{(t', P)\}$. Let $B_1, \dots, B_m$ be
the bags of type~$t'$ in~$D$.
\begin{itemize}
  \item If $m = 1$: the single hard pattern is crucial on its own.
    \textbf{Terminate.}
  \item If $m > 1$ and $\min_{i \neq j} \dist_D(B_i, B_j) \geq N$:
    \textbf{proceed to next section (merging).}
  \item If $m > 1$ and $\min_{i \neq j} \dist_D(B_i, B_j) < N$: 
  \begin{itemize}
      \item assign fresh
    concrete types $t_i := B_i$ for each~$i$ (updating $\type$ accordingly);
    \item set $\mathcal T := \{(t_1, P_1), \dots, (t_m, P_m)\}$ with each $P_i$
    the renaming of~$P$;
    \item \textbf{restart the inner loop.}
  \end{itemize}
\end{itemize}

Note that the outer loop has three cases. 
In the first case we are immediately done, since it already returns one crucial cycle.
In the second case we get a database and set of crucial cycles that are pairwise far apart, the desired output of this part of the proof. We can take this database and type, and proceed to the next step, described in the next subsection.
The last case is the actual loop, where we have not yet reached the required distance.

We now show the invariants are maintained and the procedure terminates.
Invariant~(I) is Lemma~\ref{lem:cruciality} below. For Invariants~(II)
and~(III): in Case~A they follow from the Step-1 unraveling preserving
$D_0$-compatibility and good typing
(Lemma~\ref{lem:unraveling-props}). 
Case~B does not modify the database and the
restart case of the outer loop only reassigns fresh concrete types to bags,
which keeps the typing good. Termination is Lemma~\ref{lem:termination}.

First we state some further properties about the unraveling of the previous subsection which will be useful for the proofs below.

\begin{lemma}[Properties of the unraveling]\label{lem:unraveling-props}
  Let $D'$ be the unraveling of a type~$t_i$ in a tree-like database~$D$
  with good typing function~$\type$, via the homomorphism
  $h' \colon D' \to D$ and typing function~$\type'$. Then:
  \begin{enumerate}
    \item\label{prop:iso-off}
      For every bag~$B'$ of~$D'$ whose image $h'(B')$ is \emph{not} a
      $t_i$-bag, the restriction $h'|_{B'}$ is an isomorphism onto~$h'(B')$.
      In particular, every bag of~$D'$ carrying a concrete type is such a
      fresh copy.
    \item\label{prop:nonincr}
      $h'$ maps every path of~$D'$ to a walk of~$D$ of no greater length;
      hence $\dist_D\!\big(h'(B'), h'(\hat B')\big) \le \dist_{D'}(B', \hat B')$
      for all bags $B', \hat B'$ of~$D'$.
    \item\label{prop:repl}
      Every bag~$B'$ of~$D'$ whose image $h'(B')$ is a $t_i$-bag, is a
      fresh copy of an abstract type $t' < t_i$ and
      $h_{t'}$ \emph{does not use}~$P_0$.
  \end{enumerate}
\end{lemma}

\begin{proof}
  Immediate by inspection of the unraveling construction in the previous subsection.
\end{proof}

\begin{lemma}[Localization]\label{lem:localization}
  Let $D'$ be the unraveling of~$t_i$ in~$D$ via $h' \colon D' \to D$,
  both tree-like, and let $\Gamma'$ be a derivation of
  $\bar a$ from $\Pi$ in $D'$ that uses a hard pattern~$P$.
  Then~$P$ is contained in a single bag~$B'$ of~$D'$, and:
  \begin{enumerate}
    \item if $h'(B')$ is \emph{not} a $t_i$-bag, then $h'(P)$ is a hard
      pattern of~$D$ and \mbox{$\Gamma := h' \circ \Gamma'$} is a derivation
      of $\bar a$ from $\Pi$ in $D$ that uses~$h'(P)$;
    \item if $h'(B')$ is a $t_i$-bag, then $B'$ is a replacement bag and
      $h'|_{B'} = h_{t'}$ for some abstract type $t' < t_i$.
  \end{enumerate}
\end{lemma}

\begin{proof}
  Since $D'$ is tree-like, a hard pattern is always fully contained in a single
  bag~$B'$.
  By definition of ``uses'' there is a derivation node~$v$ with
  $h^{\Gamma'}_v$ mapping~$P$ bijectively onto a chordless cycle/tetra of the induced
  CQ $q_{h,v}$. All variables of that pattern map into~$B'$.

  Let $\Gamma := h' \circ \Gamma'$ for the derivation
  obtained from $\Gamma'$ by keeping the same tree, rules, and head/leaf
  labels, and replacing each body homomorphism $h^{\Gamma'}_v$ by
  $h' \circ h^{\Gamma'}_v$. Since $h'$ is a homomorphism with
  $h'(\bar a) = \bar a$, this is again a derivation of $\bar a$ from $\Pi$ in $D$.

  \emph{(1)} If $h'(B')$ is not a $t_i$-bag then $h'|_{B'}$ is an
  isomorphism by
  Property~\ref{lem:unraveling-props}(\ref{prop:iso-off}), so $h'(P)$ is
  an isomorphic copy of~$P$, hence again a chordless cycle/tetra. At
  node~$v$ the homomorphism $h^{\Gamma}_v = h' \circ h^{\Gamma'}_v$
  maps $h'(P)$ bijectively onto the corresponding pattern of $q_{h,v}$, since
  $h'|_{B'}$ is injective. Thus $\Gamma$ uses~$h'(P)$.

  \emph{(2)} If $h'(B')$ is a $t_i$-bag then, by
  Property~\ref{lem:unraveling-props}(\ref{prop:repl}), $B'$ is a
   bag with $h'|_{B'} = h_{t'}$ for some abstract type
  $t' < t_i$.
\end{proof}

\begin{definition}[Strong automorphism]\label{def:strong}
  Let $D$ be a $D_0$-compatible tree-like database with typing $\type$ and
  composed homomorphism $\delta \colon D \to D_0$. An automorphism
  $\alpha \colon D \to D$ is \emph{strong} if $\type(\alpha(B)) = \type(B)$
  for all bags~$B$ and $\delta \circ \alpha = \delta$.
\end{definition}

The Step-3 merge relies on the following two facts about the database~$D$
produced by Step~2, with surviving type~$t$ and crucial bags
$B_1, \dots, B_m$ of type~$t$.

\begin{lemma}[Concrete survivor]\label{lem:concrete-survivor}
  The surviving type~$t$ is concrete. Consequently every crucial bag~$B_j$
  is a fresh copy of one fixed bag $\hat B \subseteq D_0$: the restriction
  $\delta|_{B_j} \colon B_j \to \hat B$ is an isomorphism, and the
  isomorphisms $h_j \colon B_0 \to B_j$ can be chosen so that
  $\delta|_{B_j} \circ h_j$ is one fixed isomorphism $B_0 \to \hat B$,
  independent of~$j$. In particular $\delta|_{B_j} \circ h_j
  = \delta|_{B_k} \circ h_k$ for all $j, k$.
\end{lemma}

\begin{proof}
  The outer-loop reset assigns fresh \emph{concrete} types to the surviving
  bags, and the inner loop only removes types (Case~A) or restricts to one
  (Case~B).
  It never turns a concrete type into an abstract one. Hence the
  final survivor type~$t$ is concrete. 
  By Property~\ref{lem:unraveling-props}(\ref{prop:iso-off}) applied along the
  unraveling chain, a concrete-typed bag is a fresh copy of~$\hat B$, so
  $\delta|_{B_j} \colon B_j \to \hat B$ is an isomorphism for every crucial
  bag~$B_j$. Fix the isomorphism $g \colon B_0 \to \hat B$ coming from the
  identification $t = (\hat B, \cdot)$ with $B_0$, and set
  $h_j := (\delta|_{B_j})^{-1} \circ g$. Then $\delta|_{B_j} \circ h_j = g$
  for all~$j$, which is the claimed common isomorphism.
\end{proof}

\begin{lemma}[Cruciality Preservation]\label{lem:cruciality}
  Invariant~(I) holds at every point during the procedure.
\end{lemma}

\begin{proof}
  We show the invariant holds at setup and is preserved by each inner-loop
  step and by the restart case of the outer loop.

  \emph{Setup.} Initially
  $\mathcal T = \{(t_1, P_1), \dots, (t_k, P_k)\}$, where the bags of
  type~$t_i$ are exactly the bags of $\type_M^{-1}(t_0)$ (one each) and
  $P_i$ is the corresponding pattern. Hence $S_{\mathcal T}$ is exactly the
  set of patterns $\type_M^{-1}(t_0)$, which is crucial by
  Definition~\ref{def:haltingcondition}.

  \emph{Inner loop.} Assume $S_{\mathcal T}$ is crucial for
  $\bar a \in \Pi(D)$ before the iteration, and let $t = (t', P)$ be the
  chosen hard pattern type.

  \emph{Case~B.} Every derivation of $\bar a$ from $\Pi$ in $D$ uses a pattern in
  $\type^{-1}(t)$, so $\type^{-1}(t)$ is crucial. Since $D$ is unchanged and
  $\mathcal T$ becomes $\{t\}$, $S_{\mathcal T}$ is crucial.

  \emph{Case~A.} We unravel~$t$, obtaining $D'$ with $h' \colon D' \to D$
  and $\mathcal T$ becoming $\mathcal{T'} = \mathcal T \setminus \{t\}$. Suppose for
  contradiction that some derivation~$\Gamma'$ of $\bar a$ from $\Pi$ in $D'$ uses
  no hard pattern from $S_{\mathcal{T'}}$. Let
  $\Gamma := h' \circ \Gamma'$ be a derivation of $\bar a$ from $\Pi$ in $D$. By the
  hypothesis $S_{\mathcal T}$ is crucial, so $\Gamma$ uses some hard
  pattern~$X$ of a type $s = (s', P') \in \mathcal T$, lying in a
  bag~$B$ of~$D$ of type~$s'$.

  \emph{If $s' \neq t'$:} the node~$v$ of $\Gamma$ witnessing the usage
  of~$X$ corresponds to node~$v$ of $\Gamma'$, whose homomorphism lands
  in a bag~$B'$ of~$D'$ with $h'(B') = B$. Since $B$ is of type $s'$ and $s' \neq t'$,  $h'|_{B'}$ is
  an isomorphism (Property~\ref{lem:unraveling-props}(\ref{prop:iso-off})).
  Hence $\Gamma'$ uses $(h'|_{B'})^{-1}(X)$, a pattern of type
  $s \in \mathcal T \setminus \{t\}$, contradicting the choice
  of~$\Gamma'$.

  \emph{If $s' = t'$:} then $X$ is an instance of~$P$ in a $t'$-bag~$B$.
  By Lemma~\ref{lem:localization}, the pattern of~$\Gamma'$ that maps onto
  $X$ lies in a single bag~$B'$ with $h'(B')$ a $t'$-bag. Hence $B'$
  is a bag with $h'|_{B'} = h_{t''}$ for some $t'' < t'$. But
  $\Gamma$ using~$X$ through this node forces $h_{t''}$ to biject a
  chordless cycle/tetra of~$B'$ onto~$P$, i.e.\ $h_{t''}$ uses~$P$,
  contradicting Property~\ref{lem:unraveling-props}(\ref{prop:repl}).

  In either case we reach a contradiction, so every derivation of
  $\bar a$ from $\Pi$ in $D'$ uses a pattern from $S_{\mathcal T \setminus \{t\}}$ making $S_{\mathcal T \setminus \{t\}}$  crucial.

  \emph{Outer loop (restart case).} The restart reassigns fresh concrete
  types to the bags of the surviving type without changing~$D$ or the set of
  hard patterns under consideration, so $S_{\mathcal T}$ is unchanged and
  cruciality is preserved.
\end{proof}

\begin{lemma}[Distance Increase]\label{lem:distance}
  Let $D$ and $\mathcal T$ be the database and type set at the \emph{start}
  of an outer-loop iteration, and let $D^*$ and
  $\mathcal T^* = \{(t', P)\}$ be the database and surviving type at its
  end. Let
  $d := \min (\dist_D(s', s''))$ over bag-types $s', s''$ with $(s',P'), (s'', P'') \in \mathcal{T}$. Then
  \[
    \dist_{D^*}(t', t') \;\ge\; 2d + 1.
  \]
\end{lemma}

\begin{proof}
  The iteration is a sequence of inner steps, each unraveling a type
  (Case~A) or restricting to one (Case~B). Case~B leaves the database
  unchanged, so it suffices to track Case~A steps (the edge case of the inner loop immediately applying Case~B is fine, since then the outer loop can terminate with a singular hard pattern). 
  Note that $t'$ is never unraveled, otherwise it could not be in $\mathcal{T}^*$.

  We argue by induction over the completed Case~A steps that the minimum
  distance between distinct $t'$-bags is at least~$2d+1$, where~$d$ is
  measured \emph{once}, at the start of the iteration. Throughout we use
  that unraveling is distance-non-decreasing
  (Property~\ref{lem:unraveling-props}(\ref{prop:nonincr})): the image of a
  path is a walk of no greater length, so distances between bags of a
  retained type never drop below their start-of-iteration value. In
  particular, $d$ measured at the start remains a valid lower bound on the
  distance between any two bags of distinct types of~$\mathcal T$ at every
  intermediate database.

  \emph{Base case.} At the start, the surviving type's bags are the
  freshly assigned single bags, so any two distinct $t'$-bags are at
  distance~$\infty \ge 2d+1$ (there are none yet, or the bound is vacuous).

  \emph{Inductive step.} Let $D'$ be the unraveling of type $s'$ in a database $D$ obtained in some Case~A step and $h' \colon D' \to D$ the corresponding homomorphism.
  Furthermore let $B_1, B_2$ be two distinct
  $t'$-bags in~$D'$. Since $t'$ was not unravelled, both $B_1$ and $B_2$ are fresh copies and $h'$ is an
  isomorphism on each.

  \emph{If $h'(B_1) \neq h'(B_2)$:} these are distinct $t'$-bags in~$D$, and
  $h'$ maps the $B_1$-$B_2$ path to a walk between them, so by
  Property~\ref{lem:unraveling-props}(\ref{prop:nonincr}) and the
  inductive hypothesis
  $\dist_{D'}(B_1, B_2) \ge \dist_D(h'(B_1), h'(B_2)) \ge 2d+1$.

  \emph{If $h'(B_1) = h'(B_2)$:} both map to a $t'$-bag $\hat B$, so at
  least one is a copy created because $\hat B$'s subtree is attached to an
  $s'$-bag~$B$ via an element~$e$ whose replacement has several preimages,
  each receiving a copy of the subtree. We decompose the $B_1$-$B_2$ path:
  \begin{itemize}
    \item \emph{$B_1$ to the replacement of~$B$.} Before unraveling,
      $h'(B_1)$ was at distance at least~$d$ from~$B$ (both have types
      in~$\mathcal T$, using the start-of-iteration bound above), and
      copying preserves the subtree below the replacement, so this segment
      has length $\ge d$.
    \item \emph{Within the replacement of~$B$.} The copies attach at
      distinct preimages $e_1 \neq e_2$ of~$e$ (else $B_1 = B_2$). Crossing
      between distinct elements costs $\ge 1$ edge.
    \item \emph{Replacement of~$B$ to~$B_2$.} Symmetrically $\ge d$.
  \end{itemize}
  Hence $\dist_{D'}(B_1, B_2) \ge d + 1 + d = 2d+1$. \qedhere
\end{proof}


\begin{lemma}\label{lem:termination}
  The procedure terminates.
\end{lemma}

\begin{proof}
  \emph{Inner loop.} Each iteration decreases $|\mathcal T|$ by at least
  one: Case~A removes a type, Case~B collapses to a single type. So the
  inner loop terminates with $|\mathcal T| = 1$.

  \emph{Outer loop.} Let $d_j$ denote the minimum pairwise distance between
  distinct bags of the surviving type at the of the $j$-th outer loop iteration. 
  Let $d_0$ be the distance at the beginning of the procedure, so in the worst case $d_0 = 0$.
  By Lemma~\ref{lem:distance} the next surviving type satisfies
  $d_{j+1} \ge 2 d_j + 1$ hence $d_j \geq 2^{j-1}$.
  The outer loop exits as soon as the
  surviving type has a single bag ($m = 1$) or the minimum distance
  reaches~$N$. Since $d_j$ grows without bound, the latter occurs after
  finitely many restarts. Therefore the procedure terminates.
\end{proof}

\begin{lemma}[Bag Automorphism]\label{lem:subtree-iso}
  Let $B, \hat B$ be bags of~$D$ with $\delta(B) = \delta(\hat B)$, and let
  $\varphi_0 \colon B \to \hat B$ be the isomorphism with
  $\delta|_{\hat B} \circ \varphi_0 = \delta|_B$. Then there is a strong
  automorphism $\alpha \colon D \to D$ with $\alpha|_B = \varphi_0$, and  $\alpha(B) = \hat B$, and $\delta \circ \alpha = \delta$.
\end{lemma}

\begin{proof}
  Recall that $D$ is obtained from~$D_0$ by a finite chain of unravelings,
  and $\delta$ is the composite homomorphism back to~$D_0$. Since
  $\delta(B) = \delta(\hat B)$ but $B \neq \hat B$, there is a step of the
  chain at which $B$ and~$\hat B$ first become distinct. Let
  $\delta = \delta_2 \circ h \circ \delta_1$, where $\delta_1 \colon D \to
  D_1$ maps into the database $D_1$ just after that step and $h$ is the
  unraveling homomorphism of the step itself. We have
  \[
    \delta_1(B) \neq \delta_1(\hat B), \qquad
    h(\delta_1(B)) = h(\delta_1(\hat B)) .
  \]
  Thus in~$D_1$ the bags $B^1 := \delta_1(B)$ and $\hat B^1 := \delta_1(\hat
  B)$ are two \emph{distinct} bags with the same image under~$h$, two
  preimages of one bag of the database below.

  We first build a strong automorphism $\alpha_1$ of~$D_1$ from $B^1$
  to $\hat B^1$. The attachment rule of the unraveling indexes the successor
  bags created at a bag~$B'$ by pairs $(e, B^*)$ with $e \in \adom(B')$
  and $B^*$ an adjacent bag in $h(D_1)$. Also every condition governing whether and how a successor is attached, which case in the unraveling to pick for $\type(B^*)$, which
  abstract types $t' < B^*$ occur, whether the image of~$e$ occurs in
  $B^*$, and the number of preimages there, depend only on the
  \emph{image} $h(e)$ of~$e$, not on~$B'$ itself. Hence two bags
  with the same image spawn identical structures. Let $\alpha_1$ be the isomorphism from $B_1$ to $\hat B_1$ obtained from this.
  By construction $\alpha_1$ preserves $h$-images, $h \circ \alpha_1 = h$.
  It never changes the image of any element, only which preimage is chosen.
  Since every attachment decision depends only on $h$-images, $\alpha_1$
  preserves all attachment data and is therefore a strong automorphism
  of~$D_1$ with $\alpha_1(B^1) = \hat B^1$.

  Now transport $\alpha_1$ up through the remaining steps. The steps
  packaged into $\delta_1 \colon D \to D_1$ are themselves unravelings, and
  $\alpha_1$ preserves $h$-images ($h \circ \alpha_1 = h$), so the same
  attachment-depends-only-on-images argument lifts $\alpha_1$ to a strong
  automorphism $\alpha$ of~$D$ with $\delta_1 \circ \alpha = \alpha_1 \circ
  \delta_1$ and $\delta \circ \alpha = \delta$. On~$B$, both $\alpha$
  and~$\varphi_0$ are the unique $\delta$-preserving isomorphism $B \to \hat
  B$, so $\alpha|_B = \varphi_0$.
\end{proof}

\begin{corollary}[Crucial-bag automorphism]\label{lem:crucial-swap}
  Let $B_j, B_k$ be crucial bags of~$D$ (both of the concrete survivor
  type~$t$), and let $h_j \colon B_0 \to B_j$ and $h_k \colon B_0 \to B_k$ be the isomorphisms from Lemma~\ref{lem:concrete-survivor}. Then there is a strong automorphism
  $\alpha_{j,k} \colon D \to D$ with
  \begin{enumerate}
    \item $\alpha_{j,k}(B_j) = B_k$;
    \item $\alpha_{j,k} \circ h_j = h_k$ on $\adom(B_0)$, hence
      $\alpha_{j,k}(P_j) = P_k$.
  \end{enumerate}
\end{corollary}

\begin{proof}
  By Lemma~\ref{lem:concrete-survivor}, $B_j$ and~$B_k$ are fresh copies of the
  same bag $\hat B \subseteq D_0$, so $\delta(B_j) = \delta(B_k) = \hat B$,
  with $\delta|_{B_j} \circ h_j = \delta|_{B_k} \circ h_k$ (a fixed
  isomorphism $B_0 \to \hat B$). Let $\varphi_0 := (\delta|_{B_k})^{-1} \circ
  \delta|_{B_j} \colon B_j \to B_k$, so that $\delta|_{B_k} \circ \varphi_0
  = \delta|_{B_j}$ and $\varphi_0 \circ h_j = h_k$. By
  Lemma~\ref{lem:subtree-iso} there is a strong automorphism $\alpha_{j,k}$ of~$D$
  with $\alpha_{j,k}|_{B_j} = \varphi_0$. Then $\alpha_{j,k}(B_j) = B_k$ and
  $\alpha_{j,k} \circ h_j = \varphi_0 \circ h_j = h_k$ on $\adom(B_0)$, whence
  $\alpha_{j,k}(P_j) = \alpha_{j,k}(h_j(P_0)) = h_k(P_0) = P_k$.
\end{proof}

%
%
%
%

\subsection*{Merging hard patterns}\label{sec:step_3}

We now describe the third and final step, which merges the distant crucial
hard patterns produced by the procedure of the previous subsection into a single crucial pattern, completing the proof of Theorem~\ref{theorem:one-cycle-is-crucial}.


The \emph{diameter} of a graph is the maximum, over all pairs of vertices,
of their distance (the length of a shortest path between them). Let $d$ be
the maximum diameter of the Gaifman graph of a rule body of~$\Pi$. We run
the distance-increasing procedure of the previous subsection with
$N := d + 1$, so that it returns a database~$D$ whose crucial bags
$B_1, \dots, B_m$ (all of the surviving concrete type~$t'$) are pairwise at
distance at least~$d + 1$.
Let $P_j := h_j(P_0)$ be the hard pattern in~$B_j$ and
$S := \{P_1, \dots, P_m\}$. By Lemma~\ref{lem:cruciality}, $S$ is crucial for
$\bar a \in \Pi(D)$.

If $m = 1$ the single hard pattern is already crucial and we are done, so
assume $m > 1$. By Lemma~\ref{lem:crucial-swap}, for each pair $j,k$ of indices with $1 \leq j, k \leq m$, there is a strong automorphism $\beta_{j,k} \colon D \to D$ with
$\beta_{j,k}(B_j) = B_k$ and $\beta_{j,k}(P_j) = P_k$.

We choose the target bag~$B_1$ as follows: if some crucial bag is the root
bag of~$D$, we let $B_1$ be that bag. Otherwise, if some crucial bag shares an element
with the root bag, we let $B_1$ be that bag (at most one can, as the $B_j$
are pairwise at distance at least $1$). Otherwise we choose $B_1$
arbitrarily.

We define the \emph{merged database} $D'$ by collapsing each crucial bag
onto~$B_1$ using the automorphisms. Let
\[
  \pi \colon \adom(D) \to \adom(D'), \qquad
  \pi(a) :=
  \begin{cases}
    \beta_{j,1}(a) & \text{if } a \in \adom(B_j) \text{ for some }
                      j \in [m], \\
    a               & \text{otherwise.}
  \end{cases}
\]
This is well defined: the crucial bags are pairwise at distance
$\geq d + 1 \geq 1$, hence pairwise disjoint. 
For $j = 1$ we have $\beta_{1,1} = \mn{id}$ so $\pi$ is the identity on $\bar a$.
The merged database is
\[
  D' := \{ R(\pi(a_1), \dots, \pi(a_r)) \mid R(a_1, \dots, a_r) \in D \}.
\]
By construction $\pi \colon D \to D'$ is a (surjective) homomorphism. We
write $X := P_1$ for the hard pattern in~$B_1$. It is the unique pattern
of~$S$ that survives in~$D'$ as all others are identified onto it.
It is easy to see that $D'$ is a bagged database, but it is no longer tree-like.

\begin{lemma}\label{obs:derivable}
  $\bar a \in \Pi(D')$.
\end{lemma}

\begin{proof}
  Since $\pi \colon D \to D'$ is a homomorphism and homomorphisms preserve
  derivations, $\bar a \in \Pi(D)$ gives $\pi(\bar a) \in \Pi(D')$. As argued
  above $\pi(\bar a) = \bar a$, so $\bar a \in \Pi(D')$.
\end{proof}

Next we consider two interesting properties that the merged database has. They help us ensure that derivations are still similar to $D$, even though we are no longer tree-like.

\begin{lemma}[New cycles are long]\label{obs:long-cycles}
  Every cycle of~$D'$ that is not the image under~$\pi$ of a cycle of~$D$
  has length at least $d + 1$. Consequently 
  the patterns a rule body can use in~$D'$ are exactly the $\pi$-images of patterns already present in~$D$.
\end{lemma}

\begin{proof}
  A cycle of~$D'$ not coming from a single cycle of~$D$ must pass $B_1$.
  It leaves~$B_1$ through an element~$e$, traverses a path
  that in~$D$ runs from one crucial copy towards another, and returns
  to~$B_1$. Since distinct crucial bags are at distance at least $d + 1$
  in~$D$ and $\pi$ does not shorten paths outside~$B_1$, the returning
  portion has length at least $d + 1$. A pattern used by a rule body is the
  image of a chordless cycle or tetra in the body, whose Gaifman diameter is
  at most~$d$. Hence the only patterns a rule body uses in~$D'$ are images
  of patterns of~$D$.
\end{proof}

\begin{lemma}[Bounded return]
  \label{obs:bounded-return}
  Let $p$ be a rule body of~$\Pi$ (Gaifman diameter at most~$d$) and $g$ a
  homomorphism from~$p$ to~$D'$. If $g$ maps two variables into the same
  element of~$B_1$, or more generally if $g$ leaves~$B_1$ through an
  element~$e$ and returns to~$B_1$, it can only return through the same
  element~$e$. 
\end{lemma}

\begin{proof}
  Two distinct elements of~$B_1$ are images under~$\pi$ of elements lying in
  crucial copies that are pairwise at distance at least $d + 1$ in~$D$ (or
  of two elements of a single copy, in which case the bound is inherited
  from~the tree-likeness of $D$). A homomorphism from a body of diameter at most~$d$ cannot
  connect two such elements except through the single shared element~$e$.
\end{proof}


What remains to show is that $\{X\}$ is crucial in $D'$.
For that purpose we first need to define how to ``lift'' derivations in $D'$ to derivations in $D$. 
It is not straightforward since derivations in $D'$ might use facts in the mapping of one body that are originally far apart in $D$.
But we identify their origin and use the automorphisms provided by Lemma~\ref{lem:crucial-swap} to lift them into the same neighbourhood in $D$.

Recall that $\pi$ is injective for elements outside of $B_1$ and identifies, for each crucial
bag~$B_j$, its elements with the corresponding elements of~$B_1$ via
$\alpha_{j,1}$. The elements of~$B_1$ in~$D'$ are thus ``merge points'' with
several preimages in~$D$, one per crucial copy.

\begin{definition}[Origin]\label{def:origin}
  Let $g$ be a homomorphism from a rule body to~$D'$, and consider a path in
  the image consisting of one fact intersecting~$B_1$ in a single element~$e$,
  together with facts not intersecting~$B_1$. By Lemma~\ref{obs:bounded-return} such a
  path intersects~$B_1$ only in~$e$. Mapping the path back to~$D$ via~$\pi$
  everywhere except at~$e$, the fact at~$e$ intersects a unique crucial
  copy~$B_k$ in one element. We call~$k$ the \emph{origin} of every element
  and fact on the path other than~$e$ itself. The merge point~$e$ has no
  origin.
\end{definition}

\begin{lemma}\label{lem:single-crucial}
  The set $\{X\}$ is crucial for $\bar a \in \Pi(D')$.
\end{lemma}
\begin{proof}
  Let $\Gamma'$ be an arbitrary derivation of $\bar a$ from $\Pi$ in $D'$. We first
  build from it a derivation~$\Gamma$ of $\bar a$ from $\Pi$ in $D$ on the same
  tree with the same rules, adjusting only the body homomorphisms. 
  Then we show $\Gamma'$ uses~$X$.

  \medskip\noindent\emph{The lifting.}
  Let $\Gamma' = (V, E, \ell', \rho, h')$. 
  We construct
  $\Gamma = (V, E, \ell, \rho, h)$, defining only the body homomorphisms
  $h_v$ (the labels~$\ell$ are read off from~$h$). We process $V$ top-down,
  carrying for each node~$v$ two pieces of data set by its parent:
  \begin{itemize}
    \item the head element $a_v$, so the mapping of the head variable at $v$, as defined by the parent node (for the root we initialize
      $a_v := \bar a$);
    \item a strong automorphism $\beta_v$ of~$D$ (for the root, $\beta_{v_0} := \mn{id}$).
  \end{itemize}
We maintain
  \[
    (\ast)\qquad a_v = \beta_v(a'_v),
  \]
  where $a'_v$ is the corresponding head element in~$\Gamma'$.
  It holds at
  the root since $\beta_{v_0} = \mn{id}$ and $a'_{v_0} = \bar a$.

  Let $p$ be the EDB part of $\rho_v = S(x) \leftarrow q$ (IDB atoms are
  unary and impose no EDB constraint; they are handled as in~$\Gamma'$). We
  define $h_v$ on~$p$ and the data $a_w, \beta_w$ of each child~$w$.

  \emph{Case 1: $h'_v(p)$ does not meet~$B_1$.} Then $\pi$ is the identity
  on $h'_v(p)$, so $h'_v$ maps~$p$ into~$D$, and we just apply the inherited
  automorphism:
  \[
    h_v := \beta_v \circ h'_v .
  \]
  This maps~$p$ into~$D$ since $\beta_v$ is an automorphism. For each
  child~$w$ with head variable~$y$ we pass on the same automorphism,
  $\beta_w := \beta_v$ and $a_w := h_v(y)$. 
  Then $(\ast)$ holds for~$w$ because $a'_w = h'_v(y)$.

  \emph{Case 2: $h'_v(p)$ intersects~$B_1$.} We first fix a \emph{target}
  copy~$j$: if $\beta_v = \mn{id}$, let $j$ be the origin
  (Definition~\ref{def:origin}) of the path from~$a'_v$ into~$B_1$.
  If $\beta_v = \beta_{i,j'}$, let $j := j'$. We then define $h_v$ atom by atom:
  \begin{itemize}
    \item a fact of~$p$ mapped by~$h'_v$ entirely inside~$B_1$ is mapped by $\beta_{1,j}$;
    \item a fact on a path leaving~$B_1$, of origin~$i$, is mapped by
      $\beta_{i,j}$;
    \item a fact whose image avoids~$B_1$ entirely (and is not on any path intersecting $B_1$) is mapped by $\pi^{-1} = \mn{id}$
  \end{itemize}
  It is straightforward to verify that origin is well defined, that is, no element is part of two distinct origins. 
  Hence $h_v$ is well defined, and it maps~$p$ into~$D$ because each $\beta_{*,j}$ is an automorphism on~$D$.

    It remains to check that $h_v$ places the head where the parent expects,
  $h_v(x) = a_v = \beta_v(a'_v)$ (using $(\ast)$ and $a'_v = h'_v(x)$). This
  is what makes $\Gamma$ a valid derivation at~$v$ and what gives $(\ast)$
  for the children. There are three possibilities for the head variable~$x$.
  If $x$ lies in a component not intersecting with~$B_1$, it is handled as in Case~1 by
  the inherited automorphism.
  So $h_v(x) = \beta_v(h'_v(x)) = a_v$, and we are
  done. Otherwise $x$ has an origin~$i_x$ (or maps inside~$B_1$, in which
  case $i_x := 1$), and by construction $h_v(x) = \beta_{i_x, j}(h'_v(x))$,
  so the goal $h_v(x) = a_v$ becomes satisfying the equation
  \[
    \beta_{i_x, j} = \beta_v \qquad \text{on } h'_v(x).
  \]
  \emph{Case $\beta_v = \mn{id}$.} Here $a_v = a'_v$, and the target was chosen
  as $j := i_x$ (the origin of the head path), so $\beta_{i_x, j} =
  \beta_{i_x, i_x} = \mn{id}$ and $h_v(x) = h'_v(x) = a_v$.
 
  \emph{Case $\beta_v = \beta_{i,j}$.} The target is the second index~$j$,
  so the equation to satisfy is $\beta_{i_x, j} = \beta_{i, j}$ on $h'_v(x)$,
  i.e.\ $i_x = i$. Recall that $i$ was recorded by the parent as the first
  index of $\beta_v = \beta_{i,j}$ for the variable shared with~$v$, the
  origin of $a'_v = h'_v(x)$ in the parent's body. If $a'_v \in B_1$ then
  $a'_v$ is a merge point, so $i_x := 1$ in~$v$ and, the shared variable
  mapping into~$B_1$ in the parent as well, $i = 1$. Otherwise
  $a'_v \notin B_1$ and (being in the remaining case) $v$'s body
  reaches~$B_1$ from~$a'_v$, so $i_x$ is found by following that path back to
  a crucial copy. By Lemma~\ref{obs:bounded-return} a path of a rule body
  (Gaifman diameter at most~$d$) reaches at most one element of~$B_1$, and
  the crucial copies are pairwise at distance at least $d+1$, so at most one
  copy is reachable from~$a'_v$ within the body: the copy~$B_i$ in
  which~$a'_v$ sits in~$D$, the same one the parent reached from the same
  element. Hence $i_x = i$ and $h_v(x) = a_v$.
 
  We set $\beta_w := \beta_{i_y, j}$ and $a_w := h_v(y)$
  for each child~$w$ with head variable~$y$ of origin~$i_y$. 
  This defines~$\Gamma$ on all of~$V$ and by construction it is a derivation of $\bar a$ from $\Pi$ in $D$.

  \medskip\noindent\emph{Cruciality transfer.}
  We now use the lifted derivation~$\Gamma$ to show that $\Gamma'$
  uses~$X$. Since $S$ is crucial for $\bar a \in \Pi(D)$
  (Lemma~\ref{lem:cruciality}), $\Gamma$ uses some hard pattern
  $P_k = h_k(P_0) \in S$ at some node~$v$. 
  The body homomorphism $h_v$
  maps a chordless cycle or tetra~$Y$ of the rule body bijectively onto~$P_k$
  in~$B_k$. We show that at this same node $h'_v$ maps~$Y$ bijectively onto $X = P_1$,
  so that $\Gamma'$ uses~$X$; as $\Gamma'$ was arbitrary, this proves
  $\{X\}$ crucial for $\bar a \in \Pi(D')$.
 
  The pattern $P_k$ lies in the single bag~$B_k$, so $h_v$ maps all
  variables of~$Y$ into~$B_k$. In particular $v$ is a Case-2 node (a Case-1
  node has $h_v = \alpha_v \circ h'_v$ with $h'_v$ avoiding~$B_1$, so its
  image avoids the crucial bags). By localization
  (used now in~$D'$) the hard pattern $h'_v(Y)$ also lies in a single bag
  of~$D'$. Among the facts of~$Y$, those mapped by $\alpha_{1,k}$ have
  $h'_v$-image inside~$B_1$, while any facts mapped by some $\alpha_{i,k}$ with
  $i \neq 1$ would have $h'_v$-image on an origin-$i$ path leaving~$B_1$,
  hence in a bag other than~$B_1$. As $h'_v(Y)$ is confined to one bag, the
  second kind cannot occur. Every fact of~$Y$ is of the first kind, $h'_v$
  maps~$Y$ inside~$B_1$, and $h_v = \alpha_{1,k} \circ h'_v$ on~$Y$. Now by
  clause~2 of Lemma~\ref{lem:crucial-swap} for the pair~$(1,k)$,
  \[
    \alpha_{1,k}(P_1) = \alpha_{1,k}(h_1(P_0)) = h_k(P_0) = P_k .
  \]
  Applying $\alpha_{1,k}^{-1}$ gives $\alpha_{1,k}^{-1}(P_k) = P_1$, hence
  \[
    h'_v(Y) = \alpha_{1,k}^{-1}(h_v(Y)) = \alpha_{1,k}^{-1}(P_k) = P_1 = X .
  \]
  Since $\alpha_{1,k}$ is an isomorphism, $h'_v$ maps~$Y$ bijectively onto~$X$, i.e.\
  $\Gamma'$ uses~$X$.
 
  Finally, Lemma~\ref{obs:long-cycles} confirms that no other pattern could have
  been used: a rule body cannot use any pattern of~$D'$ except a $\pi$-image
  of a pattern of~$D$, so the witness is indeed~$X$ and not a cycle created
  by the merge.
\end{proof}
 
Together, Lemma~\ref{obs:derivable} and Lemma~\ref{lem:single-crucial} show that the
database $D_\Pi := D'$ and the hard pattern~$X$ satisfy the requirements of
Theorem~\ref{theorem:one-cycle-is-crucial}: $\bar a \in \Pi(D_\Pi)$ and $\{X\}$ is
crucial for $\bar a \in \Pi(D_\Pi)$.
It now remains to make the resulting database finite.
To this end, we just use compactness which guarantees that there exists a finite $D_\mn{fin}$ subdatabase of $D_\Pi$ that still guarantees $\bar a \in \Pi(D_{\mn{fin}})$.
This ends the proof of Theorem~\ref{theorem:one-cycle-is-crucial}.


\section{Proof of Lemma~\ref{lem:datalog-lower-cyclecombined}.}    
In this section we aim to prove the following lemma.
\lemdataloglowercyclecombined*

We call a CQ $q$ over a binary schema a \emph{chordless cycle} if there is a chordless cycle in $q$ that contains all variables in $q$. Likewise, we call a CQ $q$ a \emph{tetra} if its domain is a tetra in~$q$.
We make use of the lower bound for tetras established by~\cite{chen2019testability}, which can be stated as follows.

\begin{theorem}[Theorem 4.7 in~\cite{chen2019testability}]
\label{theo:chen-2}
Let $q$ be a self-join free Boolean CQ that is a tetra of size~$k$.  Then, testing whether $D \not \models q$ with one-sided error requires $\Omega(n^{\frac{1}{k}})$ queries, where $n$ is the domain size of~$D$.
\end{theorem}

We split the proof of Lemma~\ref{lem:datalog-lower-cyclecombined} in the chordless cycle and tetra case, starting with the former.

\begin{restatable}{lemma}{lemdataloglowercycle}
\label{lem:datalog-lower-cycle}
  Let $\Pi$ be a Boolean  MDLog program such that there exist a database $D$ with $D \models \Pi$
  and  a chordless cycle $C$ in  $D$
  with   $\{ C \}$  crucial for
  $D \models \Pi$. Then
  falsity of\/ $\Pi$ is not constant query testable with one-sided error.
\end{restatable}
\begin{proof}

We prove the lemma by a reduction from testing freeness of a self-join free triangle shaped CQ.
We use $D_\Pi$ to denote $D$ from the lemma to make clear that it is fixed.
Let $G$ be a database over a signature consisting of three binary symbols $R_1, R_2, R_3$ and $q_{\triangle} = R_1(y_1,y_2), R_2({y_2,y_3}), R_3(y_3,y_1)$.

We associate with every input $G$ to $\mathcal{P}_{q_{\triangle}}$
an input $f(G)$ to $\mathcal{P}_\Pi$. Moreover with any $\epsilon_G \in (0,1)$ we associate $\epsilon_{f(G)} = \frac{1}{|D_{\Pi}|}\epsilon_G$. 
We prove the three conditions of Definition~\ref{def:proptestreduction} below, and since 
$q_{\triangle}$ is self-join free and $\{y_1,y_2,y_3\}$ is a $3$-tetra in $q_\triangle$,
Theorem~\ref{theo:chen-2} implies that falsity of $\Pi$ is not constant query testable with one-sided error, yielding the lemma.

Let $V$ be the domain of $G$.
Recall that we need to show the following:

\begin{itemize}\itemsep=0pt
    \item[(i)] if $G \not\models q_{\triangle}$ then $f(G) \not\models \Pi$;
    \item[(ii)] if $G$ is $\epsilon$-far from being $q_{\triangle}$-free, then $f(G)$ is $\varepsilon_{f(G)}$-far from being $\Pi$-free;
    \item[(iii)] the answer to a query to a completion or size oracle for $f(G)$ can be computed from the answers to constantly many queries to completion or size oracles for $G$.
\end{itemize}

Let $S_C = \{c_1, \dots, c_k\}$ be the set of constants in the cycle $C$, $k \geq 4$, and let $c_{k+1} = c_1$.
Intuitively, we define the database $D_G = f(G)$ as follows.
Facts derived from $D_\Pi$ that use the constants in $S_C$ will also depend on the facts in $G$ and facts that do not use constants of $S_C$ will just be ``copied'', indepently of $G$. Concretely we use the first four constants in $S_C$ to detect triangles in $G$.

Without loss of generality, we can assume that for every $i \in \{1,2,3\}$ there is at least one fact $R_i(a^\dagger ,b^\dagger) \in G$, otherwise there is no homomorphism from $q_{\triangle}$ to $G$ and we could construct a trivial $f(G)$. Moreover, we will assume that those facts are isolated, i.e. no other facts in $G$ use $a^\dagger$ nor $b^\dagger$. It is easy to see that since there is a constant, independent of $G$, number of such facts and none of those facts can be used in forming a triangle, their existence does not interfere with property testing.

The domain of $f(G)$ is $S_C {\times} V \cup \big(\dom(D_{\Pi})\setminus S_C\big) {\times} \{\bot\}$.
For every fact  $R(d_1, \dots, d_p) \in D_\Pi$,
add to $f(G)$ the facts
\[R((d_1, g_1), \dots, (d_p, g_p))\]
based on the existence of constants from $S_C$ in $R(d_1, \dots, d_p)$.  Recall that $c_{k+1} = c_1$.
\begin{itemize}
 \item If $R(d_1, \dots, d_p)$ contains two distinct constants $c_i, c_{i+1}$ with $1 \leq i \leq 3$, then, for every $R_i(v,u) \in G$, we add a fact $R((d_1,g_1), \dots, (d_p,g_p))$, where for $1 \leq j \leq p$, we set $g_j = v$ if $d_j = c_i$, and $g_j = u$ if $d_j = c_{i+1}$, and $g_j = \bot$ otherwise.
    \item If $R(d_1, \dots, d_p)$ contains two distinct constants $c_i, c_{i+1}$ with $4 \leq i \leq k$, then for every $R_3(v,u) \in G$, we add a fact $R((d_1, g_1), \dots, (d_p, g_p)$ where for $1 \leq j \leq p$, we set $g_j = u$ if $d_j \in \{c_i, c_{i+1}\}$, and $g_j = \bot$ otherwise.
    \item If there is exactly one distinct constant $c_i$ in $R(d_1, \dots, d_p)$ then we add $R((d_1, g_1), \dots, (d_p, g_p))$ where for $1 \leq j \leq p$ if $d_j \neq c_i$ then $g_j = \bot$ and if $d_j = c_i$ then $g_j = v$ for some $v$ such that there is $R_i(v,u) \in G$ with $R_i = R_3$ for $i>3$.
    \item If there are no constants from $S_C$ in $R(d_1, \dots, d_p)$ then we add $R((d_1, \bot), \dots, (d_p, \bot))$.
    %
    %
\end{itemize}

The first case does the main work of detecting triangles. It uses the first four constants in $S_C$ to check $G$ for any triangles matching $q_{\triangle}$. If and only if $G$ does contain such a triangle, will there be elements $a,b,c \in \dom(G)$ such that there is a path $(c_1, a), (c_2, b), (c_3, c), (c_4, a)$ in $f(G)$ (with relation symbols matching the cycle $C$ in $D_\Pi$).
The second case closes the cycle in $f(G)$ for the constants after $c_4$ in $S_C$, so for example from $(c_4, a)$ to $(c_1, a)$. It does so by always taking the identity on the second element for every element $u$ in $G$ where a triangle could have ``ended'', that is, $R_3(v, u)$ for some $v,u \in \dom(G)$.
So if the first case did create a connection from $(c_1, a)$ to $(c_4, a)$, then the second case would now have created a cycle isomorphic to the crucial cycle $C$ in $D_\Pi$.

The remaining two cases ensure that we can derive the rest of $\Pi$ on $f(G)$. The third case is for edges leaving the cycle on $S_C$, which we only add for elements in $S_C \times V$ that could be part of a triangle. The last case is for edges not intersecting with the cycle on $S_C$ at all.



Now, for every fact $R(d_1, \dots, d_p) \in D_\Pi$ artificially multiply one of the facts
$R((d_1, a_1), \dots, (d_p, a_p))$ added in~the above process so~the multiplicity of the facts of form $R((d_1, a'_1), \dots, (d_p, a'_p))$ in $f(G)$
is exactly~$|G|$.
If possible,  we always choose a fact corresponding to a fact $R_i(a^\dagger, b^\dagger) \in G$.

The multiplicities can always be adjusted in such a way because for every fact $R(d_1, \dots, d_p) \in D_\Pi$, we have previously added 
\begin{itemize}
    \item at least one fact to $f(G)$, every relation $R_i$, $i \in \{1,2,3\}$, contains fact $R_i(a^\dagger, b^\dagger)$,
    \item and no more than $|G|$ facts, at most one fact per fact in $G$.
\end{itemize}
In the case of $R(d_1, \dots, d_p)$s that do not use elements of $S_C$, there is exactly one candidate to be used in the adjustement.
As a~consequence, $|D_G| = |D_\Pi||G|$.

%

\medskip
We now prove the correctness of the reduction.
Note that $f(G)$ can be computed in time linear in $|G|$: for every fact $R_i(a,b) \in G$
we scan $D_\Pi$ once, and we add no~more than $|D_\Pi|$ elements per scan. The final adjustment of the multiplicities can be done by
keeping for every fact $F$ in $D_\Pi$ a counter and a candidate for the copy procedure, e.g. the first fact $F'$ computed using $F$.

To show {\bf (iii)} observe the following.
For a size query $R(b_1, \dots, b_p)$,  where $b_i = (d_i,a_i)$ for some constants $d_i \in \adom(D_\Pi)$ and $a_i \in \adom(G)$ or $b_i = \ast$, identify all occurences $o_j$ of $R(y_1, \dots, y_p)$ in the database $D_{\Pi}$.
Any completion of $R(b_1, \dots, b_p)$ was produced by one of those occurences, according to one of the rules specified in the construction, or by adjusting the multiplicities. 
For the case of the use of the rules, for every such occurence determine the number $m_j$ of completions in the database. For the occurences $R(y_1, \dots, y_p)$ which does not use variables $y_i$ from the set $S_C$, $0 \leq m_j \leq 1$ as we have added only the fact $R((y_1, \bot), \dots, (y_p, \bot))$. For the occurences $R(y_1, \dots, y_p)$ that use variables $y_i$ from the set $S_C$, this can be done by identifying which unique rule in construction was used and then computing the number of generated facts by using oracle for $G$ to identify the number of $R_i$ facts in $G$ that were used to produce adequate completions. Single oracle call per instance of $R(y_1, \dots, y_p)$.

For the case of completions originating from adjusting the multiplicities, the added facts were not counted in the previous steps, as there is no overlap.
Moreover, every clone was added to adjust the count of facts compatible with some occurence $o_j$ of $R(y_1, \dots, y_p)$. Thus, we can simply determine the counts $m'_j$ and add them to counts obtained in the previous step. 
Observe that, by the construction, the cloned facts are either of the form $R((y_1, \bot), \dots, (y_p, \bot))$, if for every $1 \leq i \leq p$ $y_i \notin S_C$,
or were obtained using one of the construction rules together with the isolated edges~$R_i(a^\dagger, b^\dagger)$.

For those occurences $o_j$ which do not use variables $y_i$ from the set $S_C$, $m'_j = |G| -1 $ if $a_i = \bot$ for all $b_i \neq \ast$ in $R(b_1, \dots, b_p)$ or $m'_j = 0$ if one of $a_i$ is not $\bot$. For the other case, $m'_j = 0$ if the query enforces inappropriate value for some $a_i$, i.e.a~value conflicting with the rule or the fact~$R_i(a^\dagger, b^\dagger)$. Otherwise, $m'_j = |G| - r$, where $r$ is the answer to a size query $R_i(b'_1,b'_2)$ on $G$ with $R_i$ chosen to be consitent with the rule, and $b'_1, b'_2$ chosen to be consistent with the original query.
To compute $|G|$, simply query the sizes of relations $R_1, R_2,R_3$ in $G$.

The final value is $m = \sum_{j} m_j + m'_j$.

For a~completion query $R(b_1, \dots, b_p)$, where $b_i = (d_i,a_i)$ for some constants $d_i \in \adom(D_\Pi)$ and $a_i \in \adom(G)$ or $b_i = \ast$, identify all occurences $o_j$ of $R(y_1, \dots, y_p)$ in the database~$D_{\Pi}$.
Any completion of $R(b_1, \dots, b_p)$ was produced by one of those occurences, according to one of the rules specified in the construction, or was produced by adjusting multiplicities.
First, identify the answer $m$ to an analogue size query and if $m = 0$ then fail.

Otherwise, based on the counts (computed in the algorithm for the size query) choose with probability $\frac{m_j}{m}$ whether to return a fact constructed for one of the occurences $o_j$ or with probablility $\frac{m'_j}{m}$ whether to return a fact that was used to adjust multiplicity of the occurence $o_j$.

Use the choice to construct a completion of $R(a_1, \dots, a_p)$ based on the construction of the database or the adjustmento of the multiplicities. 
For occurences using variables $y_i$ from the set $S_C$, this may require a single call to the completion oracle for $G$.


For {\bf (ii)} recall that we set $\varepsilon_{f(G)} =  \frac{1}{|D_{\Pi}|} \varepsilon_G$. 
Let $m \leq  \varepsilon_{f(G)}|f(G)|$ and $H$ be $f(G)$ after removing $m$ facts from $f(G)$.
Via construction of $D_G$, removal of those facts corresponds to removal of no more than $m$ facts from $G$.
Since \[m \leq \varepsilon_{f(G)} |f(G)| \leq \frac{\varepsilon_G}{|D_\Pi|} \cdot |D_\Pi||G| = \varepsilon_G |G|\]
the database $G'$ obtained from $G$ by removing those facts satisfies $G' \models q_{\triangle}$.

This removal also corresponds to a database $D'_\Pi$ obtained by projecting the constants in~the domain to the first coordinate.
Clearly, there is a homomorphism from $D'_\Pi$ to $D_\Pi$.
Conversely, there is a homomorphism from $D_\Pi$ to $D'_\Pi$.
Indeed, by construction of $D_G$, for every fact $R(d_1, \dots, d_p) \in D_\Pi$ the multiplicity of facts of form $R((d_1, a_1), \dots, (d_p, a_p)$
in $D_G$ is exactly $|G|$. Thus, we cannot remove all copies of the fact $R(d_1, \dots, d_p)$ with $m < |G|$ removals. Hence, there is a derivation $\Gamma'$ witnessing $D'_\Pi \models \Pi$.

It is clear that using the homomorphism witnessing $G' \models q_{\triangle}$ and the derivation $\Gamma'$ we can define a derivation witnessing $H \models \Pi$, which proves ${\bf (ii)}$.
Finally, for {\bf (i)} we show the following.
\\[1mm]
\textbf{Claim.} 
   $G \models q_{\triangle}$ if and only if $f(G) \models \Pi$. 
\\[1mm]
Proof of the claim. The only if direction is simple.
Since $D_{\Pi} \models \Pi$, 
there is a derivation $\Gamma_{\Pi}$ witnessing it.
Moreover, there is a triangle $v_1,v_2, \dots, v_k$ in $G$ with $v_4 = v_5 = \dots v_k = v_1$. We modify the derivation in the following way.
If a node $u$ of the derivation uses an atom $R(a_1, \dots, a_p)$ then we replace $a_i$ with $(a_i,\bot)$ for $a_i \notin S_C$.
Otherwise, for $c_i,c_j$ that are used in the atom, we replace them with $(c_i,v_i)$ and $(c_j, v_j)$, respectively. It is easy to check that the obtained tree is a derivation 
in $D$.

For the other direction we use the fact that $\{C\}$ is crucial. Let us assume that there is a~derivation of~$\Pi$ 
in~$D$.
Then, there is a derivation $\Gamma_{\Pi}$ of $\Pi$ 
in $D_{\Pi}$ obtained from $\Gamma$ by projecting out the second coordinate of constants $(d,v)$ in nodes.
Since $\{C\}$ is crucial, there is a node $u_\Pi$ in the derivation $\Gamma_{\Pi}$, a rule $P(x_1) \leftarrow q(x_1, \dots, x_p) $ used in this node, a cycle $ Y = \{y_1, \dots, y_k, y_{k+1}=y_1\}$ in $q$, 
and a homomorphism $h_{\Pi}$ from $q$ to the bag containing cycle $C$ such that $q$ restricted to $Y$ is a bijection between cycles $Y$ and $C$. Without loss of generality, by permuting cycle $Y$ we can assume that $h_{\Pi}(y_i) = c_i$. Since $C$ is chordless and $h_{\Pi}$ is a homomorphism, then $Y$ is chordless in $q$.

Let $A_1, \dots A_k$ be atoms in $q$ such that $A_i$ uses exactly $y_i$ and $y_{i+1}$. Such atoms have to exist otherwise $Y$ would not be a chordless cycle. Moreover, the atoms' images under $h_{\Pi}$ have to be all distinct. In particular, it has to hold that the fact $h_{\Pi}(A_i)$ uses exactly two constants from $S_C$, namely $c_i$ and $c_{i+1}$. Otherwise, $C$ would not be chordless.

Consider the node $u$ in derivation $\Gamma$ that produced node $u_{\Pi}$ in $\Gamma_{\Pi}$ by the projection.
By definition of $\Gamma$, the homomorphism $h$ in this node maps variables $y_i$ to constants $(c_i, v_i)$, for some $v_i \in V$, and the atom $A_i$
is mapped to some fact $h(A_i)$. Hence, for all $1 \leq i \leq k$ there is a fact $h(A_i)$ using both $(c_i, v_i)$ and $(c_{i+1}, v_{i+1})$.
By definition of $D$, this entails that $R_1(v_1,v_2) \in G$, $R_2(v_2,v_3) \in G$, $R_3(v_3,v_4) \in G$, and $v_i = v_{i+1}$ for $4 \leq i \leq k$.
In particular, this implies that $v_4 =v_1$ and, thus, $G \models q_{\triangle}$.
This finishes the proof of the claim and also of the three conditions for Definition~\ref{def:proptestreduction}.
%
\end{proof}



\bigskip
Similarly, we show the following lemma concerning the tetra case.


\begin{restatable}{lemma}{lemdataloglowerclique}
\label{lem:datalog-lower-clique}
Let $\Pi$ be a Boolean  MDLog program such that there exist a database $D$ with  $D \models \Pi$
  and  a tetra $K$ in  $D$ 
  with   $\{ K \}$  crucial for
  $D \models \Pi$.  Then 
  falsity of\/ $\Pi$ is not constant query testable with one-sided error.
\end{restatable}
\begin{proof}
We prove the lemma by a reduction from testing freeness of a self-join free tetra shaped CQ $q_K$ which we define below.
We use $D_\Pi$ to denote $D$ from the lemma to make clear that it is fixed.

We associate with every input $G$ to $\mathcal{P}_{q_{K}}$
an input $f(G)$ to $\mathcal{P}_\Pi$. Moreover, with any $\epsilon_G \in (0,1)$ we associate $\epsilon_{f(G)} = \frac{1}{|D_{\Pi}|}\varepsilon_G$. 
We prove the three conditions of Definition~\ref{def:proptestreduction} below, and since 
$q_{K}$ will be self-join free and a tetra,
Theorem~\ref{theo:chen-2} implies that falsity of $\Pi$ is not constant query testable with one-sided error, yielding the lemma.

Intuitively, we will define $q_K$ by removing from the facts in $D_\Pi$ all constants except those in $K$. A structure comprised of such
facts is a tetra. Unfortunately, this tetra is not self-join free, uses relation symbols of varied arity and, thus, makes the work with the definition of a~crucial tetra challenging. By those reasons, we further simplify this structure.

Let us fix a linear order on constants in domain of $D_{\Pi}$ and let $K = \{c_1, \dots, c_{k+1}\}$, where $c_i$s are listed in a strictly increasing order.
We define a function $\hat{\tau}$ mapping facts $A(d_1, \dots, d_\ell)$ to facts $S_{\{b_1, \dots, b_m\}}(b_1, \dots, b_m)$
where $S_{\{b_1, \dots, b_m\}}$ is a fresh symbol, $\{b_1, \dots, b_m\} = \{ d_1, \dots, d_\ell\} \cap K$, and the sequence $b_1, \dots, b_m$
is strictly increasing. Since $K$ is a tetra, no atom maps to $S_K(c_1, \dots, c_{k+1})$, but the image of the database $D_\Pi$  can still contain facts with arity that is not $k$.

To remedy this, we define a map $\tau$ in the following way.
For $F = A(d_1, \dots, d_\ell)$ such that $\hat{\tau}(F)$ is of arity $k$, set $\tau(F) = \hat{\tau}(F)$.
For $F = A(d_1, \dots, d_\ell)$ such that arity of $\hat{\tau}(F)$ is at most $k{-}1$ set $\tau(F) = S_{\{b_1, \dots, b_{k}\}}(b_1, \dots, b_{k})$
so that the sequence $(b_1, \dots, b_k)$ is the lexicographically smallest, strictly increasing sequence extending constants in $\hat{\tau}(F)$, i.e. 
$\{ d_1, \dots, d_\ell\} \cap K \subseteq \{b_1, \dots, b_k\}$.
We use the lexicographic order arbitrarily, with a single purpose to make the definition of $\tau$ unambiguous.

Now, let $D_K$ be the database consisting of the images of facts in $D_\Pi$ under $\tau$, i.e. $D_K = \{ \tau(A) \mid A \in D_\Pi \}$.
We define $q_K$ as Boolean CQ whose canonical database is $D_K$.

Since $K$ is a tetra, $q_K$ is a Boolean conjunctive query consisting only of atoms of form $S_{\{b_1, \dots, b_k\}}(b_1, \dots, b_k)$
such that $\{b_1, \dots, b_k\} = K \setminus \{c_i\}$ for some $c_i \in K$, and the sequence $b_1, \dots, b_k$ is strictly increasing.
In particular, $K$ is a tetra in $q_K$.


Since $q_{K}$ is self-join free and $K$ is a tetra in $q_{K}$, $q_K$ satisfies the preconditions of Theorem~\ref{theo:chen-2}.
Thus, to prove the lemma, it is sufficient to show a reduction from testing $q_{K}$ to testing program $\Pi$.

To show the reduction, we define a function $f$
mapping a database $G$ over the signature consisting of symbols $S_{U}$ where $U = K \setminus \{c\}$ for $c \in K$  
to a~database representing database $D_G = f(G)$ so that the following holds:

\begin{itemize}\itemsep=0pt
    \item[(i)] if $G \not \models q_{K}$ then $f(G) \not \models \Pi$;
    \item[(ii)] if $G$ is $\epsilon$-far from being $q_{K}$-free, then $f(G)$ is $\varepsilon_{f(G)}$-far from being $\Pi$-free;
    \item[(iii)] the answer to a query to a completion or size oracle for $f(G)$ can be computed from the answers to constantly many queries to completion or size oracles for $G$.
\end{itemize}

Intuitively, we will use facts using exactly $k$ constants from the set $K$ to encode $S_U$-facts from $G$ and the remaining facts to facilitate a derivation of the program.

This is achieved by the following construction. Let $G$ be a database over signature as~stated above.
Without loss of generality, we can assume that every relation $S_{\{b_1, \dots, b_k\}}$ is not empty in $G$, otherwise there is no homomorphism from $q_K$ to $G$ and we can take trivial $f(G)$. Moreover, we will assume that every such relation contains an isolated fact $S_{\{b_1, \dots, b_k\}}(a_1^\dagger, \dots, a_{k-1}^{\dagger})$, i.e. no other facts in $G$ use $a_1^\dagger$, \dots, $a_{k-1}^\dagger$. It is easy to see that since there is a constant, independent of $G$, number of such facts and none of those facts can be used in forming a tetra, their existence does not interfere with property testing.

Now, for every fact $F = A(d_1, \dots, d_p) \in D_\Pi$ with
$\tau(F) = S_{\{b_1, \dots, b_k\}}(b_1, \dots, b_k)$, and for every fact
$S_{\{b_1, \dots, b_k\}}(a_1, \dots, a_k) \in G$, add to $D_G$ the fact
\[
  A((d_1, a_{g(1)}), \dots, (d_p, a_{g(p)}))
\]
where $a_0 = \bot$ and $g \colon [1, \dots, p] \to [0, 1, \dots, k]$ is the
function such that $g(i) = 0$ if $d_i \notin K$ and $g(i) = j$ if $d_i = b_j$
for some $j$;

Then, for every fact $R(d_1, \dots, d_p) \in D_\Pi$ copy one of the added facts
of form $R((d_1, \cdot), \dots, (d_p, \cdot)$ enough times so the multiplicity of facts of form $R((d_1, \cdot), \dots, (d_p, \cdot)$ in $D_G$
is exactly $|G|$.
If possible, we always choose a fact corresponding to constants $a_i^\dagger$. If not, then by construction, there is no choice. 
For that case, the multiplied fact is $R((d_1, \bot), \dots, (d_p, \bot))$.
As~a~consequence, $|D_G| = |D_\Pi||G|$.

Note that $f(G)$ can be computed in time linear in $|G|$: for every $S_{U}$-fact in $G$
we scan $D_\Pi$ once, and we add no more than $|D_\Pi|$ elements per scan. The final adjustment of multiplicities can be done by
keeping for every fact $F$ in $D_\Pi$ a counter and a candidate for the copy procedure, e.g. the first fact $F'$ computed using $F$.

To show {\bf (iii)} observe the following.
For a size query $R(a_1, \dots, a_p)$ identify all occurences $o_j$ of $R(y_1, \dots, y_p)$ in the database $D_{\Pi}$.
Any completion of $R(a_1, \dots, a_p)$ was produced by one of those occurences, according to one of the rules specified in the construction or by adjusting multiplicities.
Then, for every such occurence determine the number $m_j$ of completions in the database. For the occurences $R(y_1, \dots, y_p)$ that does not use variables $y_i$ from the set $K$, $m_j$ can be computed during the construction. For the occurences $R(y_1, \dots, y_p)$ that use variables $y_i$ from the set $K$, this can be done by identifying the unique rule in construction in which this occurence was used and then using oracle for $G$ to compute the number of tuples in $G$ used in this point of the construction.
One oracle call per instance of $R(y_1, \dots, y_p)$.
Then, idenify the number of completions that were added to correct multiplicities. For every occurence $o_j =  R(y_1, \dots, y_p)$, we have added  $0 \leq m'_j < |G|$. To compute $m'_j$ first compute the size of $G$, using size queries $S(\ast, \dots, \ast)$ for every relation $S$ in $G$,
and then using a constant number of size queries, compute the number $r_j$ of facts were added to the database using construction rules for $R(y_1, \dots, y_p)$. The number $m'_j = |G| - r_j$, if the fact added during multiplicity matches the query, or $m'_j= 0$, otherwise. 
Finally, return $m = \sum_{j} m_j + m'_j$.
Since the number of occurences is bounded, and $|G|$ can be computed using boundedly many queries, this whole procedure also only uses a bounded number of queries.

For a~completion query $R(a_1, \dots, a_p)$ identify all occurences $o_j$ of $R(y_1, \dots, y_p)$ in the database~$D_{\Pi}$.
Any completion of $R(a_1, \dots, a_p)$ was produced by one of those occurences, according to one of the rules specified in the construction or by adjusting multiplicities.
Then, identify the values $m_j$ of size queries $R(a_1, \dots, a_p)$ for the appropriate occurence $o_j$ and the numbers $m'_j$ of the facts added during adjusting multiplicities.
Let $m = \sum_{j} m'_j + m_j$. If $m = 0$ then fail.

Otherwise, based on the counts, choose one of the occurences $o_j$ with respect to the probability $\frac{m_j}{m}$ or the fact $f_j$ used to correct multiplicity of this occurence with probability $\frac{m'_j}{m}$.
Now, use $o_j$ (or $f_j$) to construct a completion of $R(a_1, \dots, a_p)$ based on the construction of the database. 
For occurences using variables $y_i$ from the set $K$, this may require a single call to the oracle for $G$.



For {\bf (ii)} observe that it is enough to take $\varepsilon_{f(G)} =  \frac{1}{|D_{\Pi}|} \varepsilon_G$. 
Let $m \leq  \varepsilon_{f(G)}|f(G)|$ and $H$ be $f(G)$ after removing $m$ facts from $f(G)$.
Via construction of $D_G$, removal of those facts corresponds to removal of no~more than $m$ facts from $G$.
Since \[m \leq \varepsilon_{f(G)} |f(G)| \leq \frac{\varepsilon_G}{|D_\Pi|} \cdot |D_\Pi||G| = \varepsilon_G |G|,\]
the structure $G'$ obtained from $G$ by removing those edges satisfies $G' \models q_{K}$.

This removal also corresponds to a database $D'_\Pi$ obtained by projecting the constants in the domain to the first coordinate.
Clearly, there is a homomorphism from $D'_\Pi$ to $D_\Pi$.
Moreover, since, by construction of $D_G$, for every fact $R(d_1, \dots, d_p) \in D_\Pi$ the multiplicity of facts of form $R((d_1, a_1), \dots, (d_p, a_p)$
in $D_G$ is exactly $|G|$, there is a homomorphism from $D_\Pi$ to $D'_\Pi$. Hence, there is a derivation $\Gamma'$ witnessing $D'_\Pi \models \Pi$.

It is clear that using the homomorphism witnessing $G' \models q_{K}$ and the derivation $\Gamma'$ we can define a derivation witnessing $H \models \Pi$, which proves ${\bf (ii)}$.

For point {\bf (i)} we show that 
\\[1mm]
\textbf{Claim.} 
   $G \models q_K$ if and only if $f(G) \models \Pi$.
\\[1mm]
Proof of the claim.
The only if direction is simple.
Since $D_{\Pi} \models \Pi$, 
there is a derivation $\Gamma_{\Pi}$ witnessing it.
Moreover, there is homomorphism $h_G$ from $q_K$ to $G$. We modify the derivation in the following way.
If in node~$u$ of the derivation we use the fact $R(d_1, \dots, d_p)$ with the corresponding fact $S_{\{ b_1, \dots, b_\ell\}}(b_1, \dots, b_\ell)$ then we replace $d_i$ with $(d_i,\bot)$ for $d_i \notin K$ and with $(d_i, h_G(d_i))$ for $d_i \in K$.
It is easy to check that the obtained tree is a derivation 
of $\Pi$ in $f(G)$.

For the other direction we use the fact that $\{K\}$ is crucial to show that if there is a derivation of~$\Pi$ in $f(G)$ then there is a homomorphism $h_K$ from  $q_K$ to $G$. Let us assume that there is a~derivation of $\Pi$ 
in~$f(G)$.
Then, there is a derivation $\Gamma_{\Pi}$ 
of $\Pi$ in $D_{\Pi}$ obtained from $\Gamma$ by projecting out the second coordinate of constants $(d,a)$ in nodes.
Since $\{K\}$ is crucial, there is a node $u_{\Pi}$ in the derivation $\Gamma_{\Pi}$, a rule $P(x_1) \leftarrow q_u(x_1, \dots, x_p) $ used in this node, a tetra $ Y = \{y_1, \dots, y_k, y_{k+1}\}$ in $q_u$, 
and a~homomorphism $h_{\Pi}$ from $q_u$ to 
$D_\Pi$ such that $h_\Pi$ restricted to $Y$ is a bijection between sets $Y$ and $K$. Without loss of generality, we can assume that $h_{\Pi}(y_i) = c_i$, otherwise permute $Y$.

Let $h_u$ be the homomorphism in derivation $\Gamma$ that produced $h_\Pi$ in $\Gamma_\Pi$. We define the homomorphism $h_K \colon q_K \to G$
as follows: $h_K(c_i) = a$ such that $h_{u}(y_i) = (c_i, a)$. By definition of the database $f(G)$, $a \in \dom(G)$. Hence, $h_K$ is a function from variables of $q_K$ to the domain of $G$.
In particular, $a \neq \bot$.
Now, if we show that $h_K$ preserves relations, we will  prove that it is a homomorphism.


Let $S_{\{b_1, \dots, b_k\}}(b_1, \dots, b_k)$ be an atom in $q_K$ and recall that $k = |K|-1$.
We will show that $S_{\{b_1, \dots, b_k\}}(h_K(b_1), \dots, h_K(b_k)) \in G$, which shows that $h_K$ is a homomorphism.

Since $Y$ is a tetra, there is an atom $A(z_1, \dots, z_\ell)$ in $q_u$
such that $\{z_1, \dots, z_\ell\} \cap Y = \{h_{\Pi}^{-1}(b_1), \dots, h_{\Pi}^{-1}(b_k)\}$. That is, there is an atom $A(z_1, \dots, z_\ell) \in q_u$ mapped via $h_\Pi$ to a fact $F =  A(h_\Pi(z_1), \dots, h_\Pi(z_\ell)) \in D_\Pi$ such that  $\hat{\tau}(F) = S_{\{b_1, \dots, b_k\}}(b_1, \dots, b_k)$. This atom is mapped to some fact $A((d_1, a_1), \dots, (d_\ell, a_\ell)) \in f(G)$ via $h_u$.

Since $A((d_1, a_1), \dots, (d_\ell, a_\ell)) \in f(G)$, by definition of $f(G)$ this fact was added to $f(G)$ for some fact $S_{\{b_1, \dots, b_k\}}(a'_1, \dots, a'_k) \in G$ such that for $1 \leq i \leq k$ there is $0 \leq j \leq k+1$ such that
\[(b_i, a'_i) = h_u(y_j) = (c_j, a) = (c_j, h_K(c_j)) = (b_i, h_K(b_i))\]
The first equality stems from the fact that $b_i = c_j \in K$ for some $1 \leq j \leq k+1$, second is the definition of $h_u$, third
is the definition of $h_K$, and the last is once again $c_j=b_i$.

Hence, $a'_i = h_K(b_i)$ for $1\leq i \leq k$ and $S_{\{b_1, \dots, b_k\}}(h_K(b_1), \dots, h_K(b_k)) \in G$.

This proves that $h_K$ is a homomorphism, ends the proof of property {\bf (i)} of the reduction and, thus, ends the proof of the lemma.
\end{proof}

\end{document}